\documentclass[12pt,reqno,a4paper]{amsart}
\usepackage{amsmath,amssymb,dsfont,cite,latexsym,
epsf,cancel,epic,eepic,mathrsfs}
\usepackage{graphicx}
\usepackage[urlcolor=cyan,colorlinks=false]{hyperref}

\usepackage{tikz}
\usepackage{caption}
\usepackage{subcaption}
\usepackage{pgfplots}

\usepackage{mathtools}
\usepackage{algorithm}
\usepackage{algpseudocode}

\usepackage{physics}

\usepackage{multirow}

\graphicspath{ {Figures} }

\usetikzlibrary{arrows.meta}
\usetikzlibrary{arrows}
\newtheorem{theorem}{Theorem}
\newtheorem{proposition}{Proposition}
\newtheorem{lemma}[proposition]{Lemma}

\newtheorem{rem}[proposition]{Remark}

\newtheorem{definition}[proposition]{Definition}
\newtheorem{corollary}[proposition]{Corollary}
\makeatletter
\expandafter\let\expandafter
\reset@font\csname reset@font\endcsname
\def\subeqnarray{\arraycolsep1pt
   \def\@eqnnum\stepcounter##1{\stepcounter{subequation}
       {\reset@font\rm(\theequation\alph{subequation})}}
\jot5mm     \eqnarray}

\makeatother

\newcommand{\bbR}{{\mathbb R}}

\newcommand{\calR}{{\mathcal{R}}}

\def\epsilon{\varepsilon}

\def\tilde{\widetilde}

\newcommand{\br}[1]{\left( #1 \right)}

\newcommand{\cubr}[1]{\left\{ #1 \right\} }
\newcommand{\pobr}[1]{\left< #1 \right> }

\newcommand{\restr}[2]{\ensuremath{\left.#1\right|_{#2}}}
\newcommand{\innerprod}[2]{\langle #1 , #2 \rangle}

\DeclareMathOperator*{\argmin}{arg\,min}
\DeclareMathOperator*{\re}{Re}
\DeclareMathOperator*{\im}{Im}
\DeclareMathOperator{\Lip}{Lip}
\DeclareMathOperator{\Div}{div}

\DeclareMathOperator*{\res}{res}

\begin{document}
\title{Dimers and M-Curves: Limit Shapes From Riemann Surfaces}

%
\author{Alexander I. Bobenko and Nikolai Bobenko}

\thanks{Affiliations: { \tt AB: Institute of Mathematics, Technische Universität Berlin, Germany \\
NB:  Department of Mathematics, University of Geneva, Switzerland}}
\thanks{E-mail: {\tt bobenko@math.tu-berlin.de, nikolai.bobenko@unige.ch}}

\begin{abstract}
    We present a general approach for the study of dimer model limit shape problems via variational and integrable systems techniques. In particular we deduce the limit shape of the Aztec diamond and the hexagon for quasi-periodic weights through purely variational techniques.

    Putting an M-curve at the center of the construction allows one to define weights and algebro-geometric structures describing the behavior of the corresponding dimer model.
    We extend the quasi-periodic setup of \cite{BBS} to include a diffeomorphism from the spectral data to the liquid region of the dimer. 

    Our novel method of proof is purely variational and exploits a duality between the dimer height function and its dual magnetic tension minimizer and applies to dimers with gas regions. We apply this to the Aztec diamond and hexagon domains to obtain explicit expressions for the complex structure of the liquid region of the dimer as well as the height function and its dual.

 We compute the weights and the limit shapes numerically using the Schottky uniformization technique. Simulations and predicted results match completely.
\end{abstract}

\maketitle

\makeatletter
\def\@tocline#1#2#3#4#5#6#7{\relax
  \ifnum #1>\c@tocdepth 
  \else
    \par \addpenalty\@secpenalty\addvspace{#2}%
    \begingroup \hyphenpenalty\@M
    \@ifempty{#4}{%
      \@tempdima\csname r@tocindent\number#1\endcsname\relax
    }{%
      \@tempdima#4\relax
    }%
    \parindent\z@ \leftskip#3\relax \advance\leftskip\@tempdima\relax
    \rightskip\@pnumwidth plus4em \parfillskip-\@pnumwidth
    #5\leavevmode\hskip-\@tempdima
      \ifcase #1
       \or\or \hskip 1em \or \hskip 2em \else \hskip 3em \fi%
      #6\nobreak\relax
    \hfill\hbox to\@pnumwidth{\@tocpagenum{#7}}\par
    \nobreak
    \endgroup
  \fi}
\makeatother
\tableofcontents

\section{Overview and Main Results}
\label{sec:01_Introduction}

\begin{figure}
    \centering
    \begin{subfigure}[p]{.22\textwidth}
        \centering
        \includegraphics[width=\linewidth]{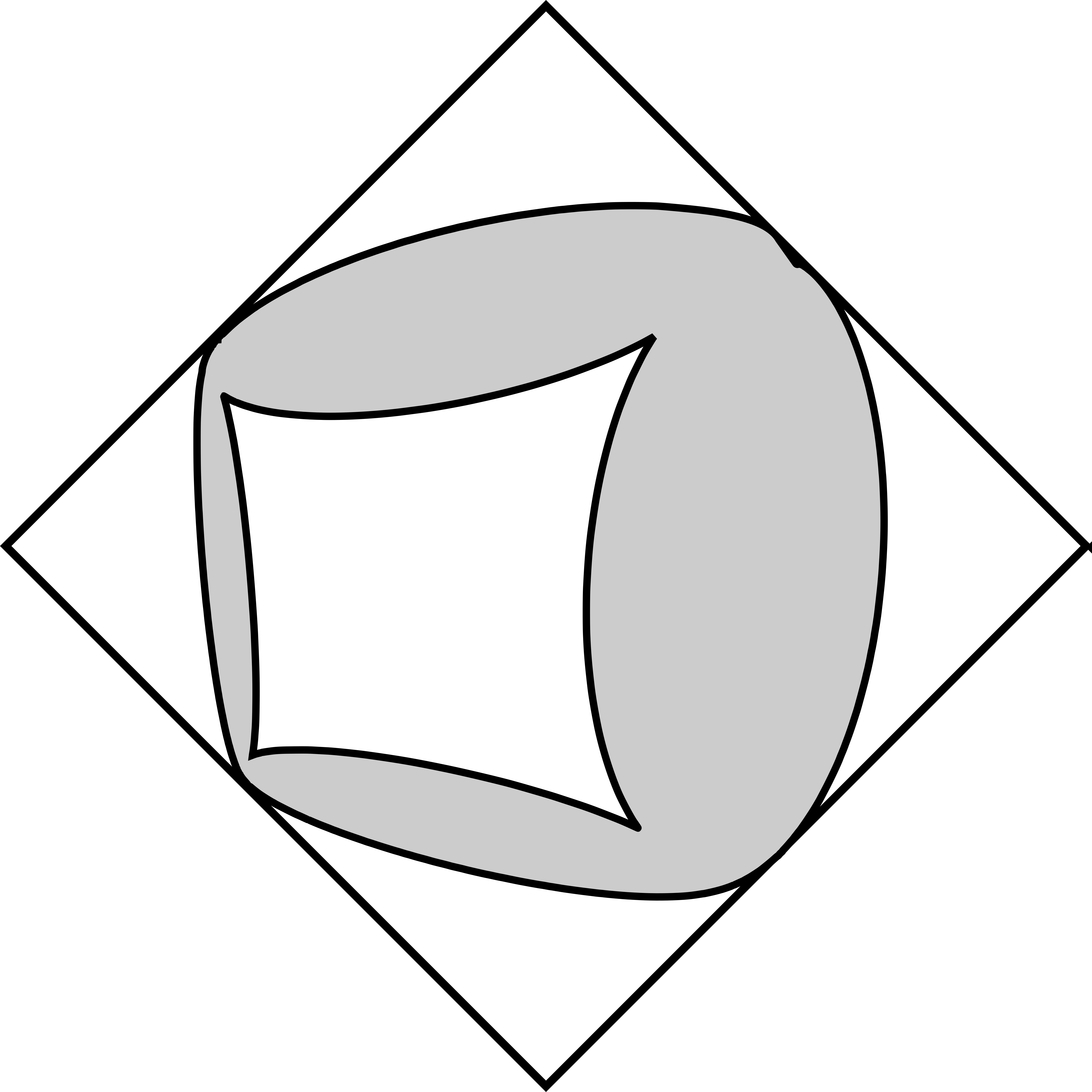}
    \end{subfigure}\hspace{0cm}%
    \begin{subfigure}[p]{.16\textwidth}
        \centering
        \includegraphics[width=\linewidth]{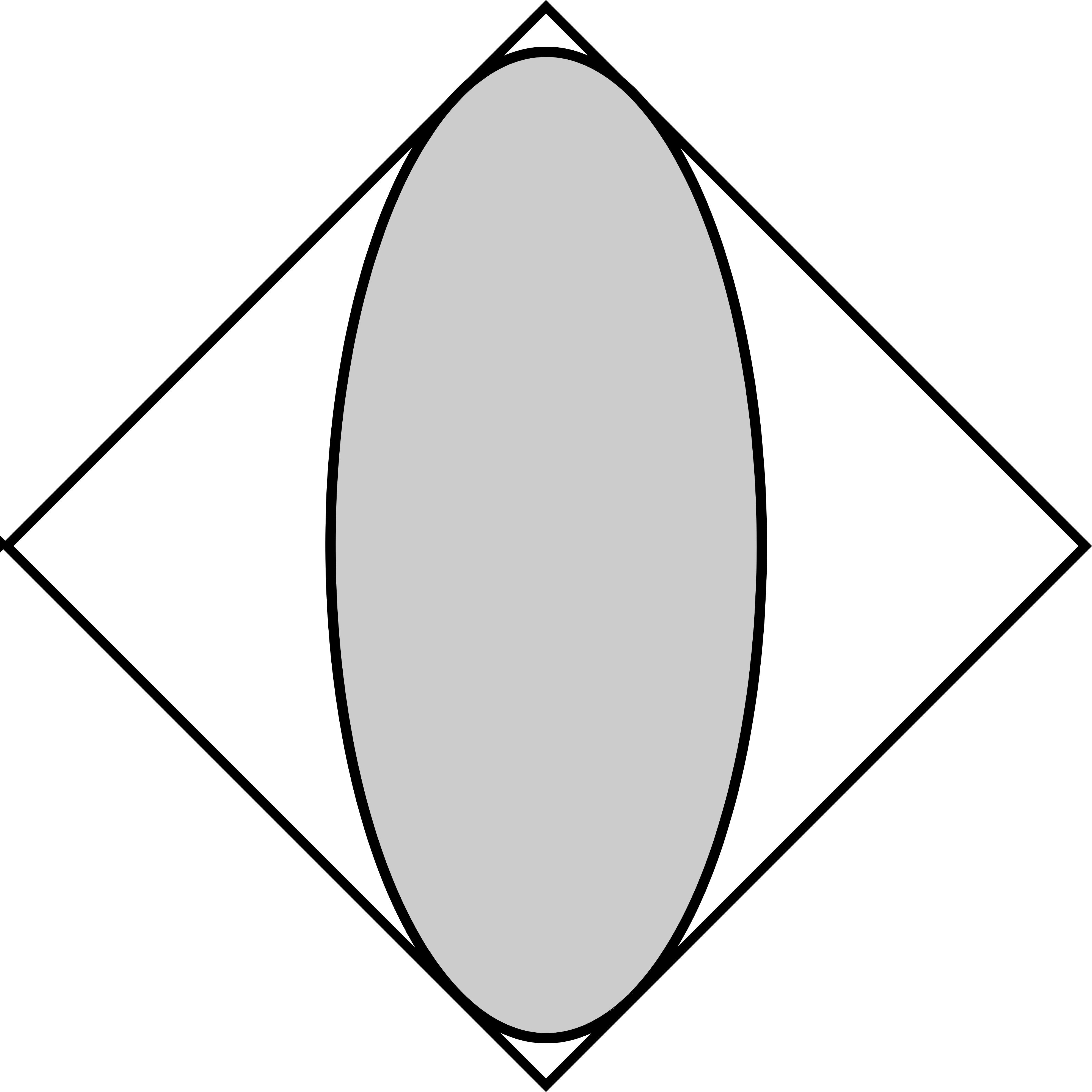}
    \end{subfigure}\hspace{0cm}%
    \begin{subfigure}[p]{.14\textwidth}
        \centering
        \includegraphics[width=\linewidth]{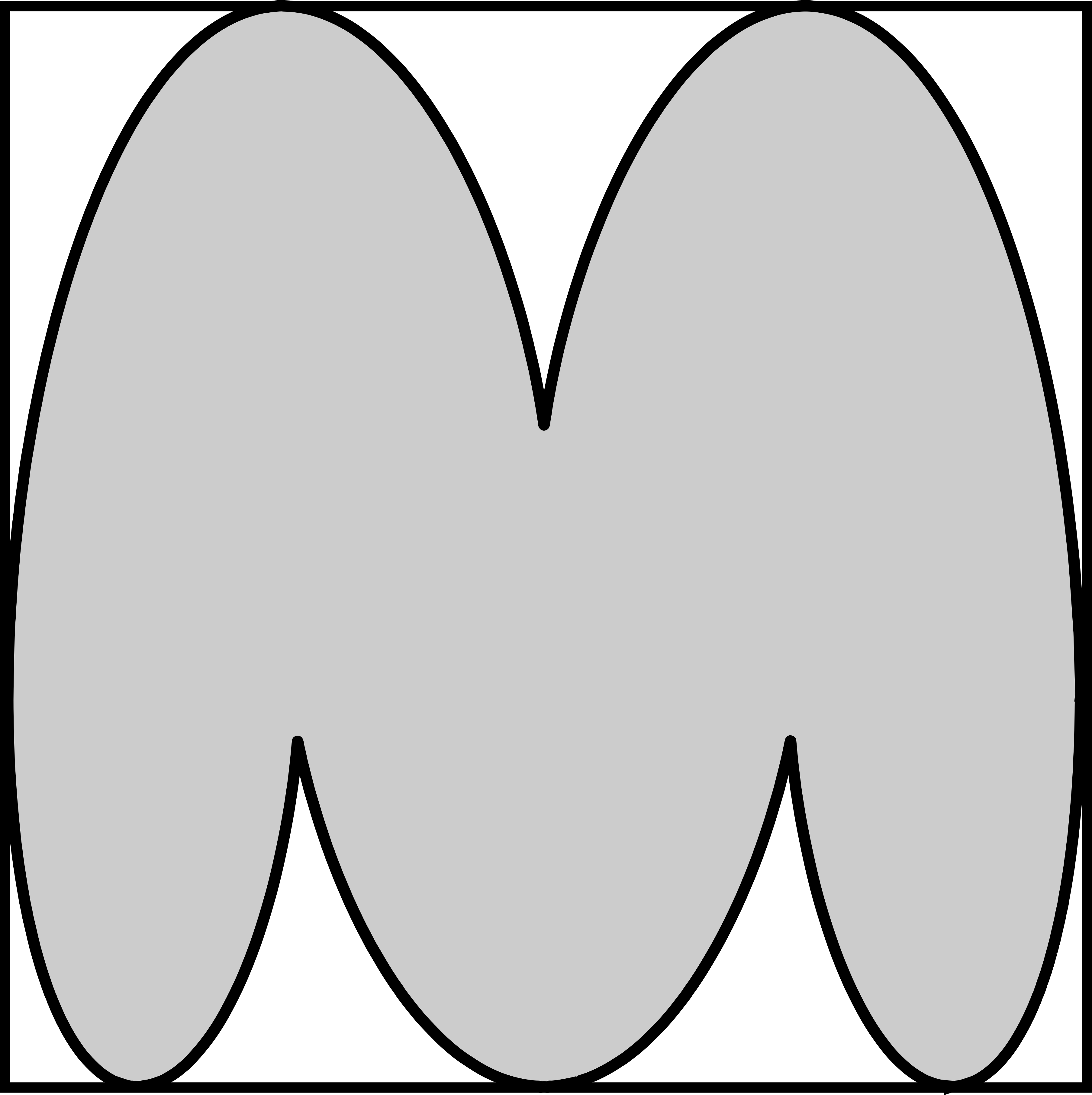}
    \end{subfigure}\hspace{0.3cm}%
    \begin{subfigure}[p]{.14\textwidth}
        \centering
        \includegraphics[width=\linewidth]{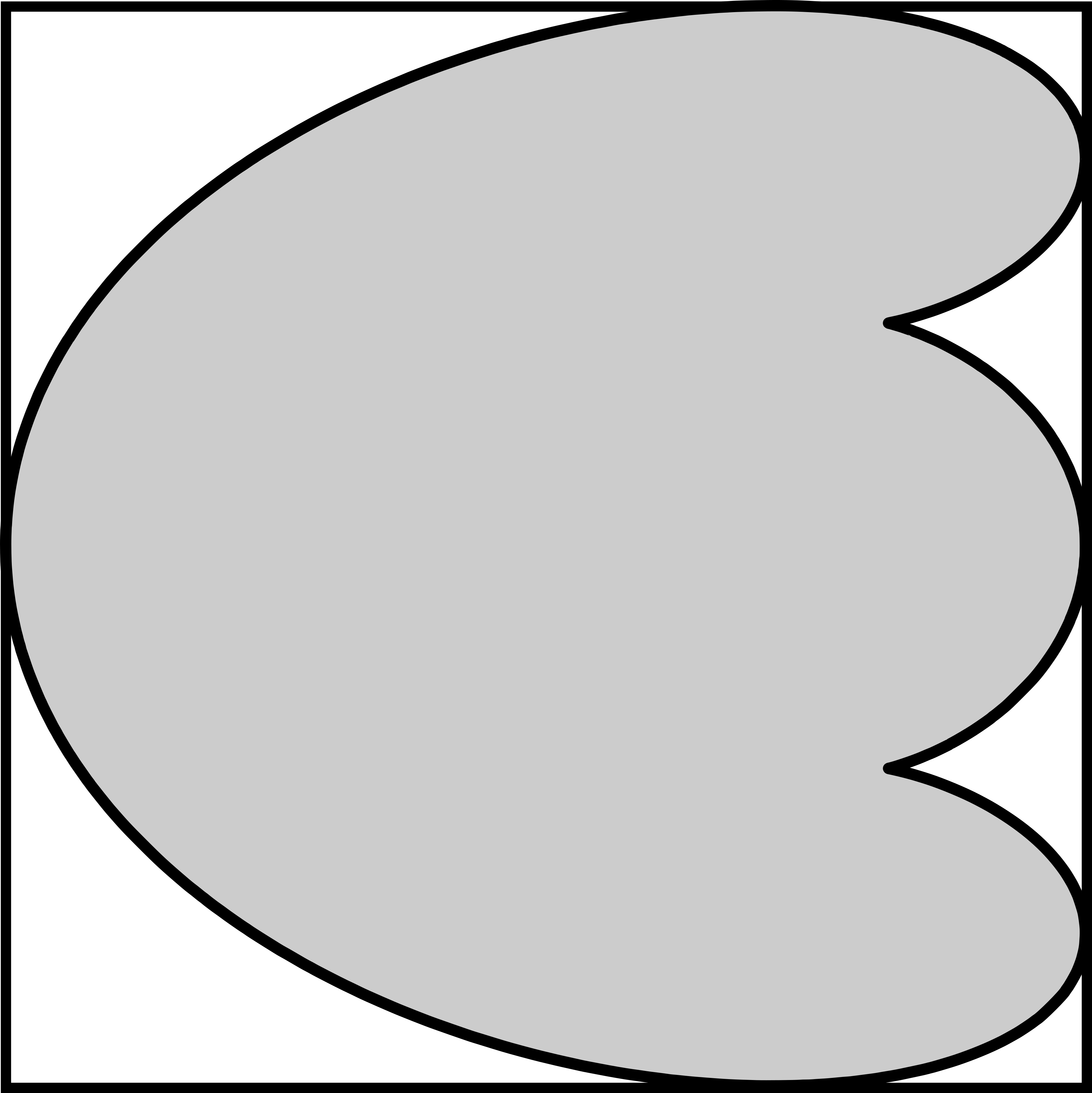}
    \end{subfigure}\hspace{0.3cm}%
    \begin{subfigure}[p]{.14\textwidth}
        \centering
        \includegraphics[width=\linewidth]{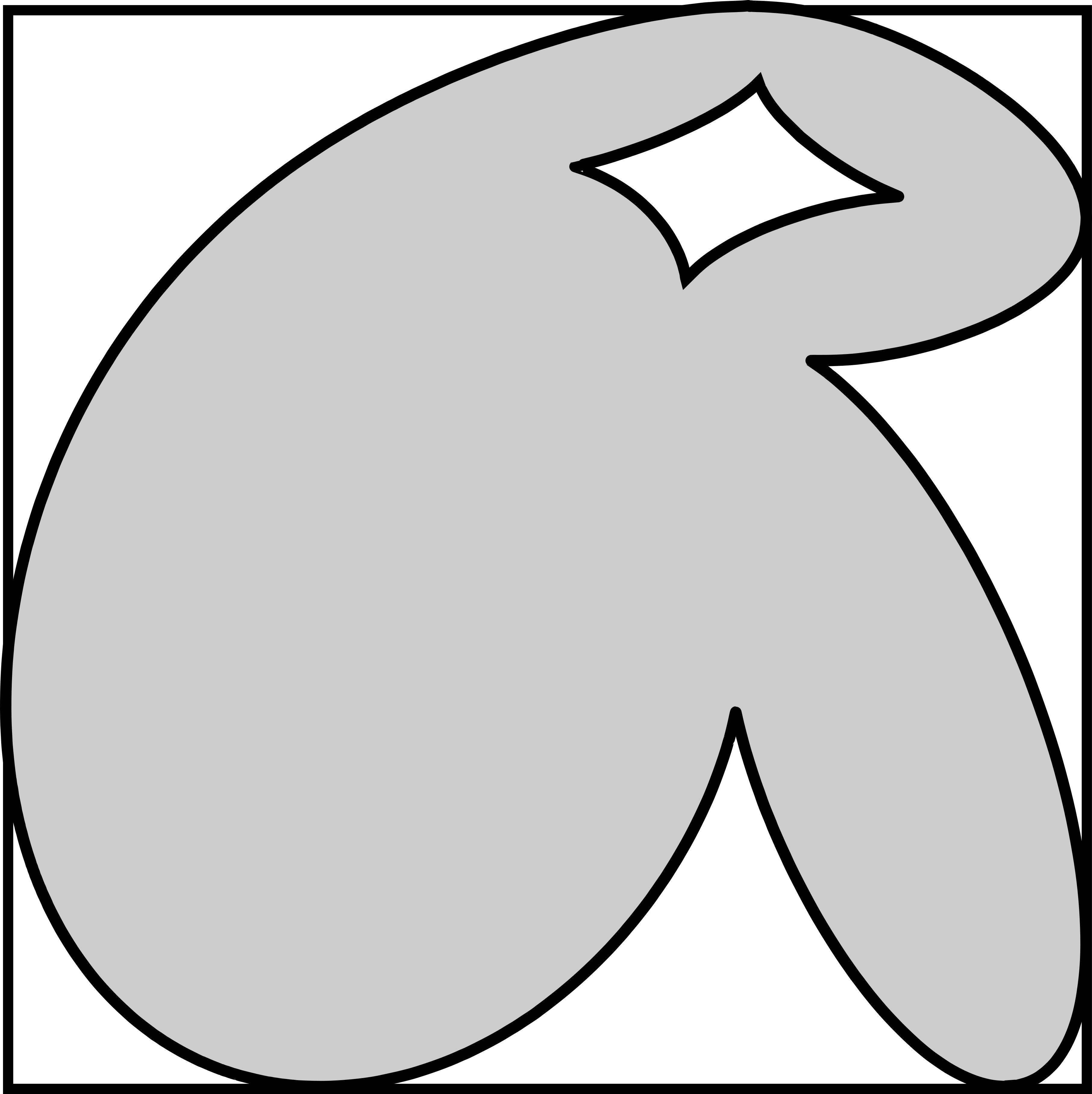}
    \end{subfigure}\hspace{0.3cm}%
    \begin{subfigure}[p]{.14\textwidth}
        \centering
        \includegraphics[width=\linewidth]{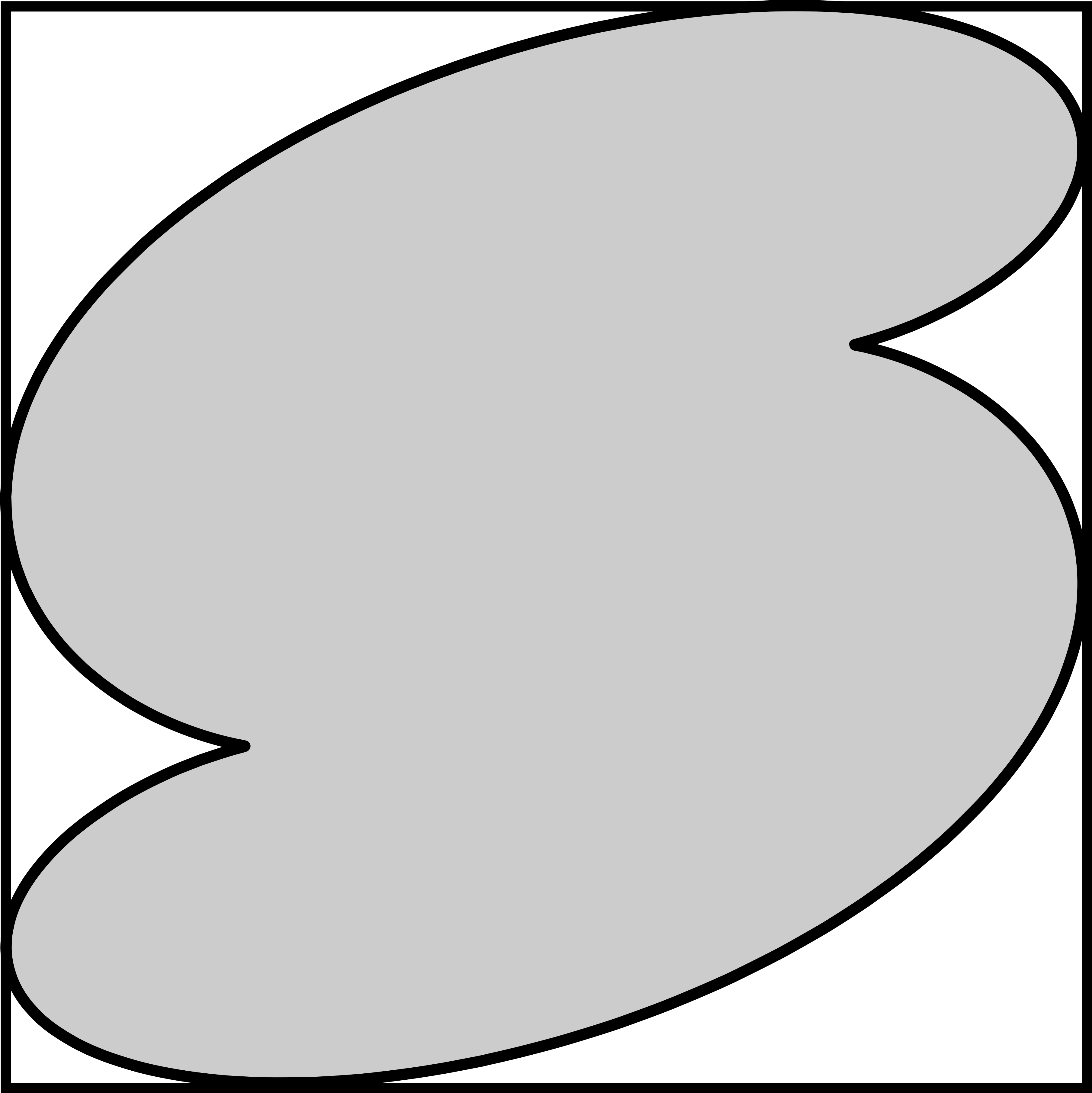}
    \end{subfigure}
    \caption{The Arctic Font: Examples of limit shapes on the Aztec diamond corresponding to different quasi-periodic weights spelling out the word Dimers. See Figure~\ref{fig:arctic_spectrum} for the spectral encoding.}
    \label{fig:Arctic_font}
\end{figure} 

The study of dimer models has established itself as an important part of statistical mechanics research over the past few decades. Connections to algebraic geometry, integrable systems and other statistical mechanics models make it a deep and diverse field \cite{goncharov_dimers_2013,Kenyon_Okounkov_Sheffield_2006, kenyon_conformal_invariance_2000, chelkak_perfect_t_embeddings_2022}. In particular the study of limit shapes exhibited by dimer models with certain boundary conditions has been of notable interest, see e.g. \cite{Kenyon_Okounkov_2007, DeSilva_minimizers_of_convex_functionals_2010, astala_dimer_2023, berggren_geometry_2023}.

In this paper we consider dimers on the classical examples of the Aztec diamond and hexagonal domains. The dimer model is a probability distribution on the set of all perfect matchings of a graph. For planar graphs on simply connected domains this is equivalent to considering a height function on the dual graph.

A striking phenomenon exhibited by dimers is that of separation of phases and limit shapes, see e.g. Figures~\ref{fig:main_aztec_config} and \ref{fig:main_hex_config}. These are figures of samples of the dimer model on large parts of the lattice with some quasi-periodic weights. We observe three types of local behavior. Along the boundary we see \emph{frozen zones} with completely deterministic behavior. There is a large \emph{liquid region} with continuously changing slope and \emph{gas regions} that form flat pieces and only exhibit small fluctuations. The curves separating the regions are known as \emph{arctic curves}.

This was first observed in the context of the Aztec diamond with uniform weights where a circle can be seen separating a liquid zone from $4$ frozen ones. A first rigorous result towards a description was given in \cite{jockusch1998randomdominotilingsarctic} establishing that in the scaling limit the curve separating the regions is indeed a circle. This was extended in \cite{Cohn_1996} to include some local statistics and some first results towards general limit shape principles that dimers exhibit. These were futher generalized in \cite{cohn_variational_2000}, where it was established that in the scaling limit for general domains filled with the square grid with uniform weights, the dimer height function concentrates around the minimizer of a surface tension functional with corresponding boundary conditions.

In \cite{Kenyon_Okounkov_Sheffield_2006, Kenyon_Okounkov_2006} a remarkable connection between dimers on general doubly periodic graphs with doubly periodic weights and algebraic geometry was established. From the Kasteleyn operator on a torus a spectral curve was defined which was shown to always be a Harnack curve, a well-studied type of real algebraic curve \cite{Mikhalkin_2000}. Many of the limiting properties of the dimer model can then be described in terms of the geometric data of the curve. In particular it was shown that there is a one-to-one correspondence between points on the Newton polygon (outside the boundary) of the curve and ergodic Gibbs measures in the full plane doubly periodic graph. The point on the Newton polygon here gives the average height change in the two fundamental directions. Further, the free energy per fundamental domain in magnetic coordinates can be given in terms of the Ronkin function $\rho$ and the surface tension $\sigma$ is given as its Legendre dual. The established variational principle then generalizes to the fact that the dimer height function $h$ on a limiting domain $\Omega$ concentrates around the deterministic minimizer of the surface tension functional 
\begin{equation}
    \label{eq:h_surface_tension_minimization}
     h = \argmin_f \int_\Omega \sigma(\nabla f)
\end{equation} 
over a class of Lipschitz functions with prescribed boundary conditions. See also \cite{kuchumov_limit_2017}.

\begin{figure}
\centering
    \includegraphics[width=.8\linewidth]{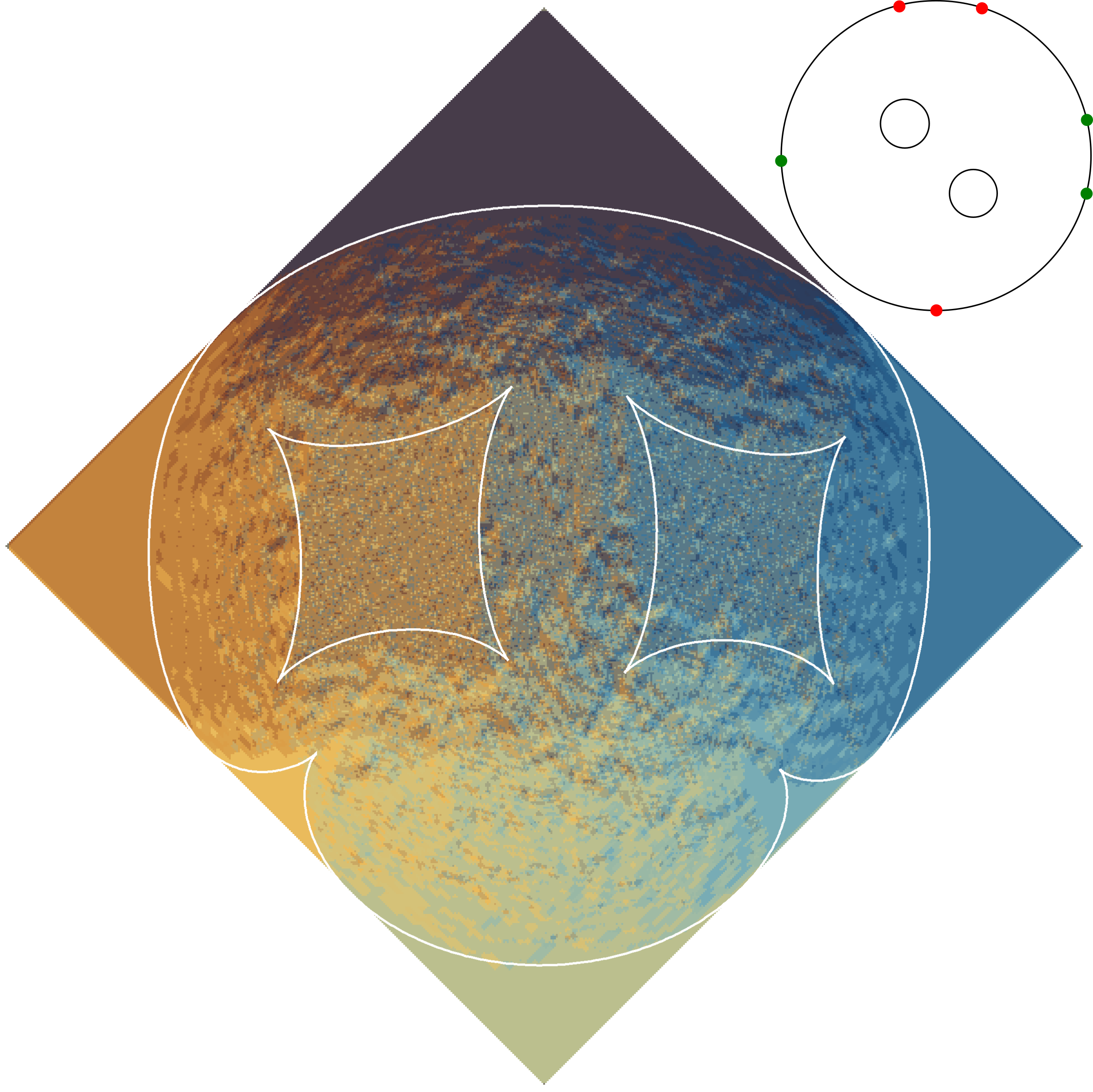}
    \caption{Dimer configuration sampled on an Aztec diamond populated with Fock weights coming from specified Harnack data $\mathcal{S} = \cubr{\calR, \mathcal{T}}$ depicted above. We overlay the Aztec diamond with predicted arctic curves calculated directly from $\mathcal{S}$. The picture of $\mathcal{S}$ represents the complex structure of the spectral curve $\calR$, that is the $z$ coordinate in the plane on this picture gives the complex structure of $\calR$. All calculations can be done numerically and we observe that predictions match the sampled result with high precision on practical scales. For all other relevant maps in this example see Figure~\ref{fig:all_maps_and_zeros}. Instead of depicting the dimer configuration we choose here to encode the local height function gradients with different colors for more visual clarity.}
    \label{fig:main_aztec_config}
\end{figure}

In \cite{Kenyon_Okounkov_2007} this was utilized to derive limit shapes for a set of polygonal boundary conditions natural to the hexagonal grid. A map from (a covering of the) spectral curve to the liquid domain of the dimer was established, defining a natural complex structure on the dimer domain. We refer to this map as the Kenyon-Okounkov map or KO map. In particular the arctic curve for uniform weights was derived to be the algebraic curve of minimal degree to touch all domain edges exactly once. The methods used here were algebraic in nature.

Having established the variational principle one can study the problem of dimer limit shapes from a purely analytical perspective. Here the surface tension $\sigma$ is known to be a function solving a Monge-Ampère equation with given singularities. These singularities are what makes the minimizer exhibit different phases and also the reason why the study of this variational problem is challenging. Uniqueness of the solution and several regularity results were established in \cite{DeSilva_minimizers_of_convex_functionals_2010} and further extended in \cite{astala_dimer_2023} with a complete classification of the regularity of the minimizers for a large class of natural domains.

The problem of finding a solution to this minimization problem for general doubly periodic dimer weights was finally resolved in \cite{berggren_geometry_2023} for the domain of the Aztec diamond. The methods employed here included asymptotic analysis of orthogonal polynomials and an algebro-geometric solution to an isospectral flow on certain quadratic matrix-polynomials. This was the culmination of previous work done in \cite{Borodin_2x2_2023, berggren_aztec_2021} and resulted in an expression of the inverse Kasteleyn operator on the doubly periodic Aztec diamond along with its asymptotic analysis, see also \cite{boutillier_focks_2024} for a generalization for Fock weights. In particular, in \cite{berggren_geometry_2023} a full description of the KO map for periodic weights as a homeomorphism from the amoeba to the liquid region as well as the limit shape in terms of the spectral data associated to the dimer was obtained.

We rely on an inverse spectral approach. That is, we start with spectral data and define a corresponding dimer model from it. To that end, complex weights for difference operators on graphs were introduced in \cite{fock_inverse_2015} and later realized to be dimer weights in \cite{boutillier_minimal_2023, boutillier_elliptic_2020} for the special case of \emph{Harnack data}. This is an M-curve with marked points on one of its real ovals satisfying an ordering property. On lattices with periodic train track parameters, this leads to \emph{quasi-periodic} weights generalizing the classical setup in a natural way \cite{BBS}. In \cite{BBS} we adapted the definitions of \cite{krichever_amoebas_2014} generalizing the amoeba map, Newton polygon, Ronkin function and surface tension to this setup while providing new formulas as well as showing that the height function for quasi-periodic weights still satisfies the variational principle \eqref{eq:h_surface_tension_minimization} in a regularized sense.

Our main contributions in this paper are as follows. First, we construct the limiting height function for the Aztec diamond with general quasi-periodic weights and provide a novel and completely variational proof of the fact that it is the surface tension minimizer. To our knowledge this is the first time that a variational proof can be applied to dimers with gas zones. See Section~\ref{sec:05_SurfaceTensionMinimization}.

To that end, we adapt the definition of the KO map from \cite{berggren_geometry_2023} to general quasi-periodic weights while formulating it as a map from the M-curve instead of the amoeba defining it through a free zero of certain differentials $d\zeta$. In this formulation the map is a diffeomorphism and we obtain explicit formulas for the liquid region as well as arctic curves and the limiting complex height function $H$ (Section~\ref{sec:03_CoordinateMap}). We study geometric properties, recovering the parallelity property from \cite{berggren_geometry_2023} and obtaining a complete geometric description of the zeros of the differentials $d\zeta$ as tangential points in the arctic curves, see Section~\ref{sec:04_GeometricProperties}.

The limiting dimer height function is given as the imaginary part of $H=g+ih$. Our proof of minimality relies on a fortunate duality of the real and imaginary parts of $H$, which extends the Legendre duality between $\rho$ and $\sigma$ to the dimer domain. The imaginary part $h$ can be extended linearly to the frozen and gas regions, while the real part $g$ has more complicated boundary conditions for extension. This duality gives as an immediate result that if one can extend $g$ in an admissible way, then $h$ is the surface tension minimizer. An admissible extension is explicitly constructed. Furthermore the real part $g$ minimizes the magnetic tension functional $\int_\Omega \rho(\nabla f)$. This is the content of Section~\ref{sec:05_SurfaceTensionMinimization}.

Second, we manage to extend this approach to the domain of the hexagon filled with the hexagonal lattice. In this case, one needs to consider a double cover of the spectral curve. Once again, we flip the script and start our construction from the double cover. We can then exploit the symmetries of the hexagon by requiring symmetric data and a similar formalism as before distinguishes a set of admissible weights for which we again construct both the KO map and the height function. This set of weights is shown to be of full dimension in the moduli space of Harnack data. Since our method of proving surface tension minimality does not rely on the heavy machinery of other models and their integrability, the proof for the hexagon is then essentially the same as for the Aztec diamond. We believe that this method has the potential to be pushed further to obtain minimizers for a larger set of boundary conditions, see Section~\ref{sec:07_hexagonal_case}. 

\begin{figure}
    \centering
    \includegraphics[width=.8\linewidth]{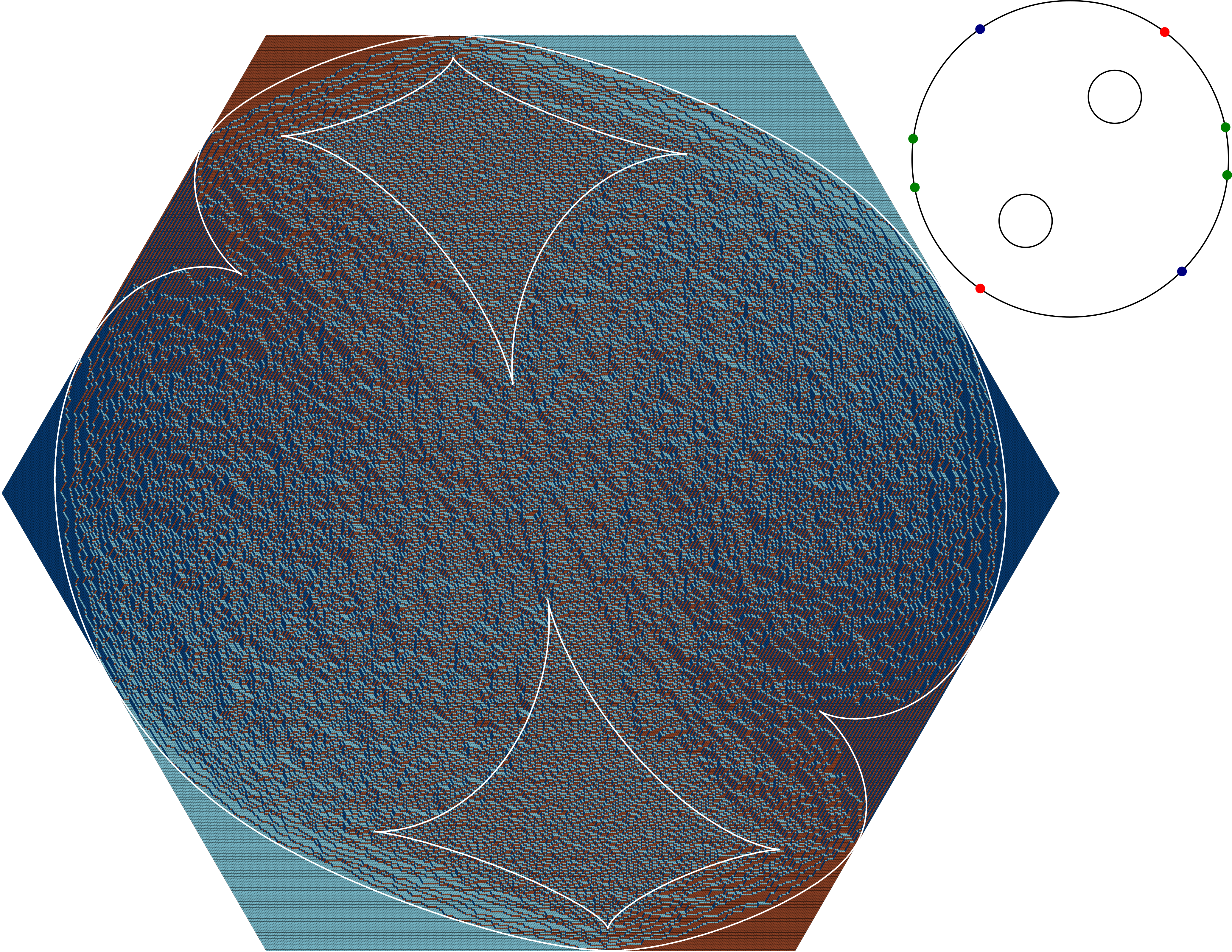}
    \caption{Dimer configuration sampled on the hexagon populated with Fock weights corresponding to the Harnack data $\hat{\mathcal{S}}$ depicted above. Note that this is a double cover of the spectral curve defining the weights allowing for a diffeomorphism to the liquid region of the dimer on the hexagon. We see that the predicted arctic curves match the observed behavior of the sampled configuration on a practical scale. For all other relevant maps in this example see Figure~\ref{fig:all_maps_and_zeros_hex}.}
\label{fig:main_hex_config}
\end{figure}

In Section~\ref{sec:06_explicit_formulas_for_aztec_diamond_G0} we consider the case of genus $0$ curves corresponding to the isoradial weights established and studied in \cite{kenyon_laplacian_2002}. In this case everything can be written in terms of rational functions. We provide explicit formulas for the surface tension and present its geometric interpretation in terms of volumes of ideal hyperbolic polyhedra. These extend Kenyon's formula for the normalized determinant of the discrete Dirac operator \cite{kenyon_laplacian_2002}. The established remarkable connection of convex variational principles for dimer models and discrete conformal maps is promising for development of both these fields.

As in \cite{BBS}, all the objects used in this paper can be computed numerically.
using the techniques of Schottky uniformizations, see Section~\ref{sec:10_SchottkyUniformization}. To demonstrate the practicality of this method, all figures in this paper, except for the diagrams illustrating proofs, are the results of computation. This includes all arctic curves, all amoebas, and all weights for all simulations. To this end we have added the Harnack data that encodes all information needed to these figures. These are the actual Riemann surfaces used along with the actual train track parameters. Furthermore we show that the parametrization used gives a natural and intuitive way to control the limit shapes, allowing us to design the shapes we want in a practical way, see e.g. Figure~\ref{fig:Arctic_font}. Note that the Harnack data in the Schottky model parametrization give a qualitatively correct intuitive picture of the corresponding dimer phases.

The suggested approach is quite general, and can be applied to further models like electrical networks, Ising model etc., which can be obtained as dimer models with symmetric data, see \cite{george_electrical_networks_2024} and the upcoming \cite{boutillier_electrical_networks_2024}.

We would like to emphasize that the Harnack data (M-curve with points on it) are in the center of our description, see Figures \ref{fig:all_maps_and_zeros},\ref{fig:all_maps_and_zeros_hex}. Traditional maps in dimer theory, going to the amoeba or the Newton polygon, are compositions of the corresponding explicit diffeomorphisms to the M-curve. This slight change of focus to the traditional description can even turn out to be useful for making progress with the proof of the Kenyon-Okounkov conjecture \cite{Kenyon_Okounkov_2007, Kenyon_honeycomb_fluctuations_2004}, considered as one of the most important and difficult problems in the theory of dimers, see \cite{Gorin_2021}. Indeed, the conjecture states that the fluctuations of the height function are closely related to the Gaussian free field in the conformal structure of the M-curves in this paper.

\vspace*{1cm}
{\bf Acknowledgements.} The first author was supported by the DFG Collaborative Research Center TRR 109 "Discretization in Geometry and Dynamics".

The second author was supported by the NCCR SwissMAP as well as by the Swiss NSF grants 200400 and 197226.

We are grateful to Kari Astala, Tomas Berggren, Alexei Borodin, Cédric Boutillier, Rick Kenyon, István Prause and Yuri Suris for helpful discussions.

\section{The Dimer Model}
\label{sec:02_TheDimerModel}
\addtocontents{toc}{\protect\setcounter{tocdepth}{1}}

For a finite graph $G = (V,E)$ we call $D \subset E$ a \emph{dimer configuration} or a perfect matching if every vertex is contained in exactly one edge in $D$. The dimer model is a probability distribution on the set $\mathcal{D}$ of all dimer configurations. Given an edge weighting $\nu: E \to \mathbb{R}_+$ one defines the weight of a dimer configuration $D$ as
\begin{equation}
    \nu(D) = \prod_{e \in D} \nu(e).
\end{equation}
The associated Boltzmann measure is then given by
\begin{equation}
    \label{eq:def_Boltzmann_measure}
    \mathbb{P}(D) = \frac{\nu(D)}{Z(G, \nu)}
\end{equation}
where $Z(G, \nu) = \sum_{D\in \mathcal{D}} \nu(D)$ is the partition function.

We will adopt the condition that $G = (V,E,F)$ is \emph{planar} and \emph{bipartite} and call the two sets of vertices $V = B \sqcup W$ the black and white vertices. An obvious necessary condition for the dimer model on $G$ to be non-trivial is $\abs{B} = \abs{W}$. Two weight functions $\nu, \nu'$ are called \emph{gauge equivalent} if they differ by factors living on the vertices. That is if there exists a function $\lambda: V \to \mathbb{R}_+$ such that for every $wb \in E$ we have $\nu'(wb) = \lambda(w)\nu(wb)\lambda(b)$. Notice that gauge equivalent weights define the same measure \eqref{eq:def_Boltzmann_measure} and thus edge weights overspecify the model.
\begin{proposition}
    \label{prop:gauge_equiv_iff_face_weights}
    Two weight functions $\nu, \nu'$ are gauge equivalent iff for each $f \in F$ with $\partial f = e_1 \cup e_2 \cup \ldots \cup e_{2n}$, there holds
    \begin{equation}
        \label{eq:face_weights_gauge_equivalence}
        W_f=
        \frac{\nu(e_1)\nu(e_3)\ldots \nu(e_{2n-1})}{\nu(e_2)\nu(e_4)\ldots\nu(e_{2n})}=
        \frac{\nu'(e_1)\nu'(e_3)\ldots \nu'(e_{2n-1})}{\nu'(e_2)\nu'(e_4)\ldots\nu'(e_{2n})}.
    \end{equation}
\end{proposition}
The \emph{Kasteleyn matrix} is a weighted adjacency operator with rows indexing the white vertices and columns indexing the black vertices given by its entries
\begin{align*}
    K_{w,b} = \left\{ \begin{array}{cl}
        \sigma(e) \nu(e), & \text{if } e \in E  \\
        0 & \text{otw.}
    \end{array} \right.
\end{align*}
where $\sigma:E \to S^1$ is a complex sign. If for every bounded face with $\partial f = w_1b_1 \cup b_1w_2 \cup \ldots \cup w_nb_n \cup b_nw_1$ we have
\[ \prod_{i=1}^n \frac{\sigma(w_ib_i)}{\sigma(w_ib_{i+1})} = (-1)^{n+1}, \]
then $K$ is said to satisfy the \emph{Kasteleyn condition} and Kasteleyn's celebrated formula first introduced in \cite{Kasteleyn1967} and then generalized to complex signs in \cite{Kuperberg_1998} states
\[ Z(G, \nu) = \abs{\det(K)}. \]
We can thus extend the notions of edge weighting and gauge equivalence to be complex functions $\nu: E \to \mathbb{C}^*$ and $\lambda: V \to \mathbb{C}^*$ respectively and $\nu$ defines a dimer model iff it satisfies the Kasteleyn condition. That is if for all bounded faces $f$ of degree $2n$ we have real face weights with
\begin{equation}
    \label{eq:Kasteleyn_condition}
    \text{sign}(W_f) = (-1)^{n+1}.
\end{equation}

\subsection*{Diamond graph.}
We use the diamond graph $G^\diamond$ to construct our weights. Here $G^\diamond$ has the vertex set $V \sqcup V^*$ and has the edge $vf$ if $v$ is a vertex of $f$ in $G$. That is all faces of $G^\diamond$ are quadrilaterals whose diagonals are in $E(G)$ and $E(G^*)$ respectively.

A train track in $G^\diamond$ is then defined as a sequence of faces $q_j$ of $G^\diamond$ such that the edges $q_{j-1} \cap q_j$ and $q_j \cap q_{j+1}$ are on opposite sides of $q_j$. We represent these train tracks as strands going through the midpoints of the edges $bw$ of $G$ and crossing the quadrilaterals in $F(G^\diamond)$ at opposite sides. Thus every edge $(bw)$ is crossed by two strips. We orient our strands and assume the convention that strands have white vertices to the left and black vertices to the right. This can be done consistently on minimal graphs \cite{boutillier_minimal_2023} of which the lattices we restrict our attention to are examples of, see Figure~\ref{fig:Fock_diamond_and_face}.

\subsection*{Dimers and M-curves.}
We utilize much of the same setup as \cite{BBS} which we summarize here. For further details and proofs we refer the reader to that paper. 

Let's recall the main ingredients needed for the construction of Fock weights \cite{fock_inverse_2015}. These are weights that can be defined on any planar bipartite minimal graph (See \cite{boutillier_minimal_2023}) from data living on a compact Riemann surface with the property that many of the limiting objects of the corresponding dimer model can be described through algebro-geometric constructs on the surface. We define our main objects.
\begin{itemize}
    \item Let $\mathcal{R}$ be a compact Riemann surface of genus $g$ and a fixed homology basis $a_1, \ldots, a_g$; $b_1, \ldots, b_g$.
    \item $\boldsymbol{\omega} = (\omega_1, \ldots, \omega_g)$ is the set of normalized holomorphic differentials dual to that basis. That is $\int_{a_j} \omega_k = \delta_{jk}$ and $\int_{b_j} \omega_k = B_{jk}$.
    \item $J(\mathcal{R})=\mathbb C^g/(\mathbb Z^g+B\mathbb Z^g)$ the Jacobi variety of $\mathcal{R}$;
    \item $A:\mathcal{R}\to J(\mathcal{R})$, $P\mapsto A(P)=\int_{P_0}^{P}\boldsymbol{\omega}\pmod {\mathbb Z^g+B\mathbb Z^g}$ the Abel map.
    \item The theta function with characteristic $\Delta = \left[\begin{array}{l} \Delta_1 \\ \Delta_2\end{array}\right]$ is given by
    \begin{eqnarray}
    \label{eq:def_theta_with_char}
        \theta[\Delta](z) & = & \sum_{m\in \mathbb Z^g} \exp\Big(\pi i\langle B(m+\Delta_1),m+\Delta_1\rangle+2\pi i\langle z+\Delta_2,m+\Delta_1\rangle \Big)
      \nonumber\\
        & = & \exp\Big(\pi i\langle B\Delta_1,\Delta_1\rangle+2\pi i\langle z+\Delta_2,\Delta_1\rangle\Big) \theta(z+\Delta_2+B\Delta_1).
    \end{eqnarray}
    where $\Delta_1, \Delta_2$ are of the form $\frac{1}{2}(\delta_1, \ldots, \delta_g)^T$ with $\delta_i \in \cubr{0,1}$. A characteristic is called even (odd) if $4\pobr{\Delta_1, \Delta_2} = 0 \;(\text{mod } 2)$ (resp. $=1$ (mod $2$)). The special case of $\Delta = 0$ is the theta function with period matrix $B$ and is denoted as 
    \[ \theta(z) = \theta(z|B) = \sum_{m\in\mathbb Z^g}\exp\Big(\pi i\langle Bm,m\rangle+2\pi i \langle z,m\rangle\Big). \]
    \item The Schottky-Klein prime form $E(\alpha, \beta)$ can be expressed through a theta function of odd characteristic $\Delta$ and holomorphic spinor $h_\Delta$ as
    \begin{equation}
    \label{eq:prime_form_E_through_theta}
        E(\alpha,\beta)=\frac{\theta[\Delta]\Big(\int_\alpha^\beta\boldsymbol{\omega}\Big)}{h_\Delta(\alpha)h_\Delta(\beta)}.
    \end{equation}
    The result is independent of the choice of $\Delta$.
    \item Every (directed) train track of $G^\diamond$ has an associated point $\alpha \in \mathcal{R}$. Then the \emph{discrete Abel map} $\eta:V \cup F = V(G^\diamond) \to J(\mathcal{R})$ is defined in the following way. Set $\eta(v_0) = 0 $ for some vertex $v_0 \in V(G^\diamond)$. From there extend the values to all other vertices by adding $A(\alpha)$ whenever crossing a train track with label $\alpha$ that's oriented from left to right and subtracting it if it is oriented from right to left. This is a consistent construction and results in the formulas
    \begin{equation}
        \label{eq:discrete_abel_map}
        \eta(b)=\eta(w)-A(\alpha)-A(\beta), \qquad \eta(f_2)=\eta(f_1)+A(\alpha)-A(\beta).
    \end{equation} 
    with the labeling as in Figure \ref{fig:Fock_diamond_and_face}.
\end{itemize}

\begin{figure}[h]
\centering
\begin{subfigure}{.35\textwidth}
    \centering
    \fontsize{10pt}{12pt}\selectfont
    \def\svgwidth{0.8\linewidth}
\begingroup%
  \makeatletter%
  \providecommand\color[2][]{%
    \errmessage{(Inkscape) Color is used for the text in Inkscape, but the package 'color.sty' is not loaded}%
    \renewcommand\color[2][]{}%
  }%
  \providecommand\transparent[1]{%
    \errmessage{(Inkscape) Transparency is used (non-zero) for the text in Inkscape, but the package 'transparent.sty' is not loaded}%
    \renewcommand\transparent[1]{}%
  }%
  \providecommand\rotatebox[2]{#2}%
  \newcommand*\fsize{\dimexpr\f@size pt\relax}%
  \newcommand*\lineheight[1]{\fontsize{\fsize}{#1\fsize}\selectfont}%
  \ifx\svgwidth\undefined%
    \setlength{\unitlength}{318.09269726bp}%
    \ifx\svgscale\undefined%
      \relax%
    \else%
      \setlength{\unitlength}{\unitlength * \real{\svgscale}}%
    \fi%
  \else%
    \setlength{\unitlength}{\svgwidth}%
  \fi%
  \global\let\svgwidth\undefined%
  \global\let\svgscale\undefined%
  \makeatother%
  \begin{picture}(1,1.41666335)%
    \lineheight{1}%
    \setlength\tabcolsep{0pt}%
    \put(0,0){\includegraphics[width=\unitlength,page=1]{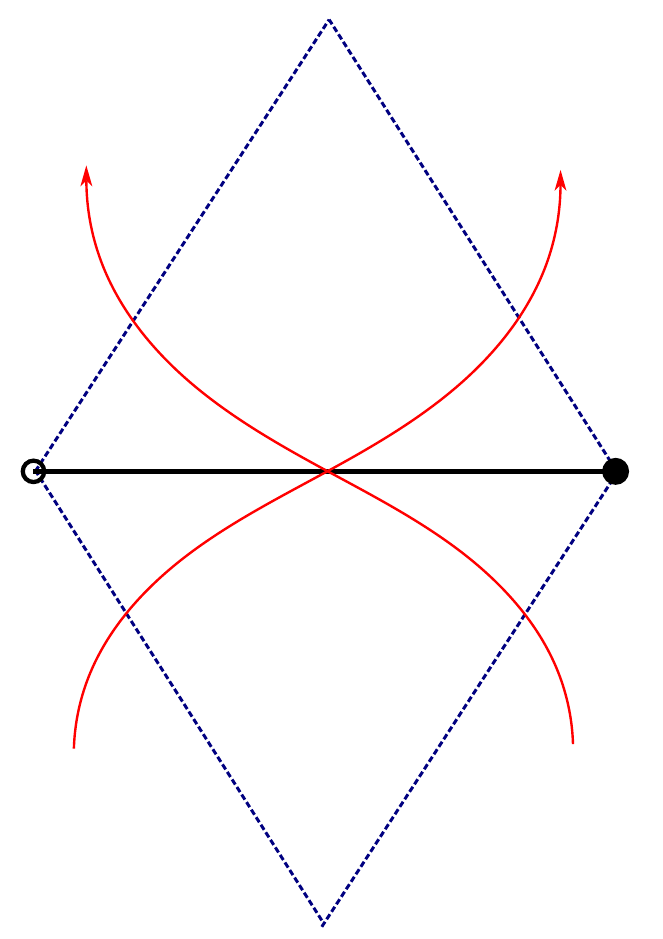}}%
    \put(-0.00214136,0.62226423){\color[rgb]{0,0,0}\makebox(0,0)[lt]{\lineheight{1.25}\smash{\begin{tabular}[t]{l}$w$\end{tabular}}}}%
    \put(0.94835858,0.62780658){\color[rgb]{0,0,0}\makebox(0,0)[lt]{\lineheight{1.25}\smash{\begin{tabular}[t]{l}$b$\end{tabular}}}}%
    \put(0.74329453,1.08504415){\color[rgb]{1,0,0}\makebox(0,0)[lt]{\lineheight{1.25}\smash{\begin{tabular}[t]{l}$\alpha$\end{tabular}}}}%
    \put(0.17798255,1.08504415){\color[rgb]{1,0,0}\makebox(0,0)[lt]{\lineheight{1.25}\smash{\begin{tabular}[t]{l}$\beta$\end{tabular}}}}%
    \put(0,0){\includegraphics[width=\unitlength,page=2]{Figures/FockWeightWithDiamond.pdf}}%
    \put(0.55762833,1.39264038){\color[rgb]{0,0,0}\makebox(0,0)[lt]{\lineheight{1.25}\smash{\begin{tabular}[t]{l}$f_2$\end{tabular}}}}%
    \put(0.55485712,0.01815644){\color[rgb]{0,0,0}\makebox(0,0)[lt]{\lineheight{1.25}\smash{\begin{tabular}[t]{l}$f_1$\end{tabular}}}}%
    \put(0,0){\includegraphics[width=\unitlength,page=3]{Figures/FockWeightWithDiamond.pdf}}%
  \end{picture}%
\endgroup%

\end{subfigure}
\begin{subfigure}{.6\textwidth}
    \centering
    \fontsize{10pt}{12pt}\selectfont
    \def\svgwidth{.9\textwidth}
    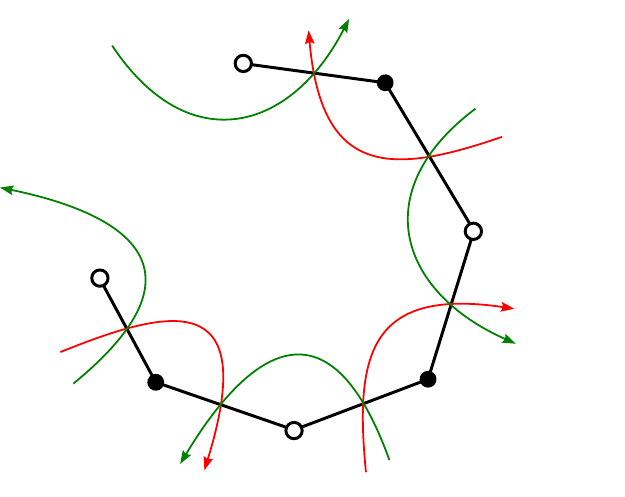
\end{subfigure}
  \caption{\textbf{Left}: A quadrilateral face of $G^\diamond$, with two directed strips through it. The discrete Abel map is defined on all $4$ vertices. The dimer model is defined on the black and white graph, the height function on its dual graph with square vertices. \textbf{Right}: A face of the graph $G$ surrounded by two different types of train tracks.}
  \label{fig:Fock_diamond_and_face}
\end{figure}

\begin{definition}
    Let $G$ be a planar bipartite graph with labels $\alpha\in \mathcal{R}$ on every train track. The \emph{Fock weight} for an edge $wb$ crossed by train track with associated labels $\alpha, \beta$ and adjacent to faces $f_1, f_2$ is then given by
    \begin{equation}
        \label{eq:def_Fock_weights}
        K_{wb}(\alpha,\beta)=\frac{E(\alpha,\beta)}{\theta(\eta(f_1)+D)\theta(\eta(f_2)+D)}.
    \end{equation}
    Here $D\in \mathbb{C}^g$ is arbitrary and $\eta$ is the discrete Abel map.
\end{definition}
In order for these complex weights to satisfy the Kasteleyn condition \eqref{eq:Kasteleyn_condition} and thus to define a dimer model, the face weights need to be real and have the correct sign. For a face $f \in F(G)$ as in Figure \ref{fig:Fock_diamond_and_face} surrounded by faces $f_1, \ldots, f_{2n}$ we have
\begin{equation}
    \label{eq:Fock_face_weight}
    W_f = \prod_{i=1}^n \frac{\theta[\Delta]\br{\int_{\alpha_i}^{\beta_i} \boldsymbol{\omega}} }{\theta[\Delta]\br{\int_{\beta_i}^{\alpha_{i+1}} \boldsymbol{\omega}} } \frac{\theta\br{\eta(f_{2i}) + D}}{\theta\br{\eta(f_{2i- 1}) + D}}.
\end{equation}
Now let $\mathcal{R}$ be an \emph{M-curve}. That is a Riemann surface of genus $g$ with an anti-holomorphic involution $\tau: \calR \to \calR$ which fixes the maximal possible number of $g+1$ real ovals $X = X_0 \sqcup \ldots \sqcup X_g$ and let all labels $\alpha_i, \beta_i$ lie on $X_0$. Then $W_f$ is real. Furthermore a cyclic condition guarantees the correct sign.
\begin{theorem}
    \label{thm:Kasteleyn_condition}
    Consider an orientation preserving homeomorphism $X_0\to S^1$, and denote by $\hat\alpha_i,\hat\beta_i\in S^1$ the images of $\alpha_i,\beta_i$ under this map. Then the Kasteleyn condition is equivalent to
    \begin{equation}
    \label{eq:Kasteleyn_clustering_general_face}
        {\rm sign}\left(\frac{(\hat\beta_1-\hat\alpha_1)(\hat\beta_2-\hat\alpha_2)\cdots(\hat\beta_n-\hat\alpha_n)}
        {(\hat\alpha_1-\hat\beta_2)(\hat\alpha_2-\hat\beta_3)\cdots (\hat\alpha_{n}-\hat\beta_1)}\right)=(-1)^{n+1}
    \end{equation}
    for all faces $f$ of $G$.
\end{theorem}

\subsection*{Quasi-periodic weights.}
We make this condition more concrete as we apply this setup to two examples of regular grids. For the \emph{regular square lattice} we have have vertical and horizontal train tracks of two different orientations respectively. We will consider periodically repeating labels $\alpha_i^\pm, \beta_j^\pm, \; i\leq m,\; j\leq n$. The Kasteleyn condition \eqref{eq:Kasteleyn_clustering_general_face} is then satisfied for every face iff the cyclic order
\begin{equation}
    \label{eq:clustering_sq}
    \alpha_i^+<\beta_j^+<\alpha_k^-<\beta_l^-<\alpha_i^+ \quad {\rm for\;\;all}\quad i,j,k,l,
\end{equation}
holds on $X_0$. See Figure \ref{fig:sq_lattice} for notation.

\begin{figure}[h]
\centering
\begin{subfigure}{.7\textwidth}
  \centering
    \fontsize{10pt}{12pt}\selectfont
    \def\svgwidth{0.9\linewidth}
    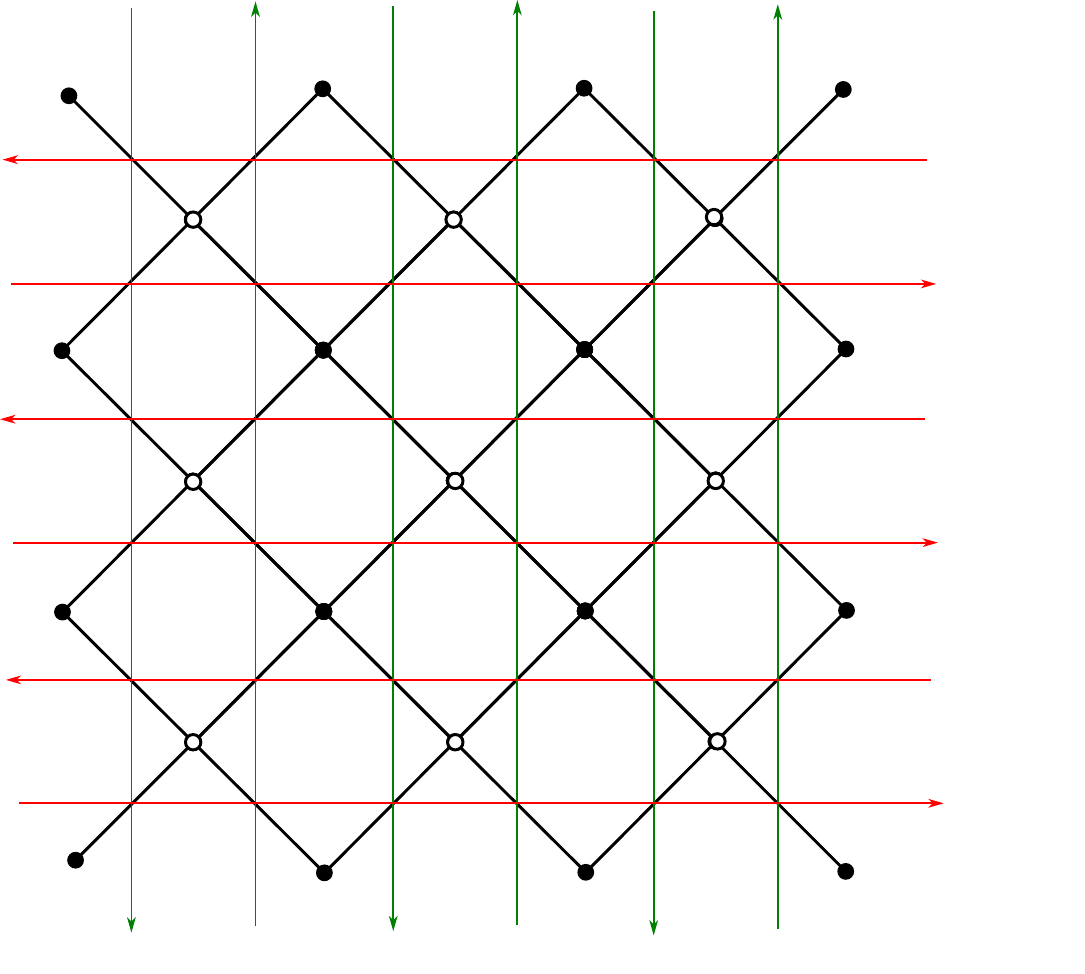
\end{subfigure}
\begin{subfigure}{.28\textwidth}
  \begin{subfigure}[t]{\textwidth}
    \centering
      \fontsize{10pt}{12pt}\selectfont
      \def\svgwidth{0.9\linewidth}
\begingroup%
  \makeatletter%
  \providecommand\color[2][]{%
    \errmessage{(Inkscape) Color is used for the text in Inkscape, but the package 'color.sty' is not loaded}%
    \renewcommand\color[2][]{}%
  }%
  \providecommand\transparent[1]{%
    \errmessage{(Inkscape) Transparency is used (non-zero) for the text in Inkscape, but the package 'transparent.sty' is not loaded}%
    \renewcommand\transparent[1]{}%
  }%
  \providecommand\rotatebox[2]{#2}%
  \newcommand*\fsize{\dimexpr\f@size pt\relax}%
  \newcommand*\lineheight[1]{\fontsize{\fsize}{#1\fsize}\selectfont}%
  \ifx\svgwidth\undefined%
    \setlength{\unitlength}{322.25247532bp}%
    \ifx\svgscale\undefined%
      \relax%
    \else%
      \setlength{\unitlength}{\unitlength * \real{\svgscale}}%
    \fi%
  \else%
    \setlength{\unitlength}{\svgwidth}%
  \fi%
  \global\let\svgwidth\undefined%
  \global\let\svgscale\undefined%
  \makeatother%
  \begin{picture}(1,0.85773971)%
    \lineheight{1}%
    \setlength\tabcolsep{0pt}%
    \put(0,0){\includegraphics[width=\unitlength,page=1]{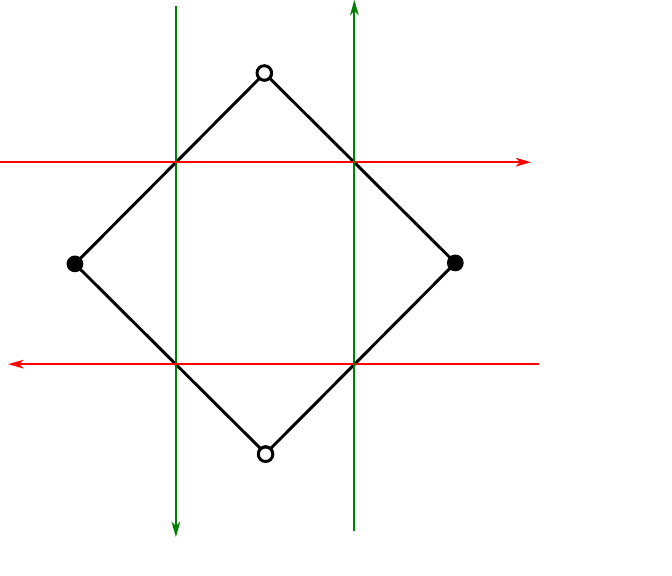}}%
    \put(0.75295355,0.6568396){\color[rgb]{1,0,0}\makebox(0,0)[lt]{\lineheight{1.25}\smash{\begin{tabular}[t]{l}$\beta_{j+1}^-$\end{tabular}}}}%
    \put(0.75676921,0.35412353){\color[rgb]{1,0,0}\makebox(0,0)[lt]{\lineheight{1.25}\smash{\begin{tabular}[t]{l}$\beta_j^+$\end{tabular}}}}%
    \put(0.23273982,0.0068903){\color[rgb]{0,0.50196078,0}\makebox(0,0)[lt]{\lineheight{1.25}\smash{\begin{tabular}[t]{l}$\alpha_i^-$\end{tabular}}}}%
    \put(0.51256123,0.00561841){\color[rgb]{0,0.50196078,0}\makebox(0,0)[lt]{\lineheight{1.25}\smash{\begin{tabular}[t]{l}$\alpha_i^+$\end{tabular}}}}%
  \end{picture}%
\endgroup%

  \end{subfigure}
  \begin{subfigure}[b]{\textwidth}
    \centering
      \fontsize{10pt}{12pt}\selectfont
      \def\svgwidth{0.9\linewidth}
\begingroup%
  \makeatletter%
  \providecommand\color[2][]{%
    \errmessage{(Inkscape) Color is used for the text in Inkscape, but the package 'color.sty' is not loaded}%
    \renewcommand\color[2][]{}%
  }%
  \providecommand\transparent[1]{%
    \errmessage{(Inkscape) Transparency is used (non-zero) for the text in Inkscape, but the package 'transparent.sty' is not loaded}%
    \renewcommand\transparent[1]{}%
  }%
  \providecommand\rotatebox[2]{#2}%
  \newcommand*\fsize{\dimexpr\f@size pt\relax}%
  \newcommand*\lineheight[1]{\fontsize{\fsize}{#1\fsize}\selectfont}%
  \ifx\svgwidth\undefined%
    \setlength{\unitlength}{301.01809837bp}%
    \ifx\svgscale\undefined%
      \relax%
    \else%
      \setlength{\unitlength}{\unitlength * \real{\svgscale}}%
    \fi%
  \else%
    \setlength{\unitlength}{\svgwidth}%
  \fi%
  \global\let\svgwidth\undefined%
  \global\let\svgscale\undefined%
  \makeatother%
  \begin{picture}(1,0.91824627)%
    \lineheight{1}%
    \setlength\tabcolsep{0pt}%
    \put(0,0){\includegraphics[width=\unitlength,page=1]{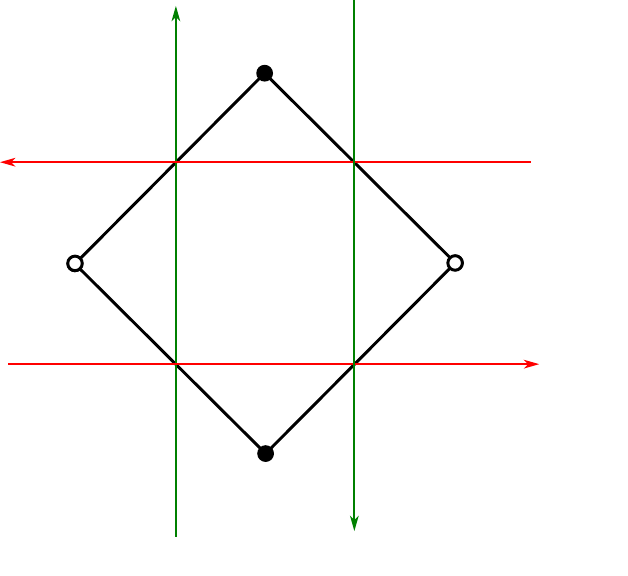}}%
    \put(0.8060683,0.70317428){\color[rgb]{1,0,0}\makebox(0,0)[lt]{\lineheight{1.25}\smash{\begin{tabular}[t]{l}$\beta_j^+$\end{tabular}}}}%
    \put(0.81015312,0.37910406){\color[rgb]{1,0,0}\makebox(0,0)[lt]{\lineheight{1.25}\smash{\begin{tabular}[t]{l}$\beta_j^-$\end{tabular}}}}%
    \put(0.24915772,0.00737635){\color[rgb]{0,0.50196078,0}\makebox(0,0)[lt]{\lineheight{1.25}\smash{\begin{tabular}[t]{l}$\alpha_i^+$\end{tabular}}}}%
    \put(0.54871825,0.00601474){\color[rgb]{0,0.50196078,0}\makebox(0,0)[lt]{\lineheight{1.25}\smash{\begin{tabular}[t]{l}$\alpha_{i+1}^-$\end{tabular}}}}%
  \end{picture}%
\endgroup%

\end{subfigure}
\end{subfigure}
\caption{\textbf{Left}: Regular square lattice on a $3 \times 3$ fundamental domain along with the train track parameters. We have $4$ different kinds of trans tracks and attach the parameters $\cubr{\alpha_i^-, \beta_j^-, \alpha_k^+, \beta_l^+}$ to them. \textbf{Right}: The strips for two types of square faces.}
\label{fig:sq_lattice}
\end{figure}

For the \emph{regular hexagonal lattice} there are three kinds of train tracks. Again we consider periodically repeating labels $\alpha_i, \beta_j, \gamma_k,\; i,j,k\leq n $ on them. See Figure \ref{fig:hex_lattice} for notation. The Kasteleyn condition then reads
\begin{equation}
    \label{eq:clustering_hex}
    \alpha_i<\beta_j<\gamma_k<\alpha_i \quad {\rm for\;\;all}\quad i,j,k.
\end{equation}

\begin{figure}[h]
\centering
\begin{subfigure}{.7\textwidth}
  \centering
    \fontsize{10pt}{12pt}\selectfont
    \def\svgwidth{0.9\linewidth}
    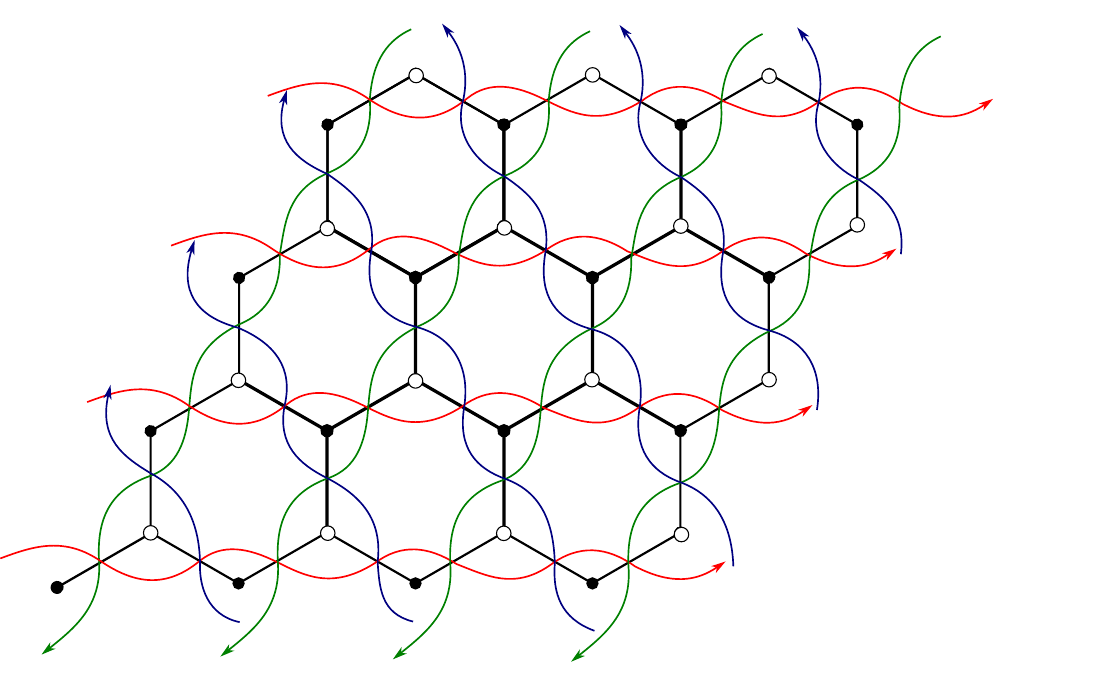
\end{subfigure}
\begin{subfigure}{.28\textwidth}
  \centering
    \fontsize{10pt}{12pt}\selectfont
    \def\svgwidth{0.9\linewidth}
\begingroup%
  \makeatletter%
  \providecommand\color[2][]{%
    \errmessage{(Inkscape) Color is used for the text in Inkscape, but the package 'color.sty' is not loaded}%
    \renewcommand\color[2][]{}%
  }%
  \providecommand\transparent[1]{%
    \errmessage{(Inkscape) Transparency is used (non-zero) for the text in Inkscape, but the package 'transparent.sty' is not loaded}%
    \renewcommand\transparent[1]{}%
  }%
  \providecommand\rotatebox[2]{#2}%
  \newcommand*\fsize{\dimexpr\f@size pt\relax}%
  \newcommand*\lineheight[1]{\fontsize{\fsize}{#1\fsize}\selectfont}%
  \ifx\svgwidth\undefined%
    \setlength{\unitlength}{298.53233119bp}%
    \ifx\svgscale\undefined%
      \relax%
    \else%
      \setlength{\unitlength}{\unitlength * \real{\svgscale}}%
    \fi%
  \else%
    \setlength{\unitlength}{\svgwidth}%
  \fi%
  \global\let\svgwidth\undefined%
  \global\let\svgscale\undefined%
  \makeatother%
  \begin{picture}(1,0.98848374)%
    \lineheight{1}%
    \setlength\tabcolsep{0pt}%
    \put(0,0){\includegraphics[width=\unitlength,page=1]{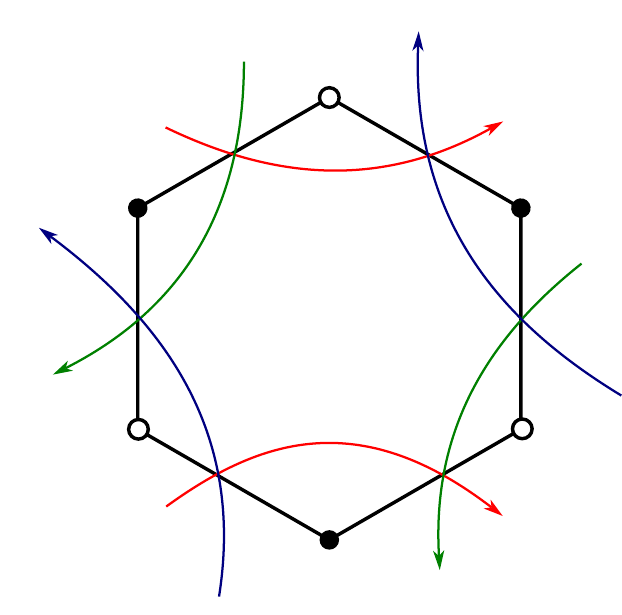}}%
    \put(-0.00046807,0.65043186){\color[rgb]{0,0,0.50196078}\makebox(0,0)[lt]{\lineheight{1.25}\smash{\begin{tabular}[t]{l}$\gamma_1$\end{tabular}}}}%
    \put(0.63933933,0.96347071){\color[rgb]{0,0,0.50196078}\makebox(0,0)[lt]{\lineheight{1.25}\smash{\begin{tabular}[t]{l}$\gamma_2$\end{tabular}}}}%
    \put(0.82469128,0.12046251){\color[rgb]{1,0,0}\makebox(0,0)[lt]{\lineheight{1.25}\smash{\begin{tabular}[t]{l}$\beta_1$\end{tabular}}}}%
    \put(0.83155622,0.77674581){\color[rgb]{1,0,0}\makebox(0,0)[lt]{\lineheight{1.25}\smash{\begin{tabular}[t]{l}$\beta_2$\end{tabular}}}}%
    \put(0.03111038,0.32778213){\color[rgb]{0,0.50196078,0}\makebox(0,0)[lt]{\lineheight{1.25}\smash{\begin{tabular}[t]{l}$\alpha_1$\end{tabular}}}}%
    \put(0.68464776,0.00787842){\color[rgb]{0,0.50196078,0}\makebox(0,0)[lt]{\lineheight{1.25}\smash{\begin{tabular}[t]{l}$\alpha_2$\end{tabular}}}}%
  \end{picture}%
\endgroup%

\end{subfigure}
\caption{\textbf{Left}: Regular hexagonal lattice on a $3\times 3$ fundamental domain. \textbf{Right}: The strips for a typical hexagonal face.}
\label{fig:hex_lattice}
\end{figure}
\begin{rem}
    Note that periodically repeating train track labels do not imply periodic weights as the discrete Abel map $\eta$ does not have to be periodic. In general this construction gives \textbf{quasi-periodic} weights. Periodicity is however guaranteed under a Abel-Jacobi condition on the choices of train track labels which then yields the special case of doubly periodic weights and a Harnack curve representation for $\calR$.This setup covers \textbf{all} doubly periodic weights and extends their class in a non-trivial way while still allowing for spectral methods of study. In fact in some ways it makes the study simpler as the Abel-Jacobi condition that is implicitly assumed with periodicity can be relaxed allowing e.g. for continuous deformations of the spectral data. See \cite[Section~7]{BBS}. 
\end{rem}

\subsection*{Amoeba, polygon, surface tension and Ronkin function.}
Let $\calR$ be an M-curve with antiholomorphic involution $\tau$ fixing $X = \bigcup_{i=0}^g X_i$ the set of real ovals. Then $X$ separates $\calR$ into two components $\calR_+, \calR_-$ with $\tau \calR_+ = \calR_-$ and $\calR_+ = \calR / \tau$. The corresponding open components are denoted by $\calR^\circ_\pm := \calR_\pm \setminus X$. For concreteness we lay out the following definition for the square and hexagonal grid. A general definition is done equivalently.

Let $\cubr{\alpha, \beta} = \cubr{\alpha^\pm_1, \ldots, \alpha^\pm_n, \beta^\pm_1, \ldots, \beta^\pm_n}$ lie on $X_0$ adhering to the clustering condition \eqref{eq:clustering_sq}. Then we call $\mathcal{S} = \cubr{\calR, \cubr{\alpha, \beta}}$ \emph{Harnack data for the square grid} with an $(n \times n)$ fundamental domain.Let us introduce   the Abelian differential of the third kind $d\zeta^{\alpha_i}$ with simple poles at $\alpha_i^\pm$, the residues
\begin{equation}
    \label{eq:residues_of_dZeta_i_sq}
    \res_{\alpha_i^-} d\zeta^{\alpha_i} = - \res_{\alpha_i^+} d\zeta^{\alpha_i} = 1,
\end{equation}
and purely imaginary periods. Define $d\zeta^{\beta_j}$ similary and introduce the two differentials
\begin{equation}
    \label{eq:dzeta_1,2}
    d\zeta_1 = \sum_{i = 1}^n d\zeta^{\alpha_i}, \quad d\zeta_2 = \sum_{i = 1}^n d\zeta^{\beta_i}.
\end{equation}

For the hexagonal grid let $\cubr{\alpha, \beta, \gamma} = \cubr{\alpha_1,\ldots,\alpha_n, \beta_1, \ldots, \beta_n, \gamma_1, \ldots, \gamma_n}$ all lie on $X_0$ and satisfy the clustering condition \eqref{eq:clustering_hex}. We then call $\mathcal{S} = \cubr{\calR,\cubr{\alpha,\beta,\gamma}}$ \emph{Harnack data for the hexagonal grid} with an $(n\times n)$ fundamental domain. Analogously, let us introduce the Abelian differentials of third kind $d\zeta^{\alpha_i}, d\zeta^{\beta_i}$ with residues
\begin{equation}
    \label{eq:residues_of_dZeta_i_hex}
    \res_{\alpha_i} d\zeta^{\alpha_i} = - \res_{\gamma_i} d\zeta^{\alpha_i} = 1 = \res_{\beta_i} d\zeta^{\beta_i} = - \res_{\gamma_i} d\zeta^{\beta_i},
\end{equation}
and purely imaginary periods. In this case the two differentials $d\zeta_i$ are also given by formula \eqref{eq:dzeta_1,2}.

Let $\mathcal{S} = \cubr{\calR, \mathcal{T}}$ be one of the two types of Harnack data where $\mathcal{T}$ is a collection of points on $X_0$. The $d\zeta_k$ are real on $X$ and $d\zeta_k(\tau P) = \overline{d\zeta_k(P)}$. This yields well defined integrals
\[ \zeta_k(P) = \int_Q^P d\zeta_k \]
on $\mathcal{R}_+ \setminus \mathcal{T}$ with $Q \in X_0$ fixed. Their real and imaginary parts are denoted as $\zeta_k(P) = x_k + i y_k$.
\begin{definition}
\label{def:amoeba_map_polygon_map}
    The map
    \[ \mathcal{A}(P) := (x_1, x_2) = (\re \zeta_1(P), \re \zeta_2(P)) \]
    is called the \emph{amoeba map}, and its image 
    \[ \mathcal{A}(\mathcal{R}_+ \setminus \mathcal{T}) =: \mathcal{A}_\mathcal{S} \subset \mathbb{R}^2 \]
    is called the \emph{amoeba}.
    
    Similarly the map
    \[ \Delta(P) = (s_1, s_2) = R(y_1, y_2) := \frac{1}{\pi}(-y_2, y_1) = \frac{1}{\pi}(- \im \zeta_2(P), \im \zeta_1(P)) \]
    we call the \emph{polygon map}, and its image
    \[ \Delta(\mathcal{R}_+ \setminus \mathcal{T}) =: \Delta_\mathcal{S} \subset \mathbb{R}^2 \]
    is called the \emph{Newton polygon}.

    We denote the images of $\calR^\circ_+$ under $\mathcal{A}, \Delta$ as $\mathcal{A}_\mathcal{S}^\circ, \Delta_\mathcal{S}^\circ$ respectively.
\end{definition}
See Figure \ref{fig:all_maps_and_zeros} for an example of these two maps. In the periodic case these two notions coincide with the algebraic definitions of the amoeba map and the map to the Newton polygon, see Table~\ref{tab:translation_aztec}.

We summarize basic facts about the introduced maps, and also about (generalized) Ronkin and surface tension functions, referring to \cite{BBS} for details and proofs.

Every bounded component of $\mathbb{R}^2\setminus \mathcal{A}_\mathcal{S}$ corresponds to a real oval $X_i, i\neq 0$ and we will denote it as $\mathcal{A}_i$. Every labeled point $\alpha \in \mathcal{T}$ gets mapped to infinity and has an associated tentacle with asymptotic direction $-(\res_\alpha d\zeta_1, \res_\alpha d\zeta_2)$. Let $\ell_{\alpha, \beta} \subset X_0$ be the arc between $\alpha, \beta \in \mathcal{T}$ not containing any other points in $\mathcal{T}$. The unbounded region of $\mathbb{R}^2 \setminus \mathcal{A}_\mathcal{S}$ with boundary $\mathcal{A}(\ell_{\alpha, \beta})$ will be denoted as $\mathcal{A}_{\ell_{\alpha, \beta}}$. Every component of $\mathbb{R}^2 \setminus \mathcal{A}_\mathcal{S}$ is convex.

The polygon map $\Delta$ maps any such arc $\ell_{\alpha, \beta}$ to an integer lattice point in $\mathbb{Z}^2$. The image $\Delta_\mathcal{S}$ is then the convex hull of these points. The other real ovals $X_i$ get mapped to individual points inside $\Delta_\mathcal{S}$. Note that these do not have to be integer points in the quasi-periodic setup.

The two maps
\[ \mathcal{A} : \calR^\circ_+ \to \mathcal{A}_\mathcal{S}^\circ,\quad \Delta: \calR^\circ_+ \to \Delta_\mathcal{S}^\circ \]
are diffeomorphisms. Thus  $(x_1, x_2)$ and $(s_1, s_2)$ are global coordinates of $\calR^\circ_+$. The functions
\begin{align}
    \label{eq:def_rho_sigma}
    \rho(P) = \rho(x_1, x_2) &= -\frac{1}{\pi} \im \int_\ell \zeta_2 d\zeta_1 + x_2 s_2 \\
    \sigma(P) = \sigma(s_1, s_2) &= \frac{1}{\pi} \im \int_\ell \zeta_2 d\zeta_1 - x_1 s_1
\end{align}
are called the \emph{(generalized) Ronkin function and surface tension}, see \cite{BBS}. 
Here $\ell$ is a path from some fixed $Q\in X_0$ to $P \in \calR^\circ_+$. 
The functions $\rho$ and $\sigma$ are convex on $\mathcal{A}^\circ_\mathcal{S}$ and $\Delta^\circ_\mathcal{S}$ respectively and are Legendre dual
\begin{equation}
    \label{eq:rho_sigma_Legendre_duality_s_coords}
    \sigma(s_1, s_2) = -\rho(x_1, x_2) + x_1 s_1 + x_2 s_2,
\end{equation}
or equivalently
\begin{equation}
    \label{eq:rho_sigma_Legendre_via_gradients}
    \nabla\sigma(s_1, s_2) = (x_1, x_2),\quad \nabla\rho(x_1, x_2) = (s_1, s_2). 
\end{equation}
From \eqref{eq:rho_sigma_Legendre_via_gradients} and the fact that $\Delta$ maps real ovals to points it can be seen that $\rho$ can be affinely extended to each connected component of $\mathbb{R}^2\setminus\mathcal{A}_\mathcal{S}$. We will henceforth denote this extension by $\rho : \mathbb{R}^2 \to \mathbb{R}$. The surface tension $\sigma$ has singularities at all marked points.

\begin{rem}
    \label{rem:uniformization_calR}
Due to the Koebe uniformization theorem every finitely connected planar domain is conformally equivalent to a circle domain.   So,  we can map $\calR_+$ conformally to the unit disk with $g$ round holes such that $X_0$ gets mapped to $S^1$. This along with the marked train track parameters on $S^1$ defines the Harnack data that we included in figures throughout this paper, e.g. Figure~\ref{fig:main_aztec_config}. In these figures the canonical coordinate $z$ encodes the complex structure of $\calR$ and we encourage the reader to think of this as the actual Riemann surface.
\end{rem}

\begin{table}
    \begin{center}
        \begin{tabular}{ |c|c|c| } 
            \hline
            Object & doubly periodic & quasi periodic \\
            \hline \hline
            Harnack data & $\det(P(z,w)) = 0$ & $\cubr{\mathcal{R}, \mathcal{T}}$ \\ \hline
            Monodromies & $(z, w)$ & $(\frac{\psi_{(1,0)}}{\psi_{(0,0)}}, \frac{\psi_{(0,1)}}{\psi_{(0,0)}})$ \\ \hline
            Main differentials & $(\frac{dz}{z}, \frac{dw}{w})$ & $(d\zeta_1, d\zeta_2)$ \\ \hline
            Amoeba map & $(\log|z|, \log|w|)$ & $(\re \zeta_1, \re \zeta_2)$ \\ \hline
            Ronkin function & $\frac{1}{(2 \pi i)^2} \int\int_{\mathbb{T}^2} \log |P(e^{x_1} z, e^{x_2} w)| \frac{dz}{z} \frac{dw}{w}$ & $-\frac{1}{\pi} \im \int_\ell \zeta_2 d\zeta_1 + x_2 s_2$\\
            \hline
        \end{tabular}
    \end{center}
\caption{For readers familiar with the doubly periodic setup we include this dictionary translating some of the main objects.}
\label{tab:translation_aztec}
\end{table}

\subsection*{The height function.}
The face weights $W_f$ have a natural interpretation in terms of the \emph{height function}. Any dimer configuration $D$ induces a discrete $1$-form $\omega$ via
\[ \omega_D(bw) = -\omega_D(wb) = \left\{ \begin{array}{l}
    1, \text{ if } (bw) \in D,\\
    0, \text{ otw.}
\end{array} \right. \]
Observe that $\omega$ has divergence $d^*\omega_D(v) = \sum_{v'\sim v} \omega_D(vv')$ of $1$ at black vertices and $-1$ at white vertices. Such a discrete $1$-form is called a unit form.

Take a fixed unit flow $\omega_0$ and consider $\omega := \omega_D - \omega_0$. Then $\omega$ is divergence free at all vertices and thus the dual form $\omega^*$ is exact, i.e. $\omega^* = dh$ for some function $h_D = h: F \to \mathbb{R}$. This $h$ is called the height function and is well defined up to a global constant. We will use the natural reference flows $\omega_0 = \frac{1}{4}$ and $\omega_0 = \frac{1}{3}$ for each edge oriented from black to white and on the square and hexagonal grid respectively. These are standard choices \cite{thurston1989groups} but note that on a finite planar graph $G$ that is a subset of the corresponding grid the vertices on the boundary do not have divergence $\pm 1$. Consequently the dual form $\omega^*$ is not divergence free around the boundary. We therefore consider a modified version of $G^*$ where for every boundary edge $e$ its dual edge $e^*$ is defined to go to a distinct vertex of degree $1$ instead of one vertex corresponding to the unbounded face of $G$. Then $h$ is well defined on this modification of $G^*$. It is independent of the choice of the dimer configuration $D$ and thus is deterministic around the boundary of $G$. Thus there is a bijection between the dimer configurations on any planar bipartite graph and height functions on its modified dual. The Boltzmann measure \eqref{eq:def_Boltzmann_measure} expressed in terms of $h$ then takes the form
\begin{equation}
    \label{eq:Boltzmann_measure_height_function}
    \mathbb{P}(h) = \frac{W(h)}{Z(G, W)} = \frac{\prod_{f\in F} W_f^{h(f)} }{\sum_{D\in \mathcal{D}} \prod_{f\in F} W_f^{h_D(f)} }.
\end{equation}

For a bounded Lipschitz domain $\Omega \in \mathbb{R}^2$, a set $E \subset \mathbb{R}^2$ and a function $f_b: \partial\Omega \to \mathbb{R}$ one defines the function space
\[ \Lip_E(\Omega, f_b) = \cubr{ f \in \Lip(\Omega, \mathbb{R}) \;|\; \nabla f(z) \in E \text{ a.e.}, \; \restr{f}{\partial\Omega} = f_b }. \]
A boundary condition $f_b$ is called \emph{feasible} if this space is not empty.

It was shown in \cite{cohn_variational_2000} for uniform weights and then extended to doubly periodic weights in \cite{kuchumov_limit_2017,Kenyon_Okounkov_Sheffield_2006} that the average height function of a dimer model on a graph approximating the domain $\Omega$ and the height function on the boundary approximating $f_b$ converges to the unique solution of the surface tension minimization problem
\begin{equation}
    \label{eq:def_surface_tension_functional}
    \argmin_{f \in \Lip_{\Delta_\mathcal{S}}(\Omega, f_b)} \int_\Omega \sigma(\nabla f).
\end{equation}
This was generalized to quasi-periodic weights in \cite{BBS} in a regularized sense. We will denote the surface tension functional as $I_\sigma(f) = \int_\Omega \sigma(\nabla f)$.

The Newton polygon $\Delta_\mathcal{S}$ is only well defined up to translations and the surface tension $\sigma$ and height function $h$ are only well defined up to affine transformations. We will henceforth assume that our normalization yields that the height function has gradients in the Newton polygon. Also affine transformations of $\sigma$ do not alter the minimization problem \eqref{eq:def_surface_tension_functional}. See \cite[Remark~38]{BBS}.

\subsection*{Boundary conditions.}
In this paper we will study the special cases of the Aztec diamond and the hexagon. Classically the Aztec diamond $A_N$ is a subset of the square grid defined by those vertices $(x,y)$ for which $|x| + |y| \leq N$. In the limit $\frac{1}{N} A_N$ converges to the diamond domain $\cubr{(x,y): |x| + |y| \leq 1}$. 

For weights defined on a fundamental domain of size $n\times n$ we will consider the embedding of $A_{nN}$ under the normalization such that the square grid is embedded diagonally as in Figure~\ref{fig:sq_lattice} and so that $A_{nN}$ converges to the domain $[-1, 1]^2$. Under this normalization the height function $h$ on the boundary is a piecewise linear function defined by 
\[ \partial_x \restr{h}{y=-1} = 0,\; \partial_y \restr{h}{x=1} = n,\; \partial_x \restr{h}{y=1} = -n,\; \partial_y \restr{h}{x=-1} = 0. \]
We note that the Newton polygon here is
\[ \Delta = [0, n]^2 = \text{conv}((0,0), (n,0), (n,n), (0,n)) =: \text{conv}(a,b,c,d) \]
and a function $f: [-1, 1]^2\to \mathbb{R}$ satisfies these boundary conditions if the gradient of $f$ lies on the appropriate parts of $\partial \Delta$. More specifically if
\begin{equation}
    \label{eq:aztec_boundary_height}
    \restr{\nabla f}{y=-1} \in \overline{ad},\; \restr{\nabla f}{x=1} \in \overline{dc},\; \restr{\nabla f}{y=1} \in \overline{cb},\; \restr{\nabla f}{x=-1} \in \overline{ba}.
\end{equation}

Similarly we will consider the regular hexagon $G_N$ as a subset of the hexagonal grid. For a fundamental domain of size $n \times n$  the embedding of $G_{nN}$ converges to 
\[ \mathcal{H} := \cubr{(u,v) \in [-1, 1]^2 \;|\; -1 \leq u + v \leq 1}. \]
This means that the hexagons in our embedding are not regular but of the same shape as $\mathcal{H}$. This is just a global linear transform and of course does not change any statistical properties. This embedding, however, gives convenient coordinates. Indeed, the Newton polygon here is the triangle $\Delta = \text{conv}((0,0), (0,n), (n,0)) =: \text{conv}(a,b,c)$ and in analogy to \eqref{eq:aztec_boundary_height} we get that a function $f: \mathcal{H} \to \mathbb{R}$ satisfies the hexagon boundary conditions if its gradients live on the appropriate part of the boundary of the Newton polygon
\begin{equation}
    \label{eq:hex_boundary_height}
    \restr{\nabla f}{x = \pm 1} \in \overline{ac} ,\; \restr{\nabla f}{x + y = \pm 1} \in \overline{cb},\; \restr{\nabla f}{y = \pm 1} \in \overline{ba}.
\end{equation}

\section{Aztec Diamond Height Function}
\label{sec:03_CoordinateMap}

\addtocontents{toc}{\protect\setcounter{tocdepth}{2}}

We now study the square grid with Aztec diamond boundary conditions. To that end let $\mathcal{R}$ be an M-curve with involution $\tau$, real ovals $X_0, \ldots, X_g$ and 
\[\cubr{\alpha, \beta} = \cubr{\alpha^\pm_1, \ldots, \alpha^\pm_n, \beta^\pm_1, \ldots, \beta^\pm_n} \in X_0\]
the corresponding train track parameters satisfying the cyclic order property \eqref{eq:clustering_sq} and thus defining a dimer model on the square grid as in Figure \ref{fig:sq_lattice}. We will consider the sequence of graphs $A_{nN}$ embedded as described in Section~\ref{sec:02_TheDimerModel} so that they converge in the Caratheodory sense to the domain $[-1, 1]^2$ with boundary conditions as in \eqref{eq:aztec_boundary_height}.

\begin{table}
    \begin{center}
        \begin{tabular}{ |c|c|c|c|c| } 
            \hline
            res & $\alpha^-_i$ & $\alpha^+_i$ & $\beta^-_i$ & $\beta^+_i$ \\
            \hline
            $d\zeta_1$ & 1 & -1 & 0 & 0 \\ 
            $d\zeta_2$ & 0 & 0 & 1 & -1 \\ 
            $d\zeta_3$ & 1 & 1 & -1 & -1 \\ 
            \hline
        \end{tabular}
    \end{center}
\caption{Residues of the differentials $d\zeta_i$ for the Aztec diamond.}
\label{tab:xi_residues_aztec}
\end{table}

Let us introduce $d\zeta_3$ as the meromorphic differential on $\mathcal{R}$ with poles in $\cubr{\alpha^\pm_i, \beta^\pm_j}$ with residues as in Table \ref{tab:xi_residues_aztec} and imaginary periods. We note that $d\zeta_3$ is linearly independent of $d\zeta_1, d\zeta_2$ as defined in \eqref{eq:dzeta_1,2} and for any $(u,v) \in \mathbb{R}^2$ we define the differential
\begin{equation}
    \label{eq:def_dZeta}
    d\zeta_{(u,v)} = d\zeta = -u d\zeta_2 + v d\zeta_1 + d\zeta_3.
\end{equation}

When convenient we will omit the explicit dependence on $(u,v)$ and just talk about the differential $d\zeta$. This differential is the central piece for our construction of the map between the spectral curve and the Aztec diamond coordinates so let us study its properties.

\begin{rem}
    The family of differentials $d\zeta$ was introduced and studied in \cite{berggren_geometry_2023} in the doubly periodic setup. We note that in their formulation $d\zeta = (1-u) d\zeta_2 - (1-v) d\zeta_1 - d\tilde{\zeta}_3$ where $d\tilde{\zeta}_3$ has $\res_{\alpha^-_i} \;d\tilde{\zeta}_3 = -1$, $\res_{\beta^-_i} \;d\tilde{\zeta}_3 = 1$ and $0$ otherwise. This is equivalent to our setup with $d\zeta_3 = d\zeta_2 - d\zeta_1 - 2d\tilde{\zeta}_3$ giving it poles in all the train track parameter points thus symmetrizing the formula.

    A similar object though only for uniform weights was also considered in \cite{kenyon_gradient_2021}.
\end{rem}
\begin{rem}
    Note that by enlarging the fundamental domain, one can always reduce a setup with periodically repeating $\alpha^\pm_i, i \leq n$, $\beta^\pm_j, j\leq m$ to the square setup $n=m$ that we consider here.
\end{rem}

\subsection{The KO Map and Basic Properties}

In \cite{Kenyon_Okounkov_Sheffield_2006} a connection between doubly periodic dimers and spectral curves was established. It was then shown in \cite{Kenyon_Okounkov_2007} that for some doubly periodic weights and certain maximal boundary conditions there is a natural complex structure on the limiting liquid domain of the dimer model. In the case of the Aztec diamond this complex structure is given by a diffeomorphism between $\calR^\circ_+$ and the liquid region. This was described for $2 \times 2$ periodic weights in \cite{Borodin_2x2_2023} and then extended to general periodic weights in \cite{berggren_geometry_2023}. We will introduce this map in our terms in full generality of quasi-periodic weights.

First note that $\tau^* d\zeta = \overline{d\zeta}$ and so all periods of $d\zeta$ are purely imaginary. The following trivial Lemma lets us control the number and location of its zeros.

\begin{lemma}
    \label{lem:zeros_of_dzeta}
    Let $d\xi$ be a meromorphic differential satisfying $\tau^* d\xi = \overline{d\xi}$ with its $4n$ poles at $\cubr{\alpha^\pm_i, \beta^\pm_j} \subset X_0$ and its residues satisfying the alternating sign structure
    \begin{equation}
        \label{eq:dZeta_alternating_residue_sign_structure}
        \res_{\alpha_i^-} d\xi > 0, \; \res_{\beta_j^+} d\xi < 0, \; \res_{\alpha_i^+} d\xi < 0, \; \res_{\beta_j^-} d\xi > 0.
    \end{equation}
    Then it has $2g + 4n - 2$ zeros. Of these, $2g$ lie in pairs on the inner ovals and $4(n - 1)$ lie in between consecutive train track angles of the same type e.g. $\alpha^+_i$ and $\alpha^+_{i+1}$. The last $2$ zeros can be either on a real oval or be a conjugate pair $P, \tau P$, $P \in \calR_+^\circ$.

    If $d\xi$ has the same sign of residue in each of the clusters of train tracks $\cubr{\alpha^-, \beta^-, \alpha^+, \beta^+}$ but does not have the alternating sign structure \eqref{eq:dZeta_alternating_residue_sign_structure}, then all zeros of $df$ are on the real ovals $X$.
\end{lemma}
\begin{proof}
    The existence of zeros on the real ovals follows simply from $d\xi$ being real with vanishing periods. If $d\xi$ does not satisfy \eqref{eq:dZeta_alternating_residue_sign_structure} then there must be at least two neighboring clusters that share a sign. Between each such pair there must a zero of $d\xi$ on $X_0$ thus identifying all zeros on $X$.
\end{proof}

See Proposition~\ref{prop:zeros_of_dZeta_tangential} and Figure~\ref{fig:tangential_zeros} for a geometric intepretation of these zeros applied to $d\zeta$.

Due to Lemma \ref{lem:zeros_of_dzeta} there can be some $(u,v)$ for which $d\zeta$ has a unique pair $(P, \tau P)$ of conjugated zeros in $\calR$. Let $P$ be the zero from that pair in $\calR_+^\circ$. Following \cite{berggren_geometry_2023} we denote the region of such $(u,v)$ as $\mathcal{F}_\mathcal{S}$ and will call this open set the \emph{liquid region}, and also define the map $\Omega: (u,v) \in \mathcal{F}_\mathcal{S} \mapsto P \in \calR_+^\circ$. Note that for $(u,v) \not\in [-1, 1]^2$ the residues of $d\zeta_{(u,v)}$ do not have the alternating sign structure of \eqref{eq:dZeta_alternating_residue_sign_structure} and therefore all its zeros lie on $X$. Therefore $\mathcal{F}_\mathcal{S} \subset [-1, 1]^2$. The map $\Omega$ defined this way is in fact a diffeomorphism. In \cite[Theorem 4.9]{berggren_geometry_2023} this map was introduced on the amoeba as $\tilde{\Omega} : \mathcal{A}_\mathcal{S} \to \mathcal{F}_\mathcal{S}$. In that formulation $\tilde{\Omega}$ is not a diffeomorphism due to mapping an unbounded region to a bounded one.

\begin{proposition}
\label{prop:calF_is_diffeomorphism}
    The map $\Omega: \mathcal{F}_\mathcal{S} \to \calR_+^\circ$ (and hence also its inverse $\mathcal{F} := \Omega^{-1}$) is a diffeomorphism.
\end{proposition}
\begin{proof}
    First, assume there exists $P\in \calR_+^\circ$ such that

    \[ (-u_1 d\zeta_2 + v_1 d\zeta_1 + d\zeta_3)(P) = 0 = (-u_2 d\zeta_2 + v_2 d\zeta_1 + d\zeta_3)(P).\]
Then their difference gives $-u d\zeta_2(P) + v d\zeta_1(P) = 0$. This differential, however, has non-alternating signs of residues at the $\alpha_i^\pm, \beta_i^\pm$ and thus there are two pairs of neighboring clusters of angles with the same sign. There must be a zero between those and therefore this differential has all of its zeros on the real ovals. Hence $\Omega$ is injective.

Now let $P \in \calR_+^\circ$ and let $z$ be a local coordinate at $P$ with $d\zeta_i(z) = f_i(z) dz$, $i \leq 3$. Then $d\zeta(P) = 0$ is equivalent to
    \begin{align*}
        M\begin{pmatrix}
                u \\ v
             \end{pmatrix} := 
         \begin{pmatrix}
                \re f_2 & -\re f_1 \\
                \im f_2 & -\im f_1
             \end{pmatrix}
         \begin{pmatrix}
            u \\ v
         \end{pmatrix}
         = \begin{pmatrix}
             \re f_3 \\ \im f_3
         \end{pmatrix}.
    \end{align*}
    Here all functions are evaluated at $P$. Surjectivity follows from $\det M = \im\br{\overline{f_1} f_2} \neq 0$. 
    Indeed, assume the contrary, i.e. $\frac{f_2}{f_1}(P) \in \mathbb{R}$. The differential
    \[ dF = d\zeta_2 - \frac{f_2}{f_1}(P) d\zeta_1 \]
    clearly vanished at $P$. But we have already seen above that it must have all of its zeros on the real ovals $X$ leading to a contradiction.
\end{proof}

Shifting the focus so that all our fundamental objects are maps from $\calR_+$ we henceforth work with $\mathcal{F} := \Omega^{-1}$. Let us now deduce explicit formulas for $\mathcal{F}$. First note that all singularities of the $d\zeta_i$ lie on the real oval $X_0$. At $\alpha_i^-$ we have $\res_{\alpha^-_i} d\zeta_1 = 1 = \res_{\alpha^-_i} d\zeta_3$ and therefore to compensate the pole we must have $v = -1$. Similar considerations at the other singularities give
\begin{equation}
    \label{eq:uv_at_train_tracks}
    \mathcal{F}(\alpha^-_i) = (u, -1), \; \mathcal{F}(\alpha^+_i) = (u, 1), \; \mathcal{F}(\beta^-_i) = (-1, v), \; \mathcal{F}(\beta^+_i) = (1, v).
\end{equation}

At the singularities our map therefore goes to the boundary of $[-1,1]^2$. In Proposition~\ref{prop:parallelity_prop} we will show that these are in fact tangential points.

On the liquid region $\mathcal{F}_\mathcal{S}$ we defined $\mathcal{F}^{-1}$ as the map to the unique zero of $d\zeta$ in $\calR_+^\circ$. Introducing $R_i = \frac{d\zeta_i}{d\zeta_3}$ and splitting $d\zeta = 0$ into real and imaginary parts we obtain

\begin{align*}
    d\zeta = 0 & \iff v R_1 - u R_2 + 1 = 0 \\
           & \iff \begin{pmatrix}
                    \re R_2 & -\re R_1 \\
                    \im R_2 & -\im R_1
                 \end{pmatrix}
             \begin{pmatrix}
                u \\ v
             \end{pmatrix}
             = \begin{pmatrix}
                 1 \\ 0
             \end{pmatrix}.
\end{align*}

Solving this system we obtain the following

\begin{proposition}
\label{prop:KO-map}
The diffeomorphism $\mathcal{F}: \calR_+^\circ \to \mathcal{F}_\mathcal{S}$ is given by
\begin{equation}
\label{eq:uv_throughR_i}
    u = \frac{\im(R_1)}{\re(R_2) \im(R_1) - \re(R_1) \im(R_2)}, \quad v = \frac{-\im(R_2)}{\re(R_2) \im(R_1) - \re(R_1) \im(R_2)}.
\end{equation}
\end{proposition}

\medskip
\textbf {Expression on the Real Ovals.} On the real ovals both $R_1$ and $R_2$ are real-valued and thus their imaginary parts vanish.
We now compute the limit of \eqref{eq:uv_throughR_i} in this case.

To that end let $z = u + i v$ be a local coordinate arounda real oval such that $\tau z = \bar{z}$, i.e. $v=0$ on the real oval. In this coordinate we have $d\zeta_i = f_i \, dz$ and that $R_i = \frac{d\zeta_i}{d\zeta_3} = \frac{f_i(z)}{f_3(z)}$ are meromorphic functions with real coefficients.

On the holomorphicity domain of $R_i$ on the real ovals we have power series representations with real coefficients 

\[ R_1 = a_0 + a_1 z + \ldots , \quad R_2 = b_0 + b_1 z + \ldots \quad ,a_i, b_i \in \bbR. \]

Therefore $\im(R_1) = a_1 v + \ldots, \; \im(R_2) = b_1 v + \ldots$ and hence $\frac{\im(R_1)}{\im(R_2)} = \frac{a_1}{b_1} = \frac{R_1'}{R_2'}(z)$. 
Since both $R_i$ are real-valued, equations \eqref{eq:uv_throughR_i} imply
\[ u = \frac{R_1'}{R_1' R_2 - R_1 R_2'} ,\quad v = \frac{R_2'}{R_1' R_2 - R_1 R_2'}. \]
The images of the real ovals under $\mathcal{F}$ are called \emph{arctic curves}.
Now plugging in $R_i = \frac{f_i}{f_3}$ and $R_i' = \frac{f_i' f_3 - f_i f_3'}{f_3^2}$ we obtain

\begin{proposition}
\label{prop:arctic}
An arctic curve is parametrized by 
\begin{equation}
    \label{eq:xy_on_real_ovals}
    u = \frac{f_1' f_3 - f_3' f_1}{f_1' f_2 - f_2' f_1} = \frac{\begin{vmatrix}f_1 & f_3 \\ f_1' & f_3' \end{vmatrix}}{\begin{vmatrix}f_1 & f_2 \\ f_1' & f_2' \end{vmatrix}} = \frac{W(f_1, f_3)}{W(f_1, f_2)} , \quad v = \frac{f_2' f_3 - f_3' f_2}{f_1' f_2 - f_2' f_1} = \frac{\begin{vmatrix}f_2 & f_3 \\ f_2' & f_3' \end{vmatrix}}{\begin{vmatrix}f_1 & f_2 \\ f_1' & f_2' \end{vmatrix}} = \frac{W(f_2, f_3)}{W(f_1, f_2)},
\end{equation}
where $d\zeta_i=f_i (z)dz$, and $W$ are the corresponding  Wronskians.
\end{proposition}

The following Lemma shows that this expression is indeed non-singular.

\begin{lemma}
    \label{lem:W(f_1,f_2)_not_zero_on_real_ovals}
    $W(f_1, f_2)(P) \neq 0$ for all $P \in X$.
\end{lemma}
\begin{proof}
    Let $W(f_1, f_2)(P) = 0$ for some $P \in X$. Then there exists $(a,b) \in \mathbb{R}^2$ such that
    \[ \begin{pmatrix} f_1 & f_2 \\ f_1' & f_2' \end{pmatrix}(P)
    \begin{pmatrix} a \\ b \end{pmatrix} = 0 \] 
    and thus $d\tilde{\zeta} = a d\zeta_1 + b d\zeta_2$ has a double zero at $P$. For $i\geq 1$, since $\int_{X_i} d\tilde{\zeta} = 0$, $d\tilde{\zeta}$ must have at least $2$ distinct zeros on $X_i$.

    Furthermore $d\tilde{\zeta}$ has at least $4 (n - 1)$ zeros lying between angle points of the same type as well as at least $2$ zeros between two different angle types since the residue signs are not alternating. All of these zeros are distinct and hence in total $d\tilde{\zeta}$ has at least
    \[ 2 + 2 g + 4 (n - 1) + 2 = 2g + 4n \]
    zeros which is a contradiction to it having $4n$ poles.
\end{proof}

This extends the bijection $\mathcal{F}$ to all of $\calR_+$.

\subsection{The Height Function}

We use the map $\mathcal{F}$ to define what will turn out to be the limiting height function on $[-1, 1]^2$. Let $P \in \calR_+$ and $\mathcal{F}(P) = (u,v) \in \overline{\mathcal{F}_\mathcal{S}}$. Define the complex height function as
\begin{equation}
    \label{eq:def_of_complex_height_function}
    H(P) = H(u,v) := \frac{1}{\pi} \int_{\ell} d\zeta_{(u,v)}.
\end{equation}
Here $\ell \subset \calR_+$ is a path going from some fixed $Q \in X_0$ to $P$. Different choices of $Q$ alter $H$ by a constant. 
Let's denote the real and imaginary parts as
\begin{equation}
    \label{eq:def_h_g}
    H = \frac{1}{\pi} g + ih, \quad  g =  \re\int_{\ell} d\zeta_{(u,v)}, \quad h = \im\frac{1}{\pi } \int_{\ell} d\zeta_{(u,v)}.
\end{equation}

Thus, due to the definition of $P$ through $d\zeta_{(u,v)}(P)=0$ we get for the partial derivatives:
\begin{equation}
\label{eq:g_partial_derivatives}
\begin{split}
    \pdv{g}{u} &= -\re \int_{\ell} d\zeta_2 + \re (d\zeta_{(u,v)}(P)\pdv{P}{u})= -x_2, \\
    \pdv{g}{v} &= x_1
\end{split}
\end{equation}
and
\begin{equation}
\label{eq:h_partial_derivatives}
\pdv{h}{u} = \im \int_{\ell} d\zeta_1 + \re (d\zeta_{(u,v)}(P)\pdv{P}{v})= -\frac{1}{\pi} y_2, \quad
    \pdv{h}{v} = \frac{1}{\pi} y_1.
\end{equation}
Finally, we have 
\begin{equation}
    \label{eq:grads_of_f_and_g}
    \nabla h = (s_1, s_2), \quad J \nabla g = (x_1, x_2),
\end{equation}
where $J := \begin{pmatrix}0 & 1 \\ -1 & 0\end{pmatrix}$ is the rotation  by $\frac{\pi}{2}$. 

We have constructed $3$ diffeomorphisms in $\Delta, \mathcal{A}$ and $\mathcal{F}$ providing us with the coordinates $(s_1, s_2)$ on $\calR_+^\circ$ and $(x_1, x_2), (u,v)$ on $\calR_+$. See Figure~\ref{fig:all_maps_and_zeros} for a summary of their connections.

\begin{lemma}
\label{lem:extension_of_h}
The function $h$ can be extended to a continuously differentiable function $\hat{h}: [-1, 1]^2 \to \mathbb{R}$ by appropriate affine functions on every connected component of $[-1,1]^2 \setminus \mathcal{F}_\mathcal{S}$.
\end{lemma}
\begin{proof}
    We know that $\Delta$ maps the real ovals $X_1, \ldots, X_g$ to distinct points inside $\Delta_\mathcal{S}$. The function $h$ then has a constant gradient on $\mathcal{F}(X_i)$. Take some $(u_0, v_0) \in \mathcal{F}(X_i)$ and consider the affine function $f$ with $\nabla f = \nabla h(u_0, v_0)$ and $f(u_0, v_0) = h(u_0, v_0)$. Since the gradients of $f$ and $h$ match we have $f = h$ on $\mathcal{F}(X_i)$ and the extension is $C^1$.

    Every segment $l$ of $X_0$ that lies between two train track angles is mapped by $\Delta$ to distinct points in $\mathbb{Z}^2 \cap \partial\Delta$. As we have seen in \eqref{eq:uv_at_train_tracks} the endpoints of $\mathcal{F}(l)$ both lie on the boundary of $[-1, 1]^2$ and thus distinguish a separate region $F_l \subset [-1, 1]^2 \setminus \mathcal{F}_\mathcal{S}$ along the boundary of which we have constant $\nabla h$. Thus there is an affine extension again.
\end{proof}

Note that so defined $\hat{h}: [-1, 1]^2 \to \mathbb{R}$ satisfies the Aztec diamond boundary conditions as it satisfies \eqref{eq:aztec_boundary_height}. 

\begin{theorem}
\label{thm:h_is_height_function}
    The function $\hat{h} : [-1, 1]^2 \to \mathbb{R}$ is the limiting dimer height function on the Aztec diamond with the corresponding Fock weights.
\end{theorem}
We postpone the proof to Section~\ref{sec:05_SurfaceTensionMinimization} and utilize the next section to examine some geometric properties of $H$ and $\mathcal{F}$.

\begin{figure}[h]
\centering
\fontsize{10pt}{12pt}\selectfont
\def\svgwidth{\linewidth}
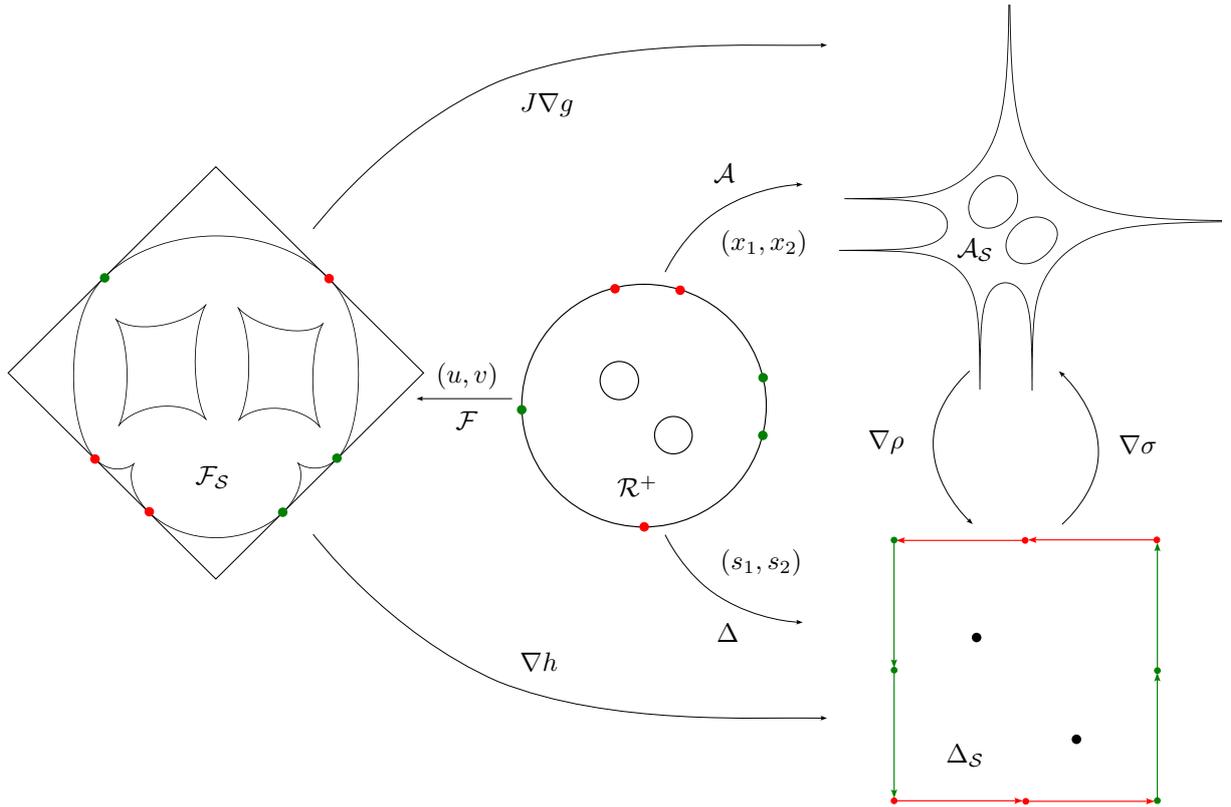
\caption{Harnack data $\mathcal{S} = \cubr{\calR, \cubr{\alpha^\pm, \beta^\pm}}$ along with the three diffeomorphisms $\mathcal{F}, \mathcal{A}, \Delta$ and their compositions given by gradients of the functions $g, h, \rho, \sigma$. Angles $\alpha^\pm$ in green and $\beta^\pm$ in red. This is an example with $m = n = 2$ and some angles being repeated thus leading to poles of order 2. We stress that the data $\mathcal{S}$ depicted here is the actual Riemann surface used, that is all pictures here are the result of computation and not just abstract representations. See Section~\ref{sec:10_SchottkyUniformization} for methods of computation. }
\label{fig:all_maps_and_zeros}
\end{figure}

\section{Geometric Properties}
\label{sec:04_GeometricProperties}


\subsection{Parallelity and Regions}
\label{sec:regions}

We examine geometric properties of the liquid region and arctic curves. We begin with the parallelity property which was already observed in \cite{berggren_geometry_2023}.

\begin{proposition}[Parallelity property]
    \label{prop:parallelity_prop}
    For any $P$ on a real oval the line through $\mathcal{A}(P)$ tangent to $\partial\mathcal{A}_\mathcal{S}$ is parallel to the line through $\mathcal{F}(P)$ tangent to $\partial\mathcal{F}_\mathcal{S}$. 
\end{proposition} 
\begin{proof}
    Henceforth all functions are evaluated at $P$. Recall the formulas derived in \eqref{eq:xy_on_real_ovals} and note that for 
    \[ W_{i,j} := W(f_i, f_j) = \begin{vmatrix}f_i & f_j \\ f_i' & f_j' \end{vmatrix} \]
    we have
    \[ W_{i,j}' = \begin{vmatrix}f_i & f_j \\ f_i'' & f_j'' \end{vmatrix}. \]
 By direct computation we obtain
    \begin{align*}
        u' &= \frac{1}{W_{1,2}^2} \br{W_{1,3}' W_{1,2} - W_{1,2}' W_{1,3}} \\
       &= \frac{1}{W_{1,2}^2} f_1 \br{f_3''\br{f_1 f_2' - f_1' f_2} - f_3' \br{f_1 f_2'' - f_1'' f_2} + f_3 \br{f_1' f_2'' - f_1'' f_2'}} \\
        &= f_1 \frac{\begin{vmatrix}
            f_1 & f_2 & f_3 \\
            f_1' & f_2' & f_3' \\
            f_1'' & f_2'' & f_3''
        \end{vmatrix}}{\begin{vmatrix}
            f_1 & f_2 \\
            f_1' & f_2'
        \end{vmatrix}^2} = f_1 \frac{W(f_1, f_2, f_3)}{W(f_1, f_2)^2} =: f_1 A,
    \end{align*}
    An equivalent calculation for $v$ yields $v' = f_2 A$.  Thus the tangent line has the direction $(f_1, f_2)$ which is also the direction of the tangent of the amoeba at the corresponding point.
\end{proof}

Remarkably, the zeros of $d\zeta$ are also characterized by tangent lines.

\begin{proposition}
\label{prop:zeros_of_dZeta_tangential}
    Let $(u,v) \in [-1,1]^2$ and $(u_1, v_1)$ on the arctic curve such that the tangent curve through $(u_1, v_1)$ goes through $(u,v)$. Then $P_1 = \mathcal{F}^{-1}(u_1, v_1)$ is a zero of $d\zeta_{(u,v)}$. We call such zeros tangential, see Figure~\ref{fig:tangential_zeros}.
\end{proposition}
\begin{proof}
    We will evaluate all functions at $P_1$. Due to Proposition \ref{prop:parallelity_prop} the condition that $(u, v)$ lies on the tangent line through $(u_1, v_1)$ is
    \[ \frac{v - v_1}{u - u_1} = \frac{f_2}{f_1}. \]
    This is equivalent to
    \begin{align*}
        f_1 \br{W_{1,2} v - W_{2,3}} = f_2 \br{W_{1,2} u - W_{1,3}} 
        \iff W_{1,2} \br{f_1 v - f_2 u} = f_1 W_{2,3} - f_2 W_{1,3}.
    \end{align*}
    But the right hand side is equal to 
    \begin{align*}
        \begin{vmatrix} f_1 f_2 & f_3 \\ f_1 f_2' & f_3' \end{vmatrix} - \begin{vmatrix} f_1 f_2 & f_3 \\ f_1' f_2 & f_3' \end{vmatrix} = 
        \begin{vmatrix} 0 & f_3 \\ f_1 f_2' - f_1' f_2  & f_3' \end{vmatrix} =
       - f_3 W_{1,2}.
    \end{align*}
    Putting this together we get 
    \[ - f_2 u + f_1 v + f_3 = 0. \]
\end{proof}

\begin{figure}[h]
\centering
\begin{subfigure}{.4\textwidth}
    \centering
    \fontsize{10pt}{12pt}\selectfont
    \def\svgwidth{\linewidth}
    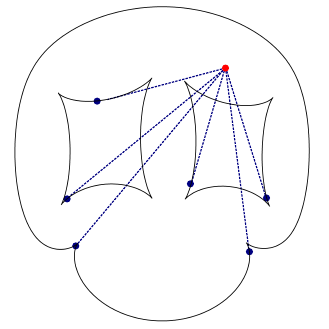
\end{subfigure}
\begin{subfigure}{.4\textwidth}
    \centering
    \fontsize{10pt}{12pt}\selectfont
    \def\svgwidth{\linewidth}
    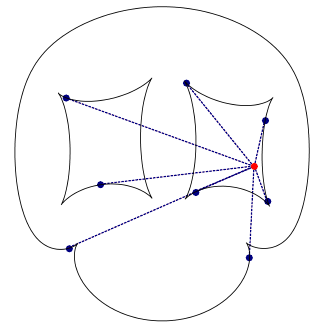
\end{subfigure}
\caption{All the zeros of the differential $d\zeta_{(u,v)}$. \textbf{Left}: For $(u,v) \in \mathcal{F}_\mathcal{S}$ we have pair of conjugated zeros corresponding to $(u,v)$ (red) and all other tangential zeros (blue). \textbf{Right}: For $(u,v) \not\in \mathcal{F}_\mathcal{S}$ all zeros are tangential with two extra zeros on the real oval corresponding to the zone $(u,v)$ is in.} 
\label{fig:tangential_zeros}
\end{figure}

This together with Proposition~\ref{prop:parallelity_prop} gives us a full characterization of the zeros of $d\zeta$.
From \eqref{eq:uv_at_train_tracks} we see that the orientation of $X_0$ is reversed between $\mathcal{A}$ and $\mathcal{F}$. Since both maps are diffeomorphisms, it follows that also the orientations of the $X_i$ must be reversed. All components of $\mathbb{R}^2 \setminus \mathcal{A}_\mathcal{S}$ are convex. Due to the reversed orientation it follows that $\mathcal{F}_\mathcal{S}$ must be locally convex except at cusps in the curves $\mathcal{F}(X_i)$ that correspond to zeros of $A$. Let us now consider the components of $[-1, 1]^2 \setminus \overline{\mathcal{F}_\mathcal{S}}$.

First, let $\ell_{(\alpha_i^\pm, \alpha_{i+1}^\pm)}$ be the arc of $X_0$ between $\alpha_i^\pm$ and $\alpha_{i+1}^\pm$ not containing any other labeled points. Only the endpoints of $\mathcal{F}(\ell_{(\alpha_i^\pm, \alpha_{i+1}^\pm)})$ satisfy $v = \pm 1$. This arc therefore separates a connected region $Q \subset [-1,1]^2\setminus \mathcal{F}_\mathcal{S}$. Since the curve is horizontal at both $\alpha_i^\pm$ and $\alpha_{i+1}^\pm$ and it goes through a $\pi$ rotation as seen from the amoeba, there must be at least one cusp along it. The number of tangential points in a region like this with $k$ cusps is $k+2$. Furthermore for any $(u,v) \in Q$ there are exactly $3$ zeros of $d\zeta_{(u,v)}$ on $\ell_{(\alpha_i^\pm, \alpha_{i+1}^\pm)}$. Thus according to Proposition \ref{prop:zeros_of_dZeta_tangential} there must be exactly one cusp and all zeros are tangential. The same argument holds for any arc $\ell_{(\beta_i^\pm, \beta_{i+1}^\pm)}$. Following \cite{astala_dimer_2023} we call regions of this type \emph{quasi-frozen} and denote the collection of them as $\mathscr{Q}$.

For $\ell_{(\alpha_m^\pm, \beta_1^\mp)}$ or $\ell_{(\beta_n^\pm, \alpha_1^\pm)}$ the two endpoints get mapped by $\mathcal{F}$ to a horizontal and a vertical tangential boundary point respectively. In this case for any $(u,v)$ in the region $F$ that is separated by this arc $d\zeta$ has exactly $2$ zeros on $\ell_{(\alpha_m^\pm, \beta_1^\mp)}$ and thus no cusps are present. We call regions like this \emph{frozen} and denote the set of all frozen regions as $\mathscr{F}$.

The image of $X_i$ under $\mathcal{F}$ for $i \neq 0$ bounds a topological disk $G_i$ in $[-1, 1]^2 \setminus \mathcal{F}_\mathcal{S}$. Under reversal of orientation and following Proposition \ref{prop:parallelity_prop} there must be at least $4$ cusps in $\mathcal{F}(X_i)$. For any $(u,v) \in G_i$ there are $4$ zeros of $d\zeta$ on $X_i$. Therefore there are exactly $4$ cusps and again all zeros are tangential. We call the $G_i$ \emph{gas regions} and denote their collection as $\mathscr{G}$.

For a frozen region $F_{(\alpha, \beta)} = F \in \mathscr{F}$ or quasi-frozen region $Q_{(\alpha, \beta)} = Q \in \mathscr{Q}$ corresponding to $\ell_{(\alpha, \beta)}$ for some neighboring points $\alpha, \beta$ on $X_0$ we will denote the corresponding unbounded component of $\mathcal{A}_\mathcal{S}^c$ as $\mathcal{A}_{(\alpha, \beta)}$ or $\mathcal{A}_Q$ and the corresponding point on the boundary of the Newton polygon as $\Delta_{(\alpha, \beta)}$ or $\Delta_Q$. Similarly for the gas bubble $G_i$ we denote the corresponding bubble in the amoeba as $\mathcal{A}_i$ and the point in the inside of the Newton polygon as $\Delta_i$.

We obtain a decomposition of the Aztec diamond into its liquid, frozen, quasi-frozen and gas regions
\begin{equation}
    \label{eq:decomposition_into_frozen_and_gas_regions}
    [-1, 1]^2 = \overline{\mathcal{F}}_\mathcal{S} \cup \bigcup_{F \in \mathscr{F}} \overline{F} \cup \bigcup_{Q \in \mathscr{Q}} \overline{Q} \cup \bigcup_{G \in \mathscr{G}} \overline{G}.
\end{equation}
Note that  all regions are open sets, and intersections of their closures are the arctic curves, see Figure~\ref{fig:all_zones_notation} for an illustration of the introduced notations.

\begin{figure}[h]
\centering
\begin{subfigure}{.49\textwidth}
    \centering
    \fontsize{10pt}{12pt}\selectfont
    \def\svgwidth{\linewidth}
    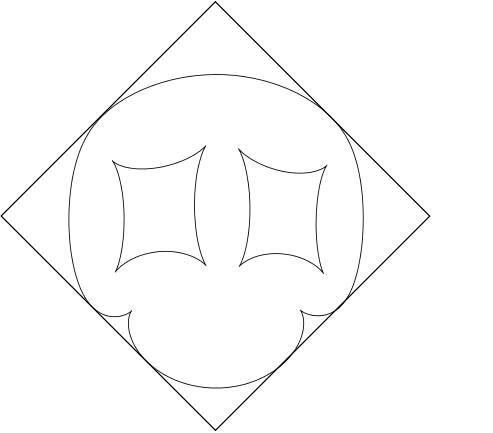
\end{subfigure}
\begin{subfigure}{.49\textwidth}
    \centering
    \fontsize{10pt}{12pt}\selectfont
    \def\svgwidth{\linewidth}
    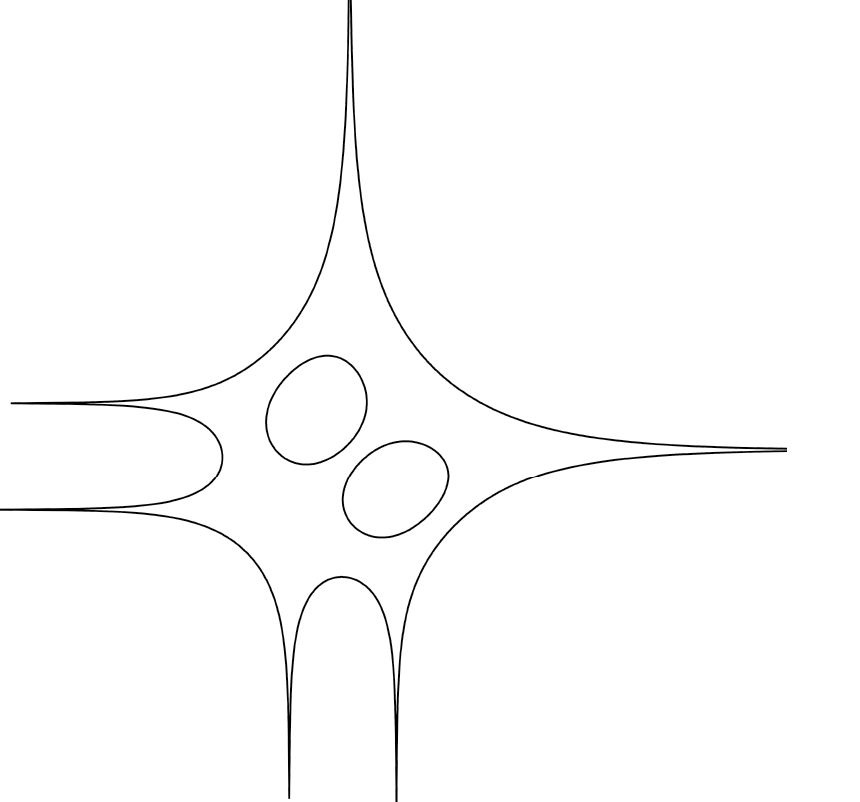
\end{subfigure}
\caption{All zones of the dimer (left) and the amoeba (right).} 
\label{fig:all_zones_notation}
\end{figure}

The following corollary gives a full geometric description of the zeros of $d\zeta$. 
\begin{corollary}
\label{cor:classification_of_all_dZeta_zeros}
    For all $(u,v) \in [-1, 1]^2$ the differential $d\zeta_{(u,v)}$ has two tangential zeros (see Proposition~\ref{prop:zeros_of_dZeta_tangential}) on every $X_i, i=1,\ldots ,g$ and one tangential zero on every arc corresponding to a quasi-frozen region. Additionally to this 
    \begin{itemize}
        \item for $(u,v) \in \mathcal{F}_\mathcal{S}$ the differential $d\zeta$ has a zero at $P:=\mathcal{F}^{-1}(u,v)\in \mathcal{R}^\circ_+$ and its conjugate $\tau P$.
        \item for $(u,v) \in [-1, 1]^2 \setminus \overline{\mathcal{F}_\mathcal{S}}$ the differential $d\zeta$ has two distinct tangential zeros on the boundary of the corresponding connected component of $[-1, 1]^2 \setminus \overline{\mathcal{F}_\mathcal{S}}$.
        \item For $(u,v)$ on one of the arctic curves $\partial \mathcal{F}_\mathcal{S}$ the differential $d\zeta$ has a double zero on the component of $X$ corresponding to $(u,v)$.
    \end{itemize}
\end{corollary}
See Figure \ref{fig:frozen_regions_tangential_zeros} for an illustration of all $3$ types of zones.

\subsection{Divergence and Burgers equations}

The main reason for defining $d\zeta$ the way we did in \eqref{eq:def_dZeta} is that then $(x_1, x_2)$ and $(y_1, y_2)$ define divergence free fields with respect to $(u,v)$.

\begin{proposition}
    \label{prop:x_y_divergence_free}
    Let $(u,v) \in \mathcal{F}_\mathcal{S}$ and $(x_1, x_2) = \mathcal{A} \circ \mathcal{F}^{-1} (u,v)$ its coordinates in the amoeba. Then
    \begin{equation}
    \label{eq:x_divergence_free}
        \Div_{(u,v)} \br{ x_1, x_2 }:=\pdv{x_1}{u}+\pdv{x_2}{v} = 0.
    \end{equation}
    Similarly for the coordinates $(y_1, y_2)$ we have
    \begin{equation}
    \label{eq:s_divergence_free}
        \Div_{(u,v)} \br{ y_1, y_2 }:= \pdv{y_1}{u}+\pdv{y_2}{v} = 0.
    \end{equation}
\end{proposition}
\begin{proof}
    Let $P(u,v) = \mathcal{F}^{-1}(u,v)$. Differentiating $d\zeta(P(u,v))=0$ we obtain
    \begin{equation}
    \label{eq:diff_dZeta_wrt_u_v}
        d\br{d\zeta(P)} \pdv{P}{u} - d\zeta_2(P) = 0, \quad d\br{d\zeta(P)} \pdv{P}{v} + d\zeta_1(P) = 0.
    \end{equation}
    We have $d\br{d\zeta(P)}\neq 0$ since $P$ is a simple zero. From \eqref{eq:diff_dZeta_wrt_u_v} we get
    \begin{equation}
        \label{eq:equivalent_to_burgers}
       d\zeta_1(P) \pdv{P}{u} + d\zeta_2(P) \pdv{P}{v} = 0.
    \end{equation}
    Finally this implies
    \begin{align*}
        \pdv{x_1}{u} + \pdv{x_2}{v} 
                                    = \re \br{d\zeta_1(P)\pdv{P}{u} + d\zeta_2(P)\pdv{P}{v} } = 0.
    \end{align*}
    The argument for $(y_1, y_2)$ is the same.
\end{proof}
\begin{rem}
    Equation \eqref{eq:equivalent_to_burgers} is equivalent to the Burgers equation given by
    \begin{equation}
        \label{eq:Burgers_equation}
        \frac{\frac{\partial P}{\partial u}}{\frac{\partial P}{\partial v}} = -\frac{d\zeta_2(P)}{d\zeta_1(P)}.
    \end{equation}
    This was introduced and studied in \cite{Kenyon_Okounkov_2007}.
\end{rem}

\begin{corollary}
    \label{cor:Euler_lagrange_h_and_g}
    Let $(u,v) \in \mathcal{F}_\mathcal{S}$. Then $h$ satisfies
    \begin{equation}
    \label{eq:h_divergence_free}
        \Div \br{ \nabla\sigma (\nabla h(u,v)) } = 0.
    \end{equation} 
    Similarly, $g$ satisfies
    \begin{equation}
    \label{eq:g_divergence_free}
        \Div \br{ J\nabla\rho (J\nabla g(u,v)) } = 0.
    \end{equation}
\end{corollary}
\begin{proof}
The derivatives are taken in the variables $\Div_{(u,v)}\br{\nabla\sigma_{(s_1, s_2)} (\nabla_{(u,v)} h)}$.
Formulas \eqref{eq:h_divergence_free} and \eqref{eq:g_divergence_free} follow directly from Proposition~\ref{prop:x_y_divergence_free} and \eqref{eq:g_partial_derivatives}, \eqref{eq:h_partial_derivatives}. 
\end{proof}

Remarkably the divergence free property Proposition~\ref{eq:x_divergence_free} extends to the real ovals via the tangential zeros from Proposition \ref{prop:zeros_of_dZeta_tangential}.

\begin{proposition}
    \label{prop:tangential_zeros_divergence_free}
    For $z = (u,v) \in [-1,1]^2$ let $z_1 = (u_1, v_1) \in \partial \mathcal{F}_\mathcal{S}$ be on the arctic curve such that the tangent curve through $(u_1, v_1)$ goes through $(u,v)$. Let $P_1 = P_1(u,v)$ be the corresponding point on one of the real ovals continuously depending on $(u,v)$. Then
    \[ \Div_{(u,v)} \mathcal{A}(P_1) = 0. \]
\end{proposition}
\begin{proof}
Let $\tilde{z}$ be a variation of $z$ and $z_2$ the tangency point of the line through $\tilde{z}$, see Fig.~\ref{fig:tangential_points_divergence_free}. Denote by $w_i = \mathcal{A}\circ\mathcal{F}^{-1}(z_i)$ the points on the amoebas boundary corresponding to $z_i$.  In the limit $\tilde{z} \to z$ we have $z_2 \to z_1$ and $w_2 \to w_1$ with $w_2 - w_1$ tangent to the amoeba at $w_1$. 

Choose the orthonormal basis $(e_1,e_2)$ such that $e_1 || (z_1 - z)$ and parametrize $w_2-w_1 = t e_1 + n e_2$. Due to the parallelity property Proposition~\ref{prop:parallelity_prop} $w_2-w_1$ infinitesimally varies in the direction of $e_1$ only, thus both directional derivatives of $n$ vanish $\partial_{e_1}n=\partial_{e_2}n=0$. On the other hand, variation of $\tilde{z}$ in the direction $e_1$ does not change the tangency point $z_1$, and thus $\partial_{e_1}t=\partial_{e_1}n=0$. Finally, this yields
    \[ \Div_{(u,v)}\mathcal{A}(P_1) = \partial_{e_1} t + \partial_{e_2} n = 0. \]

%
%
\end{proof}

\begin{figure}[h]
\centering
\begin{subfigure}{.4\textwidth}
    \centering
    \fontsize{10pt}{12pt}\selectfont
    \def\svgwidth{\linewidth}
\begingroup%
  \makeatletter%
  \providecommand\color[2][]{%
    \errmessage{(Inkscape) Color is used for the text in Inkscape, but the package 'color.sty' is not loaded}%
    \renewcommand\color[2][]{}%
  }%
  \providecommand\transparent[1]{%
    \errmessage{(Inkscape) Transparency is used (non-zero) for the text in Inkscape, but the package 'transparent.sty' is not loaded}%
    \renewcommand\transparent[1]{}%
  }%
  \providecommand\rotatebox[2]{#2}%
  \newcommand*\fsize{\dimexpr\f@size pt\relax}%
  \newcommand*\lineheight[1]{\fontsize{\fsize}{#1\fsize}\selectfont}%
  \ifx\svgwidth\undefined%
    \setlength{\unitlength}{168.08984711bp}%
    \ifx\svgscale\undefined%
      \relax%
    \else%
      \setlength{\unitlength}{\unitlength * \real{\svgscale}}%
    \fi%
  \else%
    \setlength{\unitlength}{\svgwidth}%
  \fi%
  \global\let\svgwidth\undefined%
  \global\let\svgscale\undefined%
  \makeatother%
  \begin{picture}(1,0.92320661)%
    \lineheight{1}%
    \setlength\tabcolsep{0pt}%
    \put(0,0){\includegraphics[width=\unitlength,page=1]{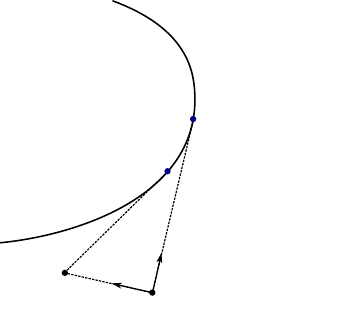}}%
    \put(0.42477118,0.00943805){\color[rgb]{0,0,0}\makebox(0,0)[lt]{\lineheight{1.25}\smash{\begin{tabular}[t]{l}$z$\end{tabular}}}}%
    \put(0.13693287,0.06504874){\color[rgb]{0,0,0}\makebox(0,0)[lt]{\lineheight{1.25}\smash{\begin{tabular}[t]{l}$\tilde{z}$\end{tabular}}}}%
    \put(0.45861278,0.5878233){\color[rgb]{0,0,0.50196078}\makebox(0,0)[lt]{\lineheight{1.25}\smash{\begin{tabular}[t]{l}$z_1$\end{tabular}}}}%
    \put(0.39092944,0.45040549){\color[rgb]{0,0,0.50196078}\makebox(0,0)[lt]{\lineheight{1.25}\smash{\begin{tabular}[t]{l}$z_2$\end{tabular}}}}%
    \put(0.39503141,0.17351895){\color[rgb]{0,0,0}\makebox(0,0)[lt]{\lineheight{1.25}\smash{\begin{tabular}[t]{l}$e_1$\end{tabular}}}}%
    \put(0.34478178,0.12737116){\color[rgb]{0,0,0}\makebox(0,0)[lt]{\lineheight{1.25}\smash{\begin{tabular}[t]{l}$e_2$\end{tabular}}}}%
    \put(0.22172101,0.71498598){\color[rgb]{0,0,0}\makebox(0,0)[lt]{\lineheight{1.25}\smash{\begin{tabular}[t]{l}$\mathcal{F}_\mathcal{S}$\end{tabular}}}}%
  \end{picture}%
\endgroup%

\end{subfigure}
\begin{subfigure}{.4\textwidth}
    \centering
    \fontsize{10pt}{12pt}\selectfont
    \def\svgwidth{\linewidth}
\begingroup%
  \makeatletter%
  \providecommand\color[2][]{%
    \errmessage{(Inkscape) Color is used for the text in Inkscape, but the package 'color.sty' is not loaded}%
    \renewcommand\color[2][]{}%
  }%
  \providecommand\transparent[1]{%
    \errmessage{(Inkscape) Transparency is used (non-zero) for the text in Inkscape, but the package 'transparent.sty' is not loaded}%
    \renewcommand\transparent[1]{}%
  }%
  \providecommand\rotatebox[2]{#2}%
  \newcommand*\fsize{\dimexpr\f@size pt\relax}%
  \newcommand*\lineheight[1]{\fontsize{\fsize}{#1\fsize}\selectfont}%
  \ifx\svgwidth\undefined%
    \setlength{\unitlength}{131.11959863bp}%
    \ifx\svgscale\undefined%
      \relax%
    \else%
      \setlength{\unitlength}{\unitlength * \real{\svgscale}}%
    \fi%
  \else%
    \setlength{\unitlength}{\svgwidth}%
  \fi%
  \global\let\svgwidth\undefined%
  \global\let\svgscale\undefined%
  \makeatother%
  \begin{picture}(1,0.7014819)%
    \lineheight{1}%
    \setlength\tabcolsep{0pt}%
    \put(0,0){\includegraphics[width=\unitlength,page=1]{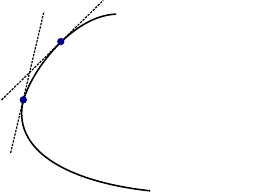}}%
    \put(-0.00519486,0.52606068){\color[rgb]{0,0,0}\makebox(0,0)[lt]{\lineheight{1.25}\smash{\begin{tabular}[t]{l}$\mathcal{A}_\mathcal{S}$\end{tabular}}}}%
    \put(0.13318235,0.27978421){\color[rgb]{0,0,0.50196078}\makebox(0,0)[lt]{\lineheight{1.25}\smash{\begin{tabular}[t]{l}$w_1$\end{tabular}}}}%
    \put(0.27254853,0.50762736){\color[rgb]{0,0,0.50196078}\makebox(0,0)[lt]{\lineheight{1.25}\smash{\begin{tabular}[t]{l}$w_2$\end{tabular}}}}%
  \end{picture}%
\endgroup%

\end{subfigure}
\caption{Point $z$ and its perturbation $\tilde{z}$ along with corresponding tangential zeros $z_1, z_2$ and their amoeba maps $w_1, w_2$.} 
\label{fig:tangential_points_divergence_free}
\end{figure}
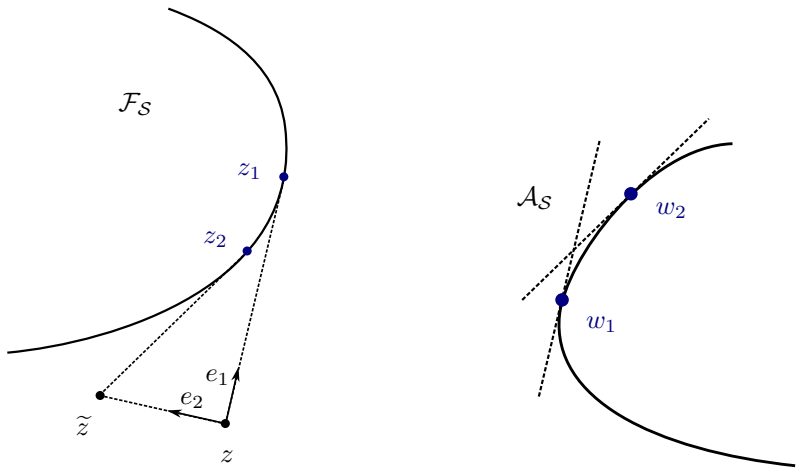

Integrating this vector field then defines a function with gradients living on the real ovals of the amoeba. This is a linear extension along any tangential line.
\begin{corollary}
    \label{cor:tangential_extension_gradient}
    Let $z = (u,v) \in [-1, 1]^2$ and $z_1 \in \partial \mathcal{F}_\mathcal{S}$ such that the tangential line to the arctic curve through $z_1$ goes through $z$ with $P_1 = \mathcal{F}^{-1}(z_1)$. Here $z_1$ changes continuously with $z$. The function 
    \[ \tilde{g}(z) := g(z_1) + \pobr{\nabla g(z_1), z - z_1} \] 
    then satisfies
    \[ J\nabla \tilde{g}(z) = J \nabla g(z_1) = \mathcal{A}(P_1). \]
\end{corollary}

\section{Proof of Surface Tension Minimization}
\label{sec:05_SurfaceTensionMinimization}

This section is dedicated to the proof of Theorem~\ref{thm:h_is_height_function}. That is we show that the function $\hat{h}$ constructed in Section~\ref{sec:03_CoordinateMap} is the global surface tension minimizer for the Aztec diamond boundary conditions. 

In \cite{Kenyon_Okounkov_2007} solutions to this minimization problem for the hexagonal lattice were constructed. The proof utilizes algebraic methods and is unfortunately incomplete. It works only for the spectral curve $P(z,w) = z + w + 1$ which corresponds to uniform weights, see \cite{astala_dimer_2023} for a discussion of the limitations. Limit shapes for the Aztec diamond with doubly periodic weights were fully described in \cite{berggren_geometry_2023}. The method of proof here relies on a bijection of the dimer model to the non-intersecting paths model. It is unclear how to extend this approach to other graphs. In \cite{astala_dimer_2023} a purely variational approach is taken and sufficient conditions are derived for a function to be the minimizer for a general set of boundary conditions. The argument here works only for the genus 0 case.

We work in full generality of quasi-periodic weights and utilize similar techniques to \cite{astala_dimer_2023} extending the construction to arbitrary genus. The key ingredient for our argument is an extension of $g$ in \eqref{eq:def_h_g} to the gas bubbles and frozen zones utilizing Corollary \ref{prop:tangential_zeros_divergence_free}. We stress that our proof is purely variational so does not utilize any special knowledge about the underlying graph. In particular no translations to other models are needed and no information about the inverse Kasteleyn matrix is required. This allows for flexibility as our construction can then be directly extended to the hexagon, see Section~\ref{sec:07_hexagonal_case}, and possibly other cases.

We start the proof of Theorem~\ref{thm:h_is_height_function} by recalling that uniqueness of the minimizer of \eqref{eq:def_surface_tension_functional} follows from strict convexity of $\sigma$ on $\Delta^\circ_\mathcal{S}$, see \cite[Proposition~4.5]{DeSilva_minimizers_of_convex_functionals_2010}. The Euler-Lagrange equation for \eqref{eq:def_surface_tension_functional} is
\begin{equation}
    \label{eq:surface_tension_Euler-lagrange_smooth}
    \Div(\nabla\sigma(\nabla f)) = 0.
\end{equation}
Here the differentiation in \eqref{eq:surface_tension_Euler-lagrange_smooth} is to be understood with respect to the coordinates 
\[\Div_{(u,v)}(\nabla_{(s_1, s_2)}\sigma(\nabla_{(u,v)} f)) = 0.\]

Since the surface tension $\sigma$ in dimer problems is not differentiable at singularities corresponding to gas and frozen points, this does not give a general condition for a function to be the unique minimizer. The minimizer can be characterized in terms of Gâteaux derivatives. 

Recall that the Gâteaux derivative of the functional $I_\sigma$ in the direction $u-f$ is given by
\[dI_\sigma[f; \; u-f] := \int_\Omega d\sigma\br{ \nabla f ;\; \nabla u - \nabla f }, \]
where for $p, p_0 \in \Delta$,
\[ d\sigma(p_0;\; p-p_0) := \lim_{t \to 0^+} \frac{\sigma\br{(1-t)p_0 + tp} - \sigma(p_0)}{t} \]
is a one-sided directional derivative of $\sigma$ at $p_0$ in the direction $p-p_0$.

\begin{lemma}
    \label{lem:characterization_of_minimizer_wrt_Gateaux_der}
    A function $f \in \Lip_{\Delta_\mathcal{S}}(\Omega, f_b)$ is the minimizer of $I_\sigma$ if and only if
    \[ dI_\sigma[f; \; u-f]  \geq 0  \quad \forall u \in \Lip_{\Delta_\mathcal{S}}(\Omega, f_b). \]
\end{lemma}
\begin{proof}
    By convexity of $\Lip_{\Delta_\mathcal{S}}(\Omega, f_b)$ and $\sigma$ we have that for any $t \in [0, 1]$ the function $f + t(u - f)$ is in $\Lip_{\Delta_\mathcal{S}}(\Omega, f_b)$ and
    \[ I_\sigma(u) \geq I_\sigma(f) + \frac{I_\sigma(f + t(u-f)) - I_\sigma(f)}{t}. \]
    Taking the limit $t \to 0^+$ then yields
    \[ I_\sigma(u) \geq I_\sigma(f) + dI_\sigma[f; u-f], \]
    and the claim follows.
\end{proof}

From convexity of $\sigma$ it is clear that
\begin{equation}
    \label{eq:Gateaux_to_subgradient}
    d\sigma(p_0;\; p-p_0) \geq \innerprod{p^*}{p-p_0} \quad \forall p^* \in \partial\sigma(p_0)
\end{equation}
where $\partial\sigma$ is the subgradient of $\sigma$. Due to the Legendre duality \eqref{eq:rho_sigma_Legendre_via_gradients} the subgradient of $\sigma$ is well understood and is given by
\begin{equation}
\label{eq:sigma_subgradients_in_amoeba}
    \partial \sigma(p_0) = \left\{ \begin{array}{lr}
        (x_1, x_2) & \text{for } p_0 = (s_1, s_2) \in \Delta^\circ_\mathcal{S} \\
        \mathcal{A}_{(\alpha, \beta)} & \text{for } p_0 = \Delta_{(\alpha, \beta)} \\
        \mathcal{A}_i & \text{for } p_0 = \Delta_i .
    \end{array} \right.
\end{equation}
That is, for the conical singularities at the points corresponding to the frozen and gas phases, the subgradient is given by the corresponding connected component of $\mathbb{R}^2 \setminus \mathcal{A}_\mathcal{S}$. In $\Delta^\circ_\mathcal{S}$ the subgradient agrees with the gradient.

We apply Lemma~\ref{lem:characterization_of_minimizer_wrt_Gateaux_der} to $\hat{h}$. The following Proposition gives a sufficient condition for a function to be the global surface tension minimizer, see \cite[Proposition 8.1]{astala_dimer_2023}. This is a generalization of the Euler-Lagrange equation \eqref{eq:surface_tension_Euler-lagrange_smooth} in terms of subgradients.

\begin{proposition}
    \label{prop:generalized_Euler_Lagrange_suff_cond}
    The function $\hat{h} \in \Lip_{\Delta_\mathcal{S}}(\Omega, h_b)$ is the surface tension minimizer if there exists a continuous $\hat{g} : \Omega \to \mathbb{R}$ for which $\nabla \hat{g}$ exists almost everywhere with
    \begin{itemize}
        \item $J\nabla \hat{g}(z) \in \partial \sigma(\nabla \hat{h}(z))$ for almost all $z \in \Omega.$
        \item $\nabla \hat{g} \in L^2(\Omega)$.
    \end{itemize}
\end{proposition}
\begin{proof}
    Let $u \in \Lip_\Delta(\Omega, h_b)$. Using \eqref{eq:Gateaux_to_subgradient} we then have for $\varphi \in C^\infty_c(\Omega)$ that
    \begin{align*}
        dI_\sigma[\hat{h}; \; u-\hat{h}] &\geq \int_\Omega \pobr{J\nabla \hat{g}, \nabla u - \nabla \hat{h}} \\
                                  & = \int_\Omega \Big< J\nabla \hat{g}, \nabla \varphi \Big> + \int_\Omega \pobr{J\nabla \hat{g}, \nabla u - \nabla \hat{h} - \nabla \varphi}.
    \end{align*}
    Since $\varphi$ is smooth, its mixed second derivaties agree and we obtain for the first summand 
    \[\int_\Omega \pobr{J\nabla \hat{g}, \nabla \varphi} = \int_\Omega \pobr{\hat{g}, \Div J^{-1} \nabla \varphi} = 0.\]
    For the second summand we apply the Hölder inequality
    \[ \abs{\int_\Omega \pobr{J\nabla \hat{g}, \nabla u - \nabla \hat{h} - \nabla \varphi}} \leq \Big\Vert J \nabla \hat{g} \Big\Vert_2 \norm{\nabla u - \nabla \hat{h} - \nabla \varphi}_2. \]
    The first factor is finite by assumption. Note that $u - \hat{h}$ is Lipschitz and has zero boundary conditions and therefore is in the Sobolev space $W^{1,2}_0(\Omega) = \overline{C^\infty_c(\Omega)}$. We can thus choose $\varphi$ to be smooth with compact support to approximate $u-f$ in $W^{1,2}$. In particular the second factor goes to zero. Hence
    \[ dI_\sigma[\hat{h}; \; u-\hat{h}] \geq 0 \quad \forall \; u \in \Lip_\Delta(\Omega, h_b), \]
    and Lemma~\ref{lem:characterization_of_minimizer_wrt_Gateaux_der} finishes the proof.
\end{proof}

The strategy to show Theorem~\ref{thm:h_is_height_function} is thus to construct a function $\hat{g}$ satisfying the properties in Proposition~\ref{prop:generalized_Euler_Lagrange_suff_cond}. That is a function $\hat{g} : [-1,1]^2 \to \mathbb{R}$ such that
\begin{itemize}
    \item for $(u,v) \in \mathcal{F}_\mathcal{S}$ we have $J\nabla\hat{g}(u,v) = (x_1, x_2)$.
    \item for $(u,v)$ in a frozen or quasi-frozen region $F$ we have $J\nabla\hat{g}(u,v)$ lies in the corresponding unbounded component of the complement of the amoeba $\mathcal{A}_F$.
    \item for $(u,v)$ in a gas region $G_i$ we have $J\nabla \hat{g}(u,v)$ lies in the corresponding amoeba gas bubble $\mathcal{A}_i$.
\end{itemize}
Due to \eqref{eq:grads_of_f_and_g}, $g$ satisfies the conditions of Proposition~\ref{prop:generalized_Euler_Lagrange_suff_cond} in the liquid region. This motivates the following definition.
\begin{definition}
    \label{def:admissible_extension_and_boundary_condition}
    We call an extension $\hat{g}:[-1,1]^2 \to \mathbb{R}$ of $g:\mathcal{F}_\mathcal{S} \to \mathbb{R}$ \emph{admissible} if it satisfies the conditions of Proposition~\ref{prop:generalized_Euler_Lagrange_suff_cond}.
\end{definition}

We now give a construction that provides an admissible extension of $g$.

\subsection{Extension to Frozen Regions.}
As we have seen in 
Section~\ref{sec:regions}
there are two kinds of frozen regions. Frozen regions in $\mathscr{F}$ correspond to segments $\ell_{(\alpha^\pm_m, \beta^\pm_1)}$ or $\ell_{(\beta^\pm_n, \alpha^\mp_1)}$ on $X_0$ and contain one of the corners of $[-1, 1]^2$. Quasi-frozen regions in $\mathscr{Q}$ correspond to segments $\ell_{(\alpha^\pm_i, \alpha^\pm_{i+1})}$ or $\ell_{(\beta^\pm_j, \beta^\pm_{j+1})}$ on $X_0$ and contain a cusp.

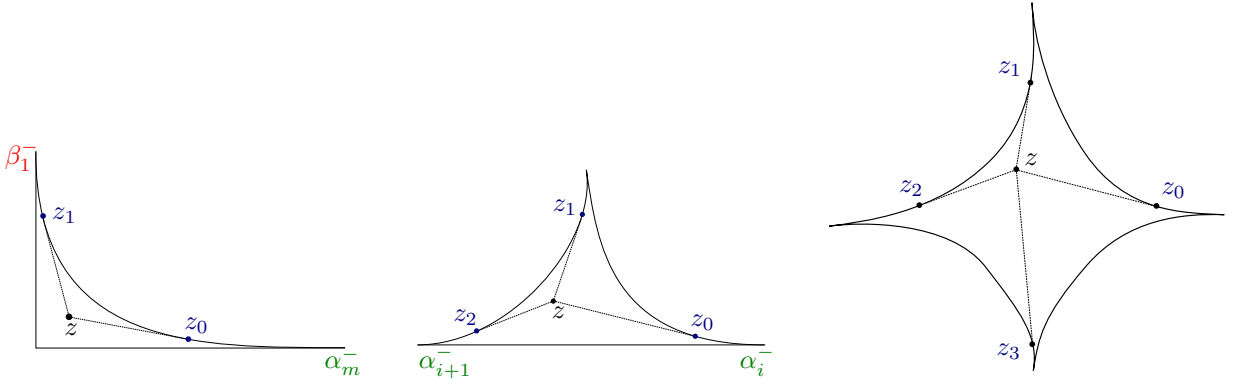
\begin{figure}[h]
\centering
\begin{subfigure}{.32\textwidth}
    \centering
    \fontsize{10pt}{12pt}\selectfont
    \def\svgwidth{\linewidth}
\begingroup%
  \makeatletter%
  \providecommand\color[2][]{%
    \errmessage{(Inkscape) Color is used for the text in Inkscape, but the package 'color.sty' is not loaded}%
    \renewcommand\color[2][]{}%
  }%
  \providecommand\transparent[1]{%
    \errmessage{(Inkscape) Transparency is used (non-zero) for the text in Inkscape, but the package 'transparent.sty' is not loaded}%
    \renewcommand\transparent[1]{}%
  }%
  \providecommand\rotatebox[2]{#2}%
  \newcommand*\fsize{\dimexpr\f@size pt\relax}%
  \newcommand*\lineheight[1]{\fontsize{\fsize}{#1\fsize}\selectfont}%
  \ifx\svgwidth\undefined%
    \setlength{\unitlength}{318.43865582bp}%
    \ifx\svgscale\undefined%
      \relax%
    \else%
      \setlength{\unitlength}{\unitlength * \real{\svgscale}}%
    \fi%
  \else%
    \setlength{\unitlength}{\svgwidth}%
  \fi%
  \global\let\svgwidth\undefined%
  \global\let\svgscale\undefined%
  \makeatother%
  \begin{picture}(1,0.55181944)%
    \lineheight{1}%
    \setlength\tabcolsep{0pt}%
    \put(0,0){\includegraphics[width=\unitlength,page=1]{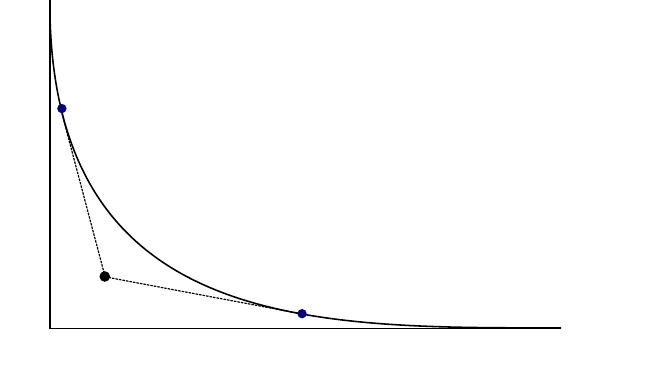}}%
    \put(0.44537085,0.10420585){\color[rgb]{0,0,0.50196078}\makebox(0,0)[lt]{\lineheight{1.25}\smash{\begin{tabular}[t]{l}$z_0$\end{tabular}}}}%
    \put(0.11624854,0.38027892){\color[rgb]{0,0,0.50196078}\makebox(0,0)[lt]{\lineheight{1.25}\smash{\begin{tabular}[t]{l}$z_1$\end{tabular}}}}%
    \put(0.14656242,0.09013157){\color[rgb]{0,0,0}\makebox(0,0)[lt]{\lineheight{1.25}\smash{\begin{tabular}[t]{l}$z$\end{tabular}}}}%
    \put(0.7945696,0.00568569){\color[rgb]{0,0.50196078,0}\makebox(0,0)[lt]{\lineheight{1.25}\smash{\begin{tabular}[t]{l}$\alpha^-_m$\end{tabular}}}}%
    \put(-0.00213905,0.51988218){\color[rgb]{1,0,0}\makebox(0,0)[lt]{\lineheight{1.25}\smash{\begin{tabular}[t]{l}$\beta^-_1$\end{tabular}}}}%
  \end{picture}%
\endgroup%

\end{subfigure}
\begin{subfigure}{.32\textwidth}
    \centering
    \fontsize{10pt}{12pt}\selectfont
    \def\svgwidth{\linewidth}
\begingroup%
  \makeatletter%
  \providecommand\color[2][]{%
    \errmessage{(Inkscape) Color is used for the text in Inkscape, but the package 'color.sty' is not loaded}%
    \renewcommand\color[2][]{}%
  }%
  \providecommand\transparent[1]{%
    \errmessage{(Inkscape) Transparency is used (non-zero) for the text in Inkscape, but the package 'transparent.sty' is not loaded}%
    \renewcommand\transparent[1]{}%
  }%
  \providecommand\rotatebox[2]{#2}%
  \newcommand*\fsize{\dimexpr\f@size pt\relax}%
  \newcommand*\lineheight[1]{\fontsize{\fsize}{#1\fsize}\selectfont}%
  \ifx\svgwidth\undefined%
    \setlength{\unitlength}{301.0596862bp}%
    \ifx\svgscale\undefined%
      \relax%
    \else%
      \setlength{\unitlength}{\unitlength * \real{\svgscale}}%
    \fi%
  \else%
    \setlength{\unitlength}{\svgwidth}%
  \fi%
  \global\let\svgwidth\undefined%
  \global\let\svgscale\undefined%
  \makeatother%
  \begin{picture}(1,0.50415053)%
    \lineheight{1}%
    \setlength\tabcolsep{0pt}%
    \put(0,0){\includegraphics[width=\unitlength,page=1]{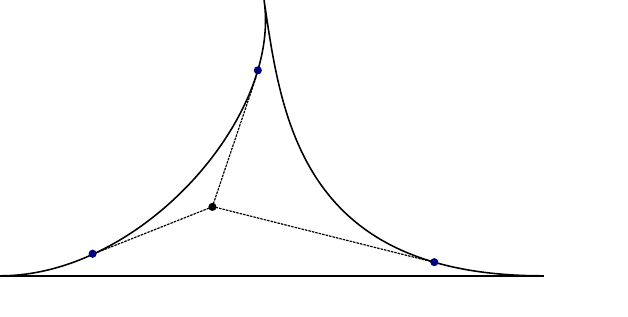}}%
    \put(0.68039029,0.1109646){\color[rgb]{0,0,0.50196078}\makebox(0,0)[lt]{\lineheight{1.25}\smash{\begin{tabular}[t]{l}$z_0$\end{tabular}}}}%
    \put(0.33723768,0.40038921){\color[rgb]{0,0,0.50196078}\makebox(0,0)[lt]{\lineheight{1.25}\smash{\begin{tabular}[t]{l}$z_1$\end{tabular}}}}%
    \put(0.08981893,0.12541766){\color[rgb]{0,0,0.50196078}\makebox(0,0)[lt]{\lineheight{1.25}\smash{\begin{tabular}[t]{l}$z_2$\end{tabular}}}}%
    \put(0.3352707,0.13312204){\color[rgb]{0,0,0}\makebox(0,0)[lt]{\lineheight{1.25}\smash{\begin{tabular}[t]{l}$z$\end{tabular}}}}%
    \put(0.80358442,0.00601391){\color[rgb]{0,0.50196078,0}\makebox(0,0)[lt]{\lineheight{1.25}\smash{\begin{tabular}[t]{l}$\alpha^-_i$\end{tabular}}}}%
    \put(0.00350232,0.00601391){\color[rgb]{0,0.50196078,0}\makebox(0,0)[lt]{\lineheight{1.25}\smash{\begin{tabular}[t]{l}$\alpha^-_{i+1}$\end{tabular}}}}%
  \end{picture}%
\endgroup%

\end{subfigure}
\begin{subfigure}{.32\textwidth}
    \centering
    \fontsize{10pt}{12pt}\selectfont
    \def\svgwidth{\linewidth}
\begingroup%
  \makeatletter%
  \providecommand\color[2][]{%
    \errmessage{(Inkscape) Color is used for the text in Inkscape, but the package 'color.sty' is not loaded}%
    \renewcommand\color[2][]{}%
  }%
  \providecommand\transparent[1]{%
    \errmessage{(Inkscape) Transparency is used (non-zero) for the text in Inkscape, but the package 'transparent.sty' is not loaded}%
    \renewcommand\transparent[1]{}%
  }%
  \providecommand\rotatebox[2]{#2}%
  \newcommand*\fsize{\dimexpr\f@size pt\relax}%
  \newcommand*\lineheight[1]{\fontsize{\fsize}{#1\fsize}\selectfont}%
  \ifx\svgwidth\undefined%
    \setlength{\unitlength}{260.45873593bp}%
    \ifx\svgscale\undefined%
      \relax%
    \else%
      \setlength{\unitlength}{\unitlength * \real{\svgscale}}%
    \fi%
  \else%
    \setlength{\unitlength}{\svgwidth}%
  \fi%
  \global\let\svgwidth\undefined%
  \global\let\svgscale\undefined%
  \makeatother%
  \begin{picture}(1,0.92400445)%
    \lineheight{1}%
    \setlength\tabcolsep{0pt}%
    \put(0,0){\includegraphics[width=\unitlength,page=1]{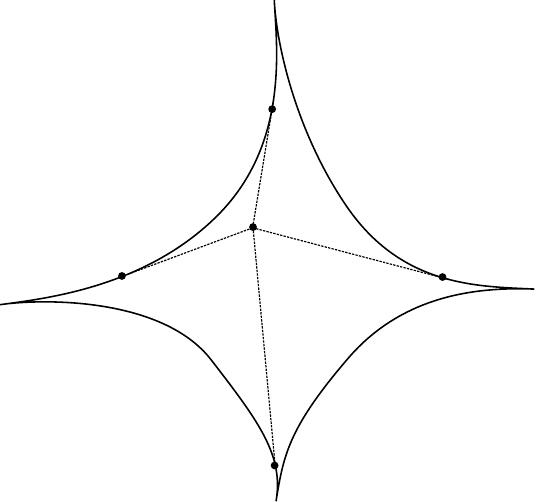}}%
    \put(0.48647394,0.51579916){\color[rgb]{0,0,0}\makebox(0,0)[lt]{\lineheight{1.25}\smash{\begin{tabular}[t]{l}$z$\end{tabular}}}}%
    \put(0.82805106,0.44379108){\color[rgb]{0,0,0.50196078}\makebox(0,0)[lt]{\lineheight{1.25}\smash{\begin{tabular}[t]{l}$z_0$\end{tabular}}}}%
    \put(0.42000482,0.74844079){\color[rgb]{0,0,0.50196078}\makebox(0,0)[lt]{\lineheight{1.25}\smash{\begin{tabular}[t]{l}$z_1$\end{tabular}}}}%
    \put(0.17443883,0.44748372){\color[rgb]{0,0,0.50196078}\makebox(0,0)[lt]{\lineheight{1.25}\smash{\begin{tabular}[t]{l}$z_2$\end{tabular}}}}%
    \put(0.41723534,0.03851448){\color[rgb]{0,0,0.50196078}\makebox(0,0)[lt]{\lineheight{1.25}\smash{\begin{tabular}[t]{l}$z_3$\end{tabular}}}}%
  \end{picture}%
\endgroup%

\end{subfigure}
\caption{Tangential zeros at the three types of components of $[-1,1]^2 \setminus \mathcal{F}_\mathcal{S}$. From left to right 
frozen, quasi-frozen and gas regions are shown.}
\label{fig:frozen_regions_tangential_zeros}
\end{figure}
%
For concreteness let $F\in \mathscr{F}$ be the frozen zone corresponding to $\ell_{(\alpha^-_m, \beta^-_1)}$. By Corollary~\ref{cor:classification_of_all_dZeta_zeros} for a point $(u,v) = z \in F$ the two tangential points $z_0, z_1$ correspond to zeros of $d\zeta_{(u,v)}$, see Figure~\ref{fig:frozen_regions_tangential_zeros}. We define
\begin{equation}
    \label{eq:g_extension_frozen_zone_F}
   g_1(z) := g(z_0) + \pobr{\nabla g(z_0), z - z_0}.
\end{equation}
This is a linear extension along each oriented tangential line and due to Corollary~\ref{cor:tangential_extension_gradient} this defines a function with $J\nabla g_1(z) = J\nabla g(z_0) \in \mathcal{A}_F$.

For a quasi-frozen zone $Q\in\mathscr{Q}$ there are three tangential points $z_0, z_1, z_2$ which are ordered on $X_0$. We take the middle one:
\begin{equation}
\label{eq:g_extension_frozen_zone_Q}
    g_1(z) := g(z_1) + \pobr{\nabla g(z_1), z - z_1}.
\end{equation}
\begin{lemma}
\label{lem:g_extension_on_frozen}
    The function $\hat{g}:=g_1(z)$ satisfies the properties of Proposition~\ref{prop:generalized_Euler_Lagrange_suff_cond} on any frozen or quasi-frozen region.
\end{lemma}
\begin{proof}
    For concreteness we consider the frozen region $F$ corresponding to $\ell_{(\alpha^-_1, \beta^+_n)}$ and the extension $g_1$ of $g$ defined by \eqref{eq:g_extension_frozen_zone_F} where the tangential zero $z_0$ is the one closer to $\alpha^-_1$.

    According to Corollary~\ref{cor:tangential_extension_gradient} we have $J \nabla g_1(z) \in \mathcal{A}_F = \partial\sigma(h(z))$. Thus it remains to show only that $J\nabla \hat{g}(z) \in L^2(F)$. This is a question about the integrability of singularities of $d\zeta_k$.

    Let $z$ be a local coordinate on $\mathcal{R}$ at $\alpha^-_1$ with $z(\alpha^-_1) = 0$ and real on the oval. In this coordinate the amoeba map behaves as follows: 
    $$
x_1=\re \zeta_1(z)=\log z +\mathcal{O}(1), \ x_2= \re \zeta_2(z)=\mathcal{O}(1),\ z\to 0.     
    $$
   Consider the arctic curve at $\alpha^-_1$ given as the graph of a function $u \mapsto v(u)$, see Fig.~\ref{fig:frozen_area_slice}. Then, due to Proposition~\ref{prop:parallelity_prop}, the angle $\alpha$ is
    \begin{equation}
    \label{eq:alpha}
    \tan \alpha := \frac{dv}{du} = \frac{d\zeta_2}{d\zeta_1}. 
    \end{equation}
    The region $F$ is swept by tangential rays starting on $\partial F$ oriented away from $\alpha^-_1$.  Let us compute the leading term of $\int_F (\re (\zeta_1))^2 dA$ at $\alpha^-_1$ i.e. at $z=0$ where $A$ is the area measure on $F$. Differentiating \eqref{eq:alpha} and using $d\zeta_2/d\zeta_1(z)=z \cdot \mathcal{O}(1)$, 
  we get  $d\alpha = dr \cdot \mathcal{O}(1)$ where $r=|z|$.
  We denote by $\Delta$ an infinitesimal area element swept by two tangential rays, see Figure~\ref{fig:frozen_area_slice}. Since all the edges in $\Delta$ are finite we have for its area:
    \[ |\Delta| = dA =  d\alpha \cdot \mathcal{O}(1). \]
    Further we have $\re \zeta_1 = \log r + \mathcal{O}(1)$ thus yielding
    \begin{align*}
        \int (\re \zeta_1)^2 dA = \int_0^{r_0} \log^2 (r) dr \cdot \mathcal{O}(1)
       = r_0 \br{\log^2(r_0) - 2 \log(r_0) + 2 }\cdot \mathcal{O}(1),
    \end{align*}
    which is bounded. Hence $\re \zeta_1$ at the tangential point is in $L^2$ on $F$ and by an equivalent argument the same holds true for $\re \zeta_2$. Thus $J\nabla \hat{g}(z)$ 
    as in \eqref{eq:g_extension_frozen_zone_F} is in $L^2$.

 The same calculation yields this result for quasi-frozen regions $Q \in \mathscr{Q}$.
\end{proof}

\begin{figure}[h]
    \centering
    \fontsize{10pt}{12pt}\selectfont
    \def\svgwidth{.5\linewidth}
\begingroup%
  \makeatletter%
  \providecommand\color[2][]{%
    \errmessage{(Inkscape) Color is used for the text in Inkscape, but the package 'color.sty' is not loaded}%
    \renewcommand\color[2][]{}%
  }%
  \providecommand\transparent[1]{%
    \errmessage{(Inkscape) Transparency is used (non-zero) for the text in Inkscape, but the package 'transparent.sty' is not loaded}%
    \renewcommand\transparent[1]{}%
  }%
  \providecommand\rotatebox[2]{#2}%
  \newcommand*\fsize{\dimexpr\f@size pt\relax}%
  \newcommand*\lineheight[1]{\fontsize{\fsize}{#1\fsize}\selectfont}%
  \ifx\svgwidth\undefined%
    \setlength{\unitlength}{173.41223545bp}%
    \ifx\svgscale\undefined%
      \relax%
    \else%
      \setlength{\unitlength}{\unitlength * \real{\svgscale}}%
    \fi%
  \else%
    \setlength{\unitlength}{\svgwidth}%
  \fi%
  \global\let\svgwidth\undefined%
  \global\let\svgscale\undefined%
  \makeatother%
  \begin{picture}(1,0.6185449)%
    \lineheight{1}%
    \setlength\tabcolsep{0pt}%
    \put(0,0){\includegraphics[width=\unitlength,page=1]{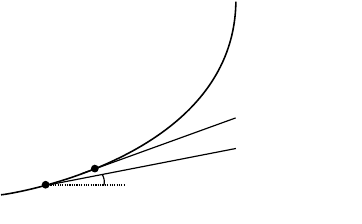}}%
    \put(0.30440919,0.12005841){\color[rgb]{0,0,0}\makebox(0,0)[lt]{\lineheight{1.25}\smash{\begin{tabular}[t]{l}$\alpha$\end{tabular}}}}%
    \put(0,0){\includegraphics[width=\unitlength,page=2]{Figures/FrozenRegionAreaL2.pdf}}%
    \put(0.59667185,0.22570827){\color[rgb]{0,0,0}\makebox(0,0)[lt]{\lineheight{1.25}\smash{\begin{tabular}[t]{l}$d\alpha$\end{tabular}}}}%
    \put(0.47904525,0.19408686){\color[rgb]{0.50196078,0,0}\makebox(0,0)[lt]{\lineheight{1.25}\smash{\begin{tabular}[t]{l}$\Delta$\end{tabular}}}}%
    \put(0,0){\includegraphics[width=\unitlength,page=3]{Figures/FrozenRegionAreaL2.pdf}}%
    \put(-0.00219709,0.01044071){\color[rgb]{0,0.50196078,0}\makebox(0,0)[lt]{\lineheight{1.25}\smash{\begin{tabular}[t]{l}$\alpha^-_1$\end{tabular}}}}%
    \put(0.66729107,0.57860665){\color[rgb]{1,0,0}\makebox(0,0)[lt]{\lineheight{1.25}\smash{\begin{tabular}[t]{l}$\beta^-_n$\end{tabular}}}}%
  \end{picture}%
\endgroup%

\caption{Infinitesimal area element $\Delta$ defined by two closeby tangential rays with angles differing by $d\alpha$.}
\label{fig:frozen_area_slice}
\end{figure}

\begin{rem}
    A similar extension to frozen regions was first introduced in \cite{astala_dimer_2023} in terms of an extension of a divergence free field. We realize this as an admissible extension of the function $g$ here which makes the construction slightly simpler, gives a nice geometric interpretation and allows us to extend this method to gas regions.
\end{rem}

\subsection{Extension to Gas Regions.}
Now let $G_i \in \mathscr{G}$ be the gas bubble corresponding to $X_i, i \neq 0$ with boundary $\partial G_i = \mathcal{F}(X_i)$. The function $g$ is defined on $\partial G_i$ with rotated gradients $J\nabla g(z) = \mathcal{A}\circ \mathcal{F}^{-1} (z) \in \partial\mathcal{A}_i$ for all $z \in \partial G_i$.

We will make use of some positivity properties of $d\zeta$. Let us fix the positive orientation on $\mathcal{R}^+$ resulting in positive orientation on $X_0$ and negative orientation on $X_i, i \neq 0$ and the same orientation of the real ovals in the amoeba. Since the map $\mathcal{F}$ is orientation-reversing, the gas region $G_i$ is oriented counter-clockwise. For two points $P,Q \in X_i$ with $p = \mathcal{F}(P), q = \mathcal{F}(Q)$ we will denote by $[P, Q]$ the path on $X_i$ from $P$ to $Q$ following our orientation as well as $[p,q] = \mathcal{F}([P,Q]) \subset \partial G_i$.

We fix a coordinate $dz$ such that it is real and positive along $X_i$ in the chosen orientation. In this coordinate we have
\[ d\zeta_{(u,v)} = \br{-u f_2(z) + v f_1(z) + f_3(z)} dz, \]
and the $f_j$ are real on $X_i$. Let $q = (u,v)$ be a point on $\partial G_i$. Then
\[ f_q(z) = f_{(u,v)}(z) := -u f_2(z) + v f_1(z) + f_3(z)\]
has a double zero at $Q \in X_i$ corresponding to $q$ and two simple zeros on $X_i$ at the points $P_1, P_2$ that correspond to the tangency points $p_1, p_2$ on $\partial G_i$, see Figure \ref{fig:gas_MCurve_orientations}. 

\begin{figure}[h]
\centering
\begin{subfigure}{.3\textwidth}
    \centering
    \fontsize{10pt}{12pt}\selectfont
    \def\svgwidth{\linewidth}
\begingroup%
  \makeatletter%
  \providecommand\color[2][]{%
    \errmessage{(Inkscape) Color is used for the text in Inkscape, but the package 'color.sty' is not loaded}%
    \renewcommand\color[2][]{}%
  }%
  \providecommand\transparent[1]{%
    \errmessage{(Inkscape) Transparency is used (non-zero) for the text in Inkscape, but the package 'transparent.sty' is not loaded}%
    \renewcommand\transparent[1]{}%
  }%
  \providecommand\rotatebox[2]{#2}%
  \newcommand*\fsize{\dimexpr\f@size pt\relax}%
  \newcommand*\lineheight[1]{\fontsize{\fsize}{#1\fsize}\selectfont}%
  \ifx\svgwidth\undefined%
    \setlength{\unitlength}{318.06577329bp}%
    \ifx\svgscale\undefined%
      \relax%
    \else%
      \setlength{\unitlength}{\unitlength * \real{\svgscale}}%
    \fi%
  \else%
    \setlength{\unitlength}{\svgwidth}%
  \fi%
  \global\let\svgwidth\undefined%
  \global\let\svgscale\undefined%
  \makeatother%
  \begin{picture}(1,1.22555425)%
    \lineheight{1}%
    \setlength\tabcolsep{0pt}%
    \put(0,0){\includegraphics[width=\unitlength,page=1]{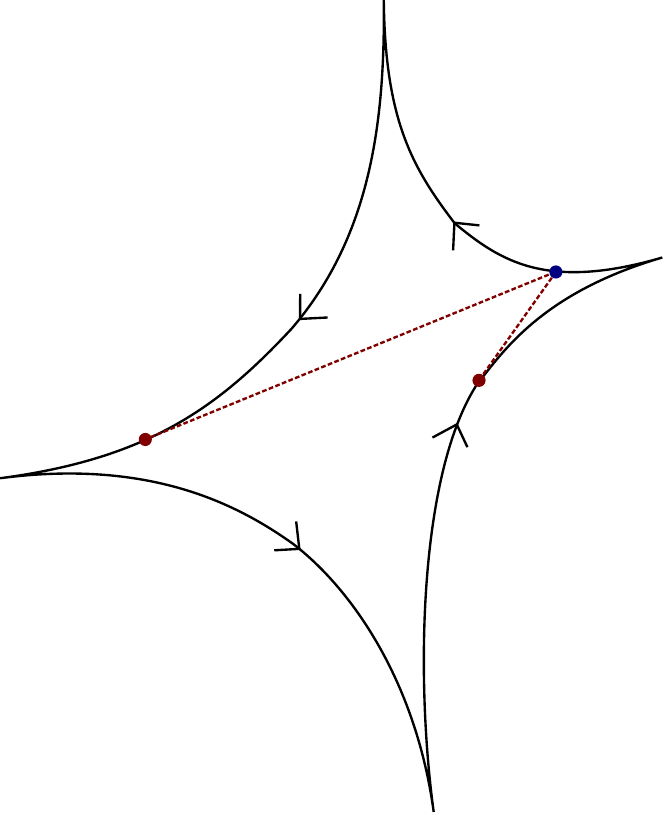}}%
    \put(0.82899022,0.84684417){\color[rgb]{0,0,0.50196078}\makebox(0,0)[lt]{\lineheight{1.25}\smash{\begin{tabular}[t]{l}$q$\end{tabular}}}}%
    \put(0.75463991,0.62478468){\color[rgb]{0.50196078,0,0}\makebox(0,0)[lt]{\lineheight{1.25}\smash{\begin{tabular}[t]{l}$p_2$\end{tabular}}}}%
    \put(0.17768183,0.60198391){\color[rgb]{0.50196078,0,0}\makebox(0,0)[lt]{\lineheight{1.25}\smash{\begin{tabular}[t]{l}$p_1$\end{tabular}}}}%
  \end{picture}%
\endgroup%

\end{subfigure}
\begin{subfigure}{.15\textwidth}
    
\end{subfigure}
\begin{subfigure}{.3\textwidth}
    \centering
    \fontsize{10pt}{12pt}\selectfont
    \def\svgwidth{\linewidth}
\begingroup%
  \makeatletter%
  \providecommand\color[2][]{%
    \errmessage{(Inkscape) Color is used for the text in Inkscape, but the package 'color.sty' is not loaded}%
    \renewcommand\color[2][]{}%
  }%
  \providecommand\transparent[1]{%
    \errmessage{(Inkscape) Transparency is used (non-zero) for the text in Inkscape, but the package 'transparent.sty' is not loaded}%
    \renewcommand\transparent[1]{}%
  }%
  \providecommand\rotatebox[2]{#2}%
  \newcommand*\fsize{\dimexpr\f@size pt\relax}%
  \newcommand*\lineheight[1]{\fontsize{\fsize}{#1\fsize}\selectfont}%
  \ifx\svgwidth\undefined%
    \setlength{\unitlength}{121.68593064bp}%
    \ifx\svgscale\undefined%
      \relax%
    \else%
      \setlength{\unitlength}{\unitlength * \real{\svgscale}}%
    \fi%
  \else%
    \setlength{\unitlength}{\svgwidth}%
  \fi%
  \global\let\svgwidth\undefined%
  \global\let\svgscale\undefined%
  \makeatother%
  \begin{picture}(1,1)%
    \lineheight{1}%
    \setlength\tabcolsep{0pt}%
    \put(0,0){\includegraphics[width=\unitlength,page=1]{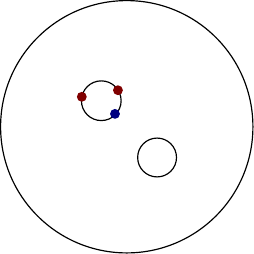}}%
    \put(0.4561399,0.69074696){\color[rgb]{0.50196078,0,0}\makebox(0,0)[lt]{\lineheight{1.25}\smash{\begin{tabular}[t]{l}$P_2$\end{tabular}}}}%
    \put(0.21648362,0.65585155){\color[rgb]{0.50196078,0,0}\makebox(0,0)[lt]{\lineheight{1.25}\smash{\begin{tabular}[t]{l}$P_1$\end{tabular}}}}%
    \put(0.42896077,0.45244406){\color[rgb]{0,0,0.50196078}\makebox(0,0)[lt]{\lineheight{1.25}\smash{\begin{tabular}[t]{l}$Q$\end{tabular}}}}%
    \put(0,0){\includegraphics[width=\unitlength,page=2]{Figures/MCurveOrientation.pdf}}%
  \end{picture}%
\endgroup%

\end{subfigure}
\caption{\textbf{Left}: Gas bubble $G_i$ with orientation and tangential zeros $p_1, p_2$ of $d\zeta_q$ following orientation. \textbf{Right}: Uniformization of $\mathcal{R}^+$ along with chosen orientation and points $P_1, P_2, Q$ on $X_i$ corresponding to $p_1, p_2, q$.}
\label{fig:gas_MCurve_orientations}
\end{figure}

Let us compute the behaviour of $f_q$ at its double zero $Q$. We have $f_q(Q)=f'_q(Q)=0$.
Using  \eqref{eq:xy_on_real_ovals} we get for the second derivative
\begin{equation}
\label{eq:second_der_as_Wronskians}
 f''(Q) = - \frac{W(f_1, f_3)}{W(f_1, f_2)}(Q) f_2''(Q) + \frac{W(f_2, f_3)}{W(f_1, f_2)}(Q) f_1''(Q) + f_3''(Q) 
           = \frac{W(f_1, f_2, f_3)}{W(f_1, f_2)}(Q).
\end{equation}

%

Since the cusps of $\partial G_i$ are characterized by $W(f_1, f_2, f_3) = 0$ as seen in the proof of Proposition \ref{prop:parallelity_prop}, and $W(f_1, f_2) \neq 0$ by Lemma~\ref{lem:W(f_1,f_2)_not_zero_on_real_ovals} we obtain 
the following
\begin{lemma}
\label{lemma:signs_of_fq(Q)}
On a real oval of a gas region the second derivative $f_q''(Q)$ changes its sign exactly at the cusps.
\end{lemma}
Thus the sign of $f_q''(Q)$ on an arc bounded by two consecutive cusps is constant. We will call such arcs \emph{positive/negative} depending on the sign of $f_q''(Q)$.

%

\begin{lemma}
    \label{lemma:signs_of_dZetaQ_Integral}
    Let $Q \in X_i$ such that $f_q''(Q) > 0$ and $P_1, P_2 \in X_i$ the two other zeros of $f_q$ such that the path $(P_2, Q, P_1)$ follows the defined orientation as in Figure \ref{fig:gas_MCurve_orientations}. Then
    \[ \int_Q^{P_2} d\zeta_q \leq 0, \int_Q^{P_1} d\zeta_q \geq 0. \]
    Similarly if $f_q''(Q) < 0$, then
    \[ \int_Q^{P_2} d\zeta_q \geq 0, \int_Q^{P_1} d\zeta_q \leq 0. \]
\end{lemma}
\begin{proof}
    We know that $f_q$ switches signs at the $P_i$. It follows from $f_q''(Q) > 0$ that $f_q \geq 0$ on the arc from $P_1$ to $P_2$ containing $Q$. That is the arc from $[P_2,P_1]$ following our defined orientation. The claim follows. 
\end{proof}

We will now use these positivity properties to construct an extension of $g$ to $G_i$ from its values on the boundary $\partial G_i$. To this end we denote the cusps around $G_i$ as $c_1, \ldots, c_4$ in order of our fixed orientation and $C_1, \ldots, C_4$ the corresponding points on $X_i$. Without loss of generality we assume that $f''_q(Q)$ is positive on $[C_4, C_1]$. By Lemma~\ref{lemma:signs_of_fq(Q)} this implies that $f''_q(Q)$ is also positive on $[C_2, C_3]$ and negative on other two arcs.

Let $p$ be on the positive arc $[c_4, c_1]$ and consider the ray $r$ starting at $p$ that is tangential to $\partial G_i$ and oriented away from $C_1$. Let $q = \mathcal{F}(Q)$ be an intersection point of $r$ and $\partial G_i$, see Fig.~\ref{fig:gas_bubble_with_p_and_qs}. As before, for any $z \in r$  the function $g_1$ is defined as
\begin{equation}
    \label{eq:def_g1}
    g_1(z) = g(P) + \pobr{\nabla g(P), z - p}.
\end{equation}
We have
\begin{equation}
\label{eq:g-g_1_as_integral}
    \begin{split}
        g(Q) - g_1(Q) &= g(Q) - g(P) - \pobr{\nabla g(P), q - p}  \\
        &= \re \br{\int_P^Q d\zeta_q + \int^P d\zeta_q - d\zeta_p} - \pobr{\nabla g(P), q - p}  \\
         &= \int_P^Q d\zeta_q + \pobr{J(q-p), \mathcal{A}(P)} - \pobr{\nabla g(P), q - p} 
         = \int_P^Q d\zeta_q.
    \end{split}
\end{equation}
Here we used \eqref{eq:grads_of_f_and_g}, and $\mathcal{A}(P)$ is the amoeba map.

We assume for now that $q$ lies on the neighboring arc $[c_3, c_4]$. Then $d\zeta_q$ has a double zero at $Q$, a zero at $P$ due to the tangential construction and a fourth zero somewhere on $[C_1, C_3]$, see Fig:~\ref{fig:gas_bubble_with_p_and_qs} left. Thus by Lemma \ref{lemma:signs_of_dZetaQ_Integral} and since $[C_3, C_4]$ is a negative arc, we have
\[ g(Q) - g_1(Q) = \int_P^Q d\zeta_q \geq 0. \]
Similarly, consider a point $\tilde{p}$ corresponding to $\tilde{P} \in [C_4, C_1]$ with its tangential ray intersecting $\partial G_i$ at $\tilde{q}$ corresponding to $\tilde{Q}$ on the positive arc $[C_2, C_3]$, see Fig.~\ref{fig:gas_bubble_with_p_and_qs}. Then $d\zeta_{\tilde{q}}$ has a double zero at ${\tilde{Q}}$, a zero at $\tilde{P}$, and the fourth one must lie on $[C_3, C_1]$ and thus again by Lemma \ref{lemma:signs_of_dZetaQ_Integral} we obtain
\[ g(\tilde{Q}) - g_1(\tilde{Q}) = \int_{\tilde{P}}^{\tilde{Q}} d\zeta_{\tilde{q}} \geq 0. \]
An analogous argument gives the same inequalities for $p$ on $[c_1, c_2]$. Similarly we define a function $g_3$ as in \eqref{eq:def_g1} as an extension along tangential rays starting at $[c_2, c_4]$ and oriented away from $c_3$. An equivalent argument yields the same inequalities for $g_3$.

\begin{figure}[h]
\centering
\begin{subfigure}{.4\textwidth}
    \centering
    \fontsize{10pt}{12pt}\selectfont
    \def\svgwidth{\linewidth}
    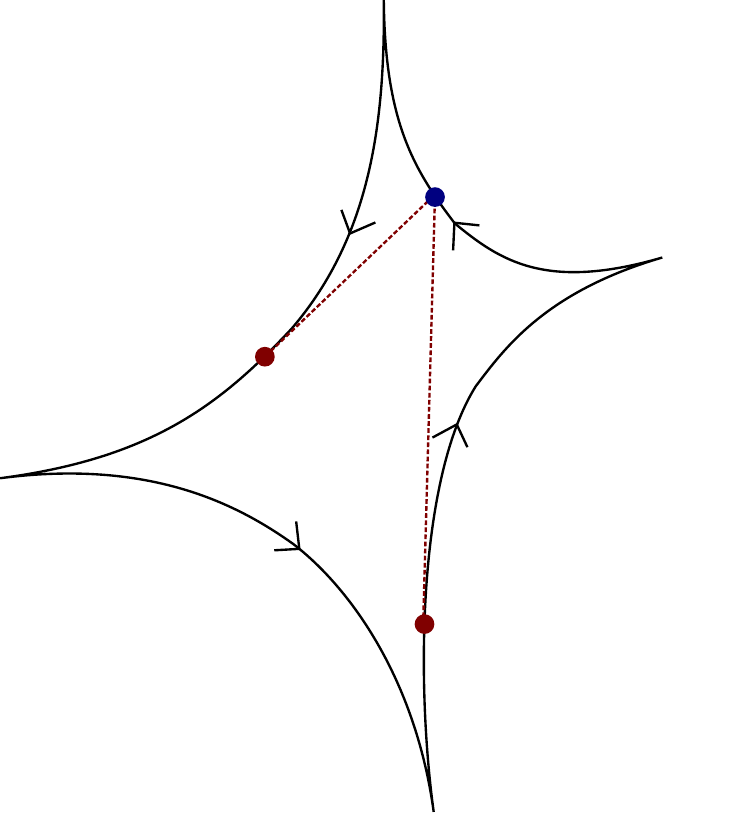
\end{subfigure}
\begin{subfigure}{.15\textwidth}
    
\end{subfigure}
\begin{subfigure}{.4\textwidth}
    \centering
    \fontsize{10pt}{12pt}\selectfont
    \def\svgwidth{\linewidth}
    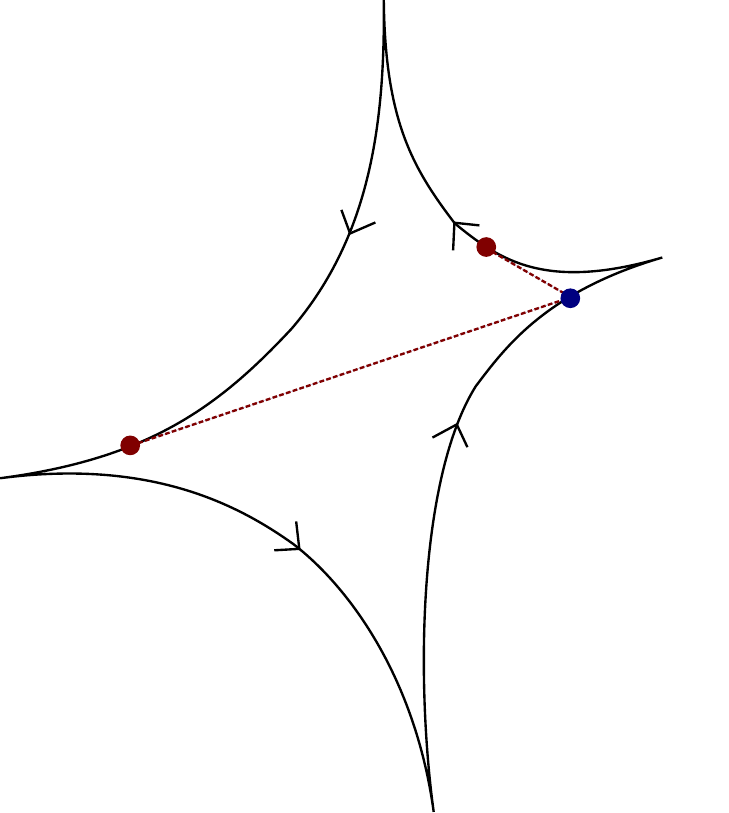
\end{subfigure}
\caption{Gas bubble $G_i$ with the two possible configurations of points on $[c_4, c_1]$ with signs of $f_q, f_{\tilde{q}}$. \textbf{Left}: The oriented ray starting at $p \in [c_4, c_1]$ intersects $\partial G_i$ at the point $q \in [c_3, c_4]$. The function $f_q$ is negative on $[P_1, P]$ thus yielding $\int_P^Q d\zeta_q \geq 0$. \textbf{Right}: The oriented ray starting at $\tilde{p} \in [c_4, c_1]$ intersects $\partial G_i$ at the point $\tilde{q} \in [c_2, c_3]$. The function $f_{\tilde{q}}$ is then positive on $[\tilde{P}, \tilde{P}_1]$ thus again yielding $\int_{\tilde{P}}^{\tilde{Q}} d\zeta_{\tilde{q}} \geq 0$. }
\label{fig:gas_bubble_with_p_and_qs}
\end{figure}

We have thus constructed two functions $g_1, g_3$ on $G_i$ with the following properties:
\begin{equation}
    \left\{ 
   \begin{array}{lr}
        J\nabla g_j(z) \in \mathcal{A}_i \; \forall z \in G_i \\
        \restr{g_1-g}{[C_4, C_2]} = 0, \; \restr{g_1-g}{[C_2, C_4]} \geq 0 \\
        \restr{g_3-g}{[C_2, C_4]} = 0, \; \restr{g_3-g}{[C_4, C_2]} \geq 0.
    \end{array} 
    \right.
\end{equation}

Taking
\[ \hat{g}:= \min(g_1, g_3) \]
on $G_i$ then yields a Lipschitz function with boundary values $\restr{\hat{g}-g}{\partial G_i} = 0$, and rotated gradients existing almost everywhere and lying in $\mathcal{A}_i$. These gradients are bounded.

By setting $\hat{g}$ to be the described extension of $g$ to both gas regions and frozen regions as in Lemma~\ref{lem:g_extension_on_frozen} we have constructed a function satisfying all conditions of Proposition \ref{prop:generalized_Euler_Lagrange_suff_cond} and Theorem \ref{thm:h_is_height_function} follows. 

\begin{rem}
\label{rem:g_extension_non_unique}
    Admissible extensions $\hat{g}$ of $g$ are not unique. For gas regions we have given an explicit construction for a maximal extension. Similarly using the other two cusps with the extension $\hat{g} = \max(g_2, g_4)$ gives a minimal extension. Any function that lies within the envelope of these maximal and minimal extensions while still satisfying the gradient condition gives a possible extension for $g$. Similarly, on the frozen and quasi-frozen zones we have provided one choice of extension but there is functional freedom there just as on the gas regions.

    By construction, the graphs of the functions $g_1$ and $g_3$ on gas regions are ruled surfaces that we have shown to intersect above $G_i$. The locus of intersection is a collection of smooth curves and $\hat{g}$ is smooth everywhere except at those curves. Due to convexity of $\mathcal{A}_i$ by mollifying $\hat{g}$ locally around its non-differentiabilities we obtain a smooth extension of $g$ that satisfies the conditions of Proposition~\ref{prop:generalized_Euler_Lagrange_suff_cond}. Such a smooth extension has rotated gradients inside $\mathcal{A}_i$ and not just on the real oval $\partial \mathcal{A}_i$.
\end{rem}

We note that any admissible extension $\hat{g}$ satisfies the Euler-Lagrange equation 
\begin{equation}
\label{eq:magnetic_Euler_Lagrange}
    \Div(J\nabla\rho(J\nabla \hat{g})) = 0 
\end{equation}
in the liquid region by Corollary~\ref{cor:Euler_lagrange_h_and_g}. Furthermore \eqref{eq:magnetic_Euler_Lagrange} is also satisfied on the regions of extension because $\nabla\rho$ is constant on each of them. Thus \eqref{eq:magnetic_Euler_Lagrange} is satisfied on all of $[-1, 1]^2$ for a smooth $\rho$. We obtain

\begin{theorem}
    \label{thm:g_is_magnetic_tension_minimizer}
    Let $\hat{g}:[-1,1]^2 \to \mathbb{R}$ be an admissible extension of $g : \mathcal{F}_\mathcal{S} \to \mathbb{R}$ and $g_b = \restr{\hat{g}}{\partial [-1,1]^2}$ its boundary conditions. Then $\hat{g}$ is a minimizer of the \emph{magnetic tension functional}
    \[ I_\rho(\varphi) = \int_{[-1,1]^2} \rho(J \nabla \varphi) \]
    over the Sobolev space $W^{1,2}([-1,1]^2, \mathbb{R})$ restricted to functions with boundary values $g_b$.
\end{theorem}

The freedom of extension stems from the fact that $\rho$ is not strictly convex on the connected components of $[-1,1]^2 \setminus \mathcal{A}_\mathcal{S}$. Note that the liquid region is the same for all minimizers over the set of admissible boundary conditions. 

\begin{rem}
    Dimer height functions and their convergence to the surface tension minimizer $\hat{h}$ have been studied extensively in the literature as we have seen above. With the magnetic tension minimizer $\hat{g}$ we have introduced a natural dual function on the level of limiting objects. It is a very interesting problem to build an understanding for it on the discrete level. This is an open question which is conjecturally related to an optimal gauge transformation.
\end{rem}

\section{Explicit Formulas for Isoradial Weights}
\label{sec:06_explicit_formulas_for_aztec_diamond_G0}

Let $\mathcal{R}$ be an M-curve of genus 0. It can be represented as the extended complex plane $\mathcal{R}=\mathbb{C}\cup\{ \infty\}$ with the anti-holomorphic involution $\tau(z)=\frac{1}{\bar{z}}$. It has only one real oval $X_0=\{z\in\mathbb{C} \mid |z|=1\}$, and $\mathcal{R}_+^\circ=\{z\in\mathbb{C} \mid |z|<1\}$.  The Harnack data are given by the clustered numbers on the unit circle:
\begin{equation}
\label{eq:isoradial_alphabeta}
\alpha^\pm_i=e^{ia^\pm_i}, \ \beta^\pm_i=e^{ib^\pm_i}, \ |\alpha^\pm_i |=|\beta^\pm_j |=1.
\end{equation}

The algebro-geometric approach based on Baker-Akhiezer functions \cite{BBS} can be applied to this case as well. The weights
\begin{equation}
    \label{eq:isoradial_weights}
    K_{\alpha,\beta}=\alpha-\beta=e^{ia}-e^{ib}
\end{equation}
become the dimer weights on isoradial graphs \cite{kenyon_laplacian_2002}. 

The differentials are
\begin{equation}
\label{eq:dzeta_i_isoradial}
    d\zeta_k(z)=f_k(z)dz, \; k=1,2,3, \; d\zeta_{(u,v)}(z)=f_{(u,v)}(z)dz
\end{equation}
with
\begin{align*}
f_1(z)&=\sum_{i=1}^m \left( \frac{1}{z-\alpha^-_i} - \frac{1}{z-\alpha^+_i}\right),\\
f_2(z)&=\sum_{i=1}^m \left( \frac{1}{z-\beta^-_j} - \frac{1}{z-\beta^+_j}\right),\\
f_3(z)&=\sum_{i=1}^m \left( \frac{1}{z-\alpha^-_i} + \frac{1}{z-\alpha^+_i} - \frac{1}{z-\beta^-_j} - \frac{1}{z-\beta^+_j}\right),\\
f_{(u,v)}(z)&= -u f_2(z)+vf_1(z)+f_3(z).
\end{align*} 
The diffeomorphism $\mathcal{F}$ from the open unit disc to the liquid region  
$$
\mathcal{R}_+^\circ \ni z_0 \mapsto (u,v) \in \mathcal{F}_\mathcal{S}
$$
is given by formulas \eqref{eq:uv_throughR_i}. The point $z_0$ is the only zero of $f_{(u,v)}(z)$ in the open unit disc. This map can be continuously extended to the boundary $|z|=1$ by formulas \eqref{eq:xy_on_real_ovals}, which gives an explicit parametrization of the arctic curve.

The height function at $(u,v)\in \mathcal{F}_\mathcal{S}$ corresponding to $z_0$ is given by integration of $f_{(u,v)}$:
\begin{align*}
h(u,v) =& \frac{1}{\pi}\im \int^{z_0} f_{(u,v)}(z)dz\\
=&\frac{1}{\pi}\sum_{i=1}^m (1+v) \arg (\alpha_i^- -z_0)+ (1-v) \arg (\alpha_i^+ -z_0) \\ 
&- (1+u) \arg (\beta_i^- -z_0)-(1-u) \arg (\beta_i^+ -z_0).
\end{align*}


For completeness we also derive here explicit representations for the surface tension $\sigma$ in the genus zero case, i.e. in the case when the gas phase is absent. For simplicity we consider the case of the square grid. Computations on general isoradial graphs are similar. 

We compute the surface tension given by (\ref{eq:def_rho_sigma}). First observe that a change of the integrals by constants $\zeta_1\mapsto \zeta_1+c_1$, 
$\zeta_2\mapsto \zeta_2+c_2$ changes $\sigma$ by an affine function, which leads to the same variational description, see \cite[Remark~38]{BBS}. We use this freedom to simplify the formulas: 
\begin{equation}
\label{eq:zeta_i_isoradial}
    \zeta_1(z)=\sum_i \log\left(\frac{z-\alpha^-_i}{z-\alpha^+_i} \right),\quad
    \zeta_2(z)=\sum_j \log\left(\frac{z-\beta^-_j}{z-\beta^+_j} \right).
\end{equation}
Then the term in $\sigma$ in (\ref{eq:def_rho_sigma}) corresponding to the edge labeled by the train tracks $\alpha=\alpha^-_i$ and $\beta=\beta^-_j$ is
\begin{equation}
\label{eq:sigma_alpha-beta}
\sigma_{(\alpha^-,\beta^-)}=\frac{1}{\pi}\left(\im\int^z \log(t-\beta^-)\frac{dt}{t-\alpha^-} - \log|z-\alpha^- |\arg(z-\beta^-) \right).
\end{equation}
Here the indices $i,j$ in notations are omitted. The lower integration limit is irrelevant since we compute $\sigma$ up to an affine function. By direct computation, which we perform up to constant terms, we obtain with $s=\dfrac{\alpha -z}{\alpha -\beta} $:

\begin{align}
\label{eq:dilog_formula}
&\im\int^z \log(t-\beta)\frac{dt}{t-\alpha}- \log | z-\alpha |\arg(z-\beta)\nonumber \\
=& \im\int^s \log((\alpha-\beta)(1-w))\frac{dw}{w}- \log | z-\alpha | \arg(z-\beta) \nonumber \\
=&-\im {\rm Li}_2 (s)+\log |\alpha-\beta | \arg(\alpha -z) +\arg(\alpha-\beta)\log |\alpha -z |- \log|z-\alpha^- |\arg(z-\beta^-) \nonumber \\
=&-\im {\rm Li}_2 (s)- \arg(1-s)\log s+ \log|\alpha-\beta | \arg \frac{z-\alpha}{z-\beta}+ \rm{const}.
\end{align}
Here
\[
{\rm Li}_2 (z)=-\int_0^z\frac{\log(1-t)}{t}dt.
\]
is Euler's dilogarithm function.
For computations and geometric interpretation it is more convenient to use the closely related Bloch-Wigner function (see \cite{Zagier, Bloch})
\begin{equation}
\label{eq:Bloch-Wigner}
D(z):=\im {\rm Li}_2 (z)+\arg(1-z)\log |z|,
\end{equation}
where $\arg$ denotes the branch between $-\pi$ and $\pi$. This is a continuous function, which is real analytic on $\mathbb{C}\setminus \{0,1\}$, and positive for $\im z>0$. It vanishes on the real axis and satisfies $D(\bar{z})=-D(z)$. Moreover it possesses the following symmetry properties:
\begin{equation}
\label{eq:D-symmetry}
D(z)=D(1-\frac{1}{z})=D(\frac{1}{1-z})=-D(\frac{1}{z})=-D(1-z)=-D(-\frac{z}{1-z}).
\end{equation}
It is natural to treat this function as a function of the cross-ratio of four complex variables:
\begin{equation*}
\tilde{D}(z_1,z_2,z_3,z_4):=D(z), \quad \text{where} \ z=\{z_1,z_2,z_3,z_4 \}:=\frac{(z_1-z_2)(z_3-z_4)}{(z_2-z_3)(z_4-z_1)}.
\end{equation*}
Let us consider the upper half space model of the tree dimensional hyperbolic space
\[
H^3:=\{ (z,t)\in \mathbb{C}\cross\mathbb{R}_+\}.
\]
The Bloch-Wigner function can be interpreted as the volume $V$ of the ideal hyperbolic tetrahedron $\Delta (z_1,z_2,z_3,z_4)$ with the vertices $z_1,z_2,z_3,z_4$:
\begin{equation}
    \label{eq:volume_as_Bloch_Wigner}
    V(\Delta(z_1,z_2,z_3,z_4))=\tilde{D}(z_1,z_2,z_3,z_4).
\end{equation}

The symmetry relations (\ref{eq:D-symmetry}) are just the transformations of the cross-ratio under the permutations of the points of the tetrahedron.  The minus sign there is due to the change of the orientation. Since the cross-ratio is M\"obius invariant, three points can be normalized, and we have $\tilde{D}(z_1,z_2,z_3,z_4)=\tilde{D}(z,0,1,\infty)$. The volume is positive for positively oriented tetrahedra, which correspond to tetrahedra $\Delta (z,0,1,\infty)$ with a positively oriented triangle $(z,0,1)\subset\mathbb{C}$, i.e. $\im z> 0$. See Figure~\ref{fig:ideal_hyp_tetrahedron}. For general ideal tetrahedra $\Delta (z_1,z_2,z_3,z_4)$ this is the condition for the cross-ratio $\im \{z_1,z_2,z_3,z_4 \}>0$, which is equivalent to $\tilde{D}(z_1,z_2,z_3,z_4)>0$.

\begin{figure}[h]
    \centering
    \includegraphics[width=.55\textwidth]{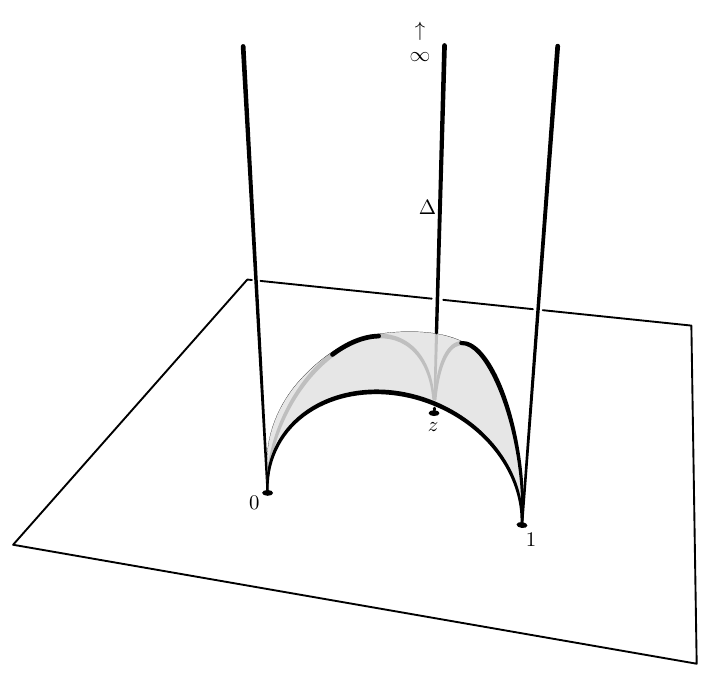}
    \caption{Ideal hyperbolic tetrahedron $\Delta(z, 0 , 1, \infty)$. Its volume is given by the Bloch-Wigner function $D(z)$, see \eqref{eq:volume_as_Bloch_Wigner}. }
    \label{fig:ideal_hyp_tetrahedron}
\end{figure}

Using the last identity from (\ref{eq:D-symmetry}) we obtain from (\ref{eq:dilog_formula})
\begin{align*}
\sigma_{(\alpha^-,\beta^-)} & =\frac{1}{\pi}\left(D(\frac{z-\alpha^-}{z-\beta^-})+ \log |\alpha^--\beta^- |\arg \frac{z-\alpha^-}{z-\beta^-} \right)\\
                                            & = -\frac{1}{\pi}\left( V(z, \alpha^-, \beta^-, \infty)+\phi_1\lambda_1 \right),
\end{align*}
where $\phi_1$ is the dihedral angle at the edge $(z,\infty)$ and $\lambda_1=\log |\alpha^--\beta^- |$ is the logarithmic length of the corresponding edge. Same formulas hold true for the other three combinations $(\alpha^+,\beta^-)$, $(\alpha^-,\beta^+)$, $(\alpha^+,\beta^+)$ in (\ref{eq:sigma_alpha-beta}). In particular,
\[
\sigma_{(\alpha^+,\beta^-)}=-\frac{1}{\pi}\left( V(z, \beta^-, \alpha^+, \infty)+\phi_2\lambda_2 \right),
\]
where $\phi_2$ is the dihedral angle of $\Delta(z, \beta^-, \alpha^+, \infty)$ at the edge $(z,\infty)$ and $\lambda_2=\log |\alpha^+-\beta^- |$.
Note that the pairs $(\alpha^-,\beta^-)$ and $(\beta^-,\alpha^+)$, see Fig.~\ref{fig:sq_lattice}, correspond to the same orientation of the train tracks introduced  in Fig.~\ref{fig:Fock_diamond_and_face} as $(\alpha,\beta)$.  We conclude with the following theorem.

\begin{theorem}
The surface tension function is given by the sum over the edges of the graph $G$
\begin{align}
\label{eq:sigma_BW}
\sigma(z)& = \frac{1}{\pi}\sum_{e\in E(G)} D(\frac{z-\alpha}{z-\beta})+\log |\alpha - \beta | \arg \frac{z-\alpha}{z-\beta}=
\\
\label{eq:sigma_hyperbolic}
    & - \frac{1}{\pi}\sum_{e\in E(G)} V(z,\alpha,\beta,\infty)+\phi(e)\lambda(e),
\end{align}
where $\alpha, \beta\in\mathbb{C}, |\alpha |= |\beta |=1$ are labels of the train tracks of the edge $e$ arranged as in Fig.~\ref{fig:Fock_diamond_and_face}, $D(z)$ is the Bloch-Wigner function (\ref{eq:Bloch-Wigner}), $V$ is the volume of the corresponding ideal hyperbolic tetrahedron $\Delta (z,\alpha,\beta,\infty)$ with the dihedral angle $\phi(e)$ at the edge $(z,\infty)$, and $\lambda(e)=\log |\alpha - \beta | $ is the logarithmic length of the corresponding edge.
\end{theorem}

This theorem holds true for general isoradial graphs $G$. It can be proven by the same computation, using the multidimensional consistency of the Dirac operators on quad graphs, see \cite{BBS}.

\begin{rem}
For unitary arguments one has 
\[
D(e^{2i\theta})=\im {\rm Li}_2 (e^{2i\theta})=2 L(\theta),
\]
where $L(\theta)=-\int_0^\theta \log 2\sin t dt$ is the Lobachevsky function. So, substituting $z=0$ to (\ref{eq:sigma_BW}) we recover Kenyon's formula \cite{kenyon_laplacian_2002} for the normalized determinant of the discrete Dirac operator for isoradial embeddings
\[
\sigma (e, z=0)= \frac{2}{\pi} (L(\theta)+\theta \log 2\sin \theta ) ,
\]
where $2\theta=\beta-\alpha$.
\end{rem}

The functional given by (\ref{eq:sigma_BW},\ref{eq:sigma_hyperbolic}) is closely related to the functionals describing discrete conformal mappings \cite{BPS_discrete_conformal}.

See Figure~\ref{fig:genus_0_configurations} for some genus 0 examples of configurations and arctic curves.

\begin{figure}
    \centering
    \begin{subfigure}[p]{.49\textwidth}
        \centering
        \includegraphics[width=\linewidth]{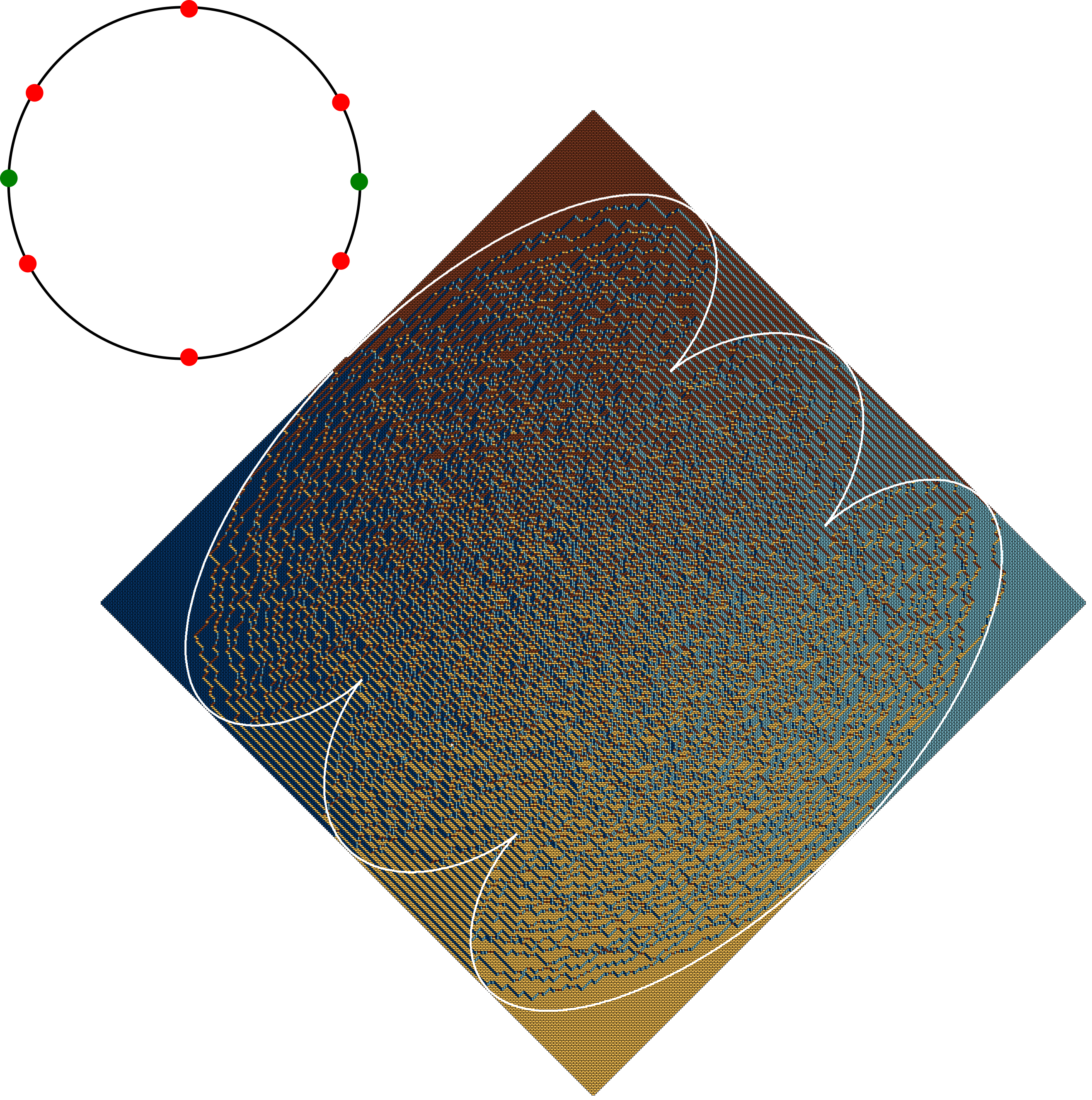}
    \end{subfigure}
    \begin{subfigure}[p]{.49\textwidth}
        \centering
        \includegraphics[width=\linewidth]{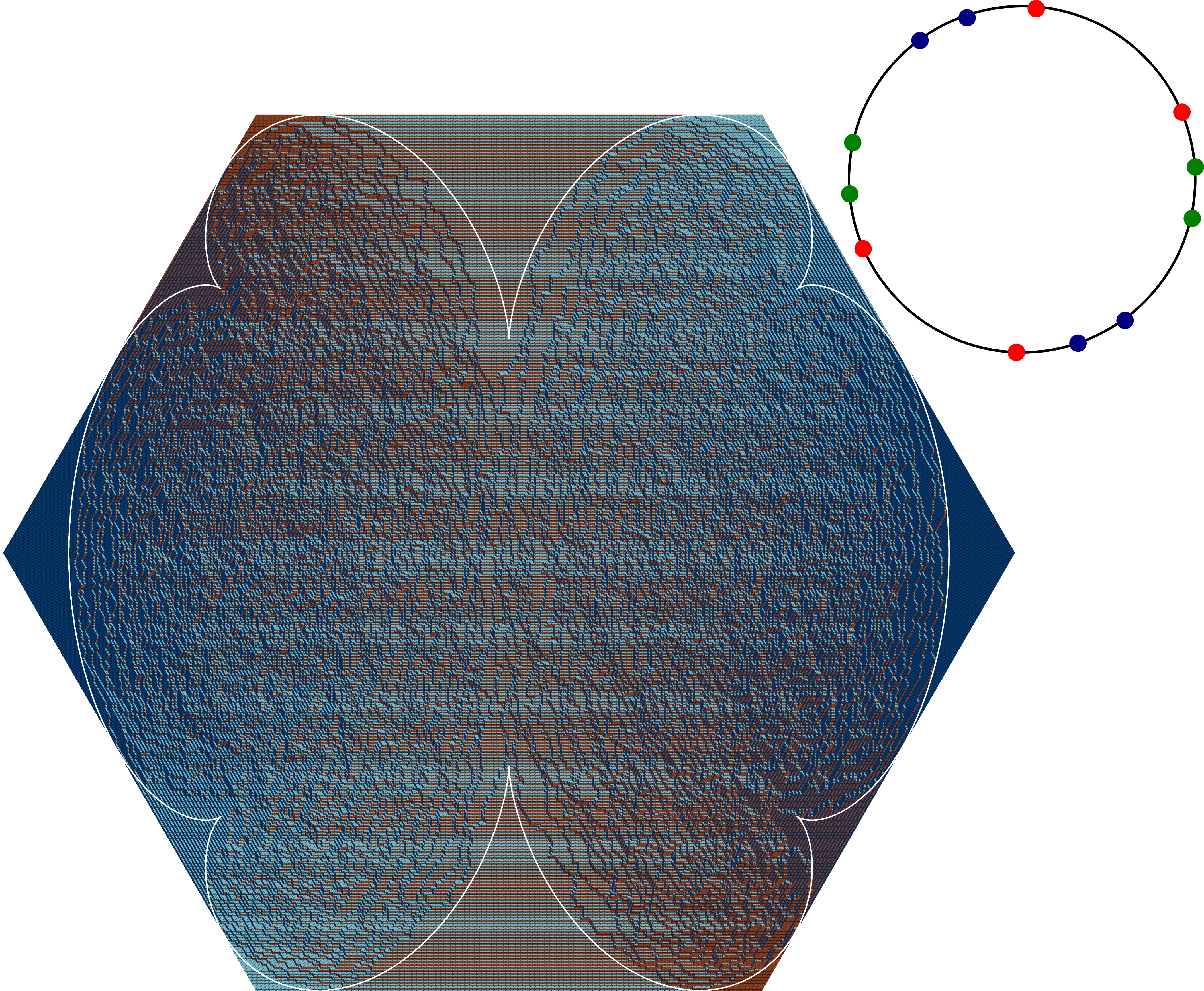}
    \end{subfigure}
    \caption{Two genus $0$ dimer configurations along with predicted arctic curves. The dimer models are populated with isoradial weights \eqref{eq:isoradial_weights} defined through the Harnack data depicted above them. All the involved differentials can be computed through rational functions. \textbf{Left}: Aztec diamond along with its spectral data $\cubr{\calR_+, \cubr{\alpha^\pm_i, \beta^\pm_j}}$ given for a $3\times 3$ domain with triple train tracks at the green points and simple ones at the red ones. \textbf{Right}: Hexagon with its spectral data $\cubr{\hat{\calR}_+, \cubr{\alpha^\pm_i, \beta^\pm_j, \gamma^\pm_l}}$ given for a $2 \times 2$ domain. We recall that this is already the double cover with branch point $0$ of the spectral curve $\calR$ that defines the weights.}
    \label{fig:genus_0_configurations}
\end{figure}

\section{Hexagonal Case}
\label{sec:07_hexagonal_case}

We now consider the case of the hexagon. For this let $G_N$ be a subset of the regular hexagonal lattice with side lengths $N,N,N$ and filled with copies of a fundamental domain of size $n \times n$ with doubly periodically repeating train track parameters as in \eqref{fig:hex_lattice} and embedded as described in Section~\ref{sec:02_TheDimerModel}.

Just as the Aztec diamond this is a classical and extensively studied problem. Dimer configurations on $G_N$ are in bijection to boxed plane partitions of size $N$ providing deep connections to combinatorics. The classical result computing the number of boxed plane partitions or equivalently the partition function for of the dimer model on $G_N$ with uniform weights was first obtained in \cite{MacMahon_1915}. In \cite{Cohn1998} a description of the limit shape under uniform weights was given. To our knowledge a (close to) complete description of limit shapes for (quasi) periodic weights on the hexagon does not exist to this date. We provide one in this section.

We begin with an informal discussion pointing out some differences between the case of the hexagon and the Aztec diamond. For simplicity let us consider uniform weights on the square lattice, i.e. $\mathcal{R}_+$ = $\mathbb{D}$ with the train track parameters being $4$th roots of unity. On the square lattice there are $4$ different frozen Gibbs measures corresponding to the $4$ different types of regular brickwork patterns or equivalently the four corners of the Newton polygon. We find one of these brickwork patterns in each of the $4$ corners of the Aztec diamond giving a one-to-one correspondence. 

On the hexagonal lattice with uniform weights there are $3$ frozen Gibbs measures that correspond to the $3$ different lozenges that exist and the $3$ corners of the Newton polygon (which is a triangle in this case). The Hexagon $G_N$ however, has $6$ corners with boundary conditions such that opposite corners should house the same kind of frozen region. Indeed as is seen in simulations and was first shown in \cite{Cohn1998} the arctic curve for this case is a circle, yielding $3$ pairs of frozen regions. Since the limit shape generically is $C^1$ as seen in \cite{DeSilva_minimizers_of_convex_functionals_2010}, in the neighborhood of these frozen regions $h$ must have a slope in the corresponding neighborhoods of the corners of the Newton polygon. Therefore one should not expect a one-to-one mapping between $\mathcal{R}_+$ and the liquid region. For uniform weights this turns out to be a two-to-one map and we will dedicate this section to constructing it for arbitrary quasi-periodic weights.

The challenge here is that starting with Harnack data $\mathcal{S} = \cubr{\mathcal{R}, \mathcal{T}}$ defining the weights one needs to find a double cover $\hat{\mathcal{R}}$ from which a one-to-one map to the liquid region can then be constructed. In full generality this is a hard problem and is further complicated for more general domains where a more complicated covering structure should be found \cite{Kenyon_Okounkov_2007}. 

Keeping in line with the philosophy of treating the reverse problem we put the double cover $\hat{\calR}$ at the center of our construction and start with it as given data. The map $\mathcal{F}$ to the dimer domain is then defined in the same way as for the Aztec diamond but now lives on $\hat{\calR}$ and maintains all the same properties. There are two separate cases to consider here. The case of ramified and unramified coverings. Both can indeed appear for some choice of weights. By exploiting the symmetries of the hexagon we present the corresponding Schottky models allowing for computation, see Section~\ref{sec:10_SchottkyUniformization}.

\subsection{Ramified Cover}

Let $\hat{\calR}$ be an M-curve of genus $\hat{g} = 2g$ with antiholomorphic involution $\tau$ fixing the real ovals $\hat{X} = \bigcup_{i=0}^{2g} \hat{X}_i$ that split $\hat{\calR}$ into the two halves $\hat{\calR}_+, \hat{\calR}_-$. Let $\pi$ be a holomorphic involution with two fixed points $P_0 \in \hat{\calR}_+, \tau P_0 \in \hat{\calR}_-$ which maps $\hat{X}_i$ onto $\hat{X}_{i+g}$ for $1 \leq i \leq g$ and maps $X_0$ onto itself. 
We have a double cover $p: \hat{\calR} \to \calR = \hat{\calR} / \pi$ with branch points $P_0, \tau P_0$.

Let now $\mathcal{S} = \cubr{\calR, \cubr{\alpha, \beta, \gamma}}$ be a Harnack data with differentials $d\zeta_1, d\zeta_2$ on  $\calR$ as in \eqref{eq:residues_of_dZeta_i_hex}. Their pullback to $ \hat{\calR}$ defines $\pi$-symmetric differentials. The preimage of every point in $\cubr{\alpha, \beta, \gamma}$ consists of two points on $\hat{X_0}$. We will denote these as $\cubr{\alpha^\pm, \beta^\pm, \gamma^\pm}$ where we order them such that on $\hat{X}_0$ we have the cyclic order
\begin{equation}
    \label{eq:cyclic_order_hex_double_cover}
    \alpha_i^- < \beta_j^- < \gamma^-_k < \alpha_i^+ < \beta_j^+ < \gamma^+_k < \alpha_i^-, \quad \forall i,j,k.
\end{equation}
The pullback of $\mathcal{S}$ under $p$ is the Harnack data $\hat{\mathcal{S}} := \cubr{\hat{\calR}, \cubr{\alpha^\pm, \beta^\pm, \gamma^\pm}}$. 

\begin{table}
    \begin{center}
        \begin{tabular}{ |c|c|c|c|c|c|c| } 
            \hline
            res & $\alpha^-_i$ & $\beta^-_i$ & $\gamma^-_i$ & $\alpha^+_i$ & $\beta^+_i$ & $\gamma^+_i$ \\
            \hline
            $d\zeta_1$ & 1 & 0 & -1 & 1 & 0 & -1 \\ 
            $d\zeta_2$ & 0 & 1 & -1 & 0 & 1 & -1 \\ 
            $d\zeta_3$ & 1 & -1 & 1 & -1 & 1 & -1 \\ 
            \hline
        \end{tabular}
    \end{center}
\caption{Residues of the differentials $d\zeta_i$ for the regular hexagon.}
\label{tab:xi_residues_hex}
\end{table}
\begin{definition}
    \label{def:hex_Harnack_data_admissible_ramified}
    We define $d\zeta_3$ to be the differential on $\hat{\calR}$ with residues as in Table \ref{tab:xi_residues_hex} and purely imaginary periods. We call the Harnack data $\hat{\mathcal{S}}$ \emph{admissible} if $d\zeta_3(P_0) = 0$.
\end{definition}

Note that on $\hat{\calR}$ we have the antisymmetry $\pi^* d\zeta_3 = -d\zeta_3$ and the symmetries $\pi^*d\zeta_i = d\zeta_i, i \leq 2$. Thus while $d\zeta_1, d\zeta_2$ are defined on $\calR$, the differential $d\zeta_3$ is well defined only on the double cover $\hat{\calR}$. 

The Fock weights that $G_N$ is populated with are determined purely from $\mathcal{S}$ and hence do not depend on any boundary condition considerations.

\begin{rem}
\label{rem:uniformization_of_hex_curve}
    As in Remark~\ref{rem:uniformization_calR} we fix a Schottky uniformization for Riemann surface $\hat{\calR}$ such that $\hat{\mathcal{R}}_+$ is mapped to the unit disk with $2g$ round disks removed so that $X_0$ is mapped to the unit circle and $P_0$ is mapped to $0$ so that the holomorphic involution is given by $\pi z = -z$. This uniformization data along with the clustered points on the unit circle is the actual Harnack data $\hat{\mathcal{S}}$.
 As in the case of the Aztec diamond all described objects can be numerically computed,  see Section~\ref{sec:10_SchottkyUniformization}.
\end{rem}

Henceforth we assume that $\hat{\mathcal{S}}$ is admissible. Analogously to \eqref{eq:def_dZeta} we consider the family of differentials
\begin{equation}
    \label{eq:def_dZeta_hex}
    d\zeta_{(u,v)} = d\zeta = -u d\zeta_2 + v d\zeta_1 + d\zeta_3.
\end{equation}

For any $(u,v) \in \mathbb{R}^2$ the differential $d\zeta$ has $6n$ poles on $\hat{\calR}$ and hence also has $6 n + 4g - 2$ zeros. Again due to being real on the real ovals we know that there are at least two zeros on each $\hat{X}_i$, $i\geq 1$, hence a total of $4g$ zeros. Furthermore there is one zero between any two train track points of the same type as they have the same residue, accounting for a total of $6(n-1)$ zeros. Note that $d\zeta_1(P_0) = d\zeta_2(P_0) = 0$ due to the symmetry and $d\zeta_3(P_0) = 0$ by admissibility and thus $d\zeta(P_0) = d\zeta(\tau P_0) = 0$. This again leaves us with one pair of free zeros. We define the liquid region $\mathcal{F}_{\hat{\mathcal{S}}} \subset \mathbb{R}^2$ to be the open set of $(u,v)$ such that the free zeros of $d\zeta_{(u,v)}$ are $P \in \hat{\calR}^\circ_+$ and $\tau P$. Then again in analogy to Proposition~\ref{prop:calF_is_diffeomorphism} we have
\begin{proposition}
    The map $\mathcal{F}: \hat{\calR}^\circ_+ \to \mathcal{F}_{\hat{\mathcal{S}}}$ is a diffeomorphism.
\end{proposition}

Note that we have an alternating sign structure of residues for $d\zeta_3$ and thus also for $d\zeta_{(u,v)}$ with
\[ (u,v) \in \mathcal{H} = \cubr{(u,v) \in [-1, 1]^2 \;|\; -1 \leq u + v \leq 1}. \]
For any $(u,v) \not\in \mathcal{H}$ the alternating sign structure is not preserved and we hence have that the two free zeros of $d\zeta$ are on $X_0$. Thus $\mathcal{F}_{\hat{\mathcal{S}}} \subset \mathcal{H}$.

The arguments from Sections~\ref{sec:03_CoordinateMap} and \ref{sec:04_GeometricProperties} can be directly applied to this setup with $\hat{\calR}$ now being the Riemann surface where $d\zeta$ is defined. The following results are immediate.
\begin{itemize}
    \item All marked train track points get mapped to the boundary of $\mathcal{H}$, see \eqref{eq:uv_at_train_tracks}.
    \item The formulas \eqref{eq:xy_on_real_ovals} for $(u,v)$ hold.
    \item The parallelity property from Proposition~\ref{prop:parallelity_prop} is preserved.
    \item We have a decomposition of $\mathcal{H}$ into $\mathcal{F}_{\hat{\mathcal{S}}}$ and frozen, quasi-frozen as well as gas regions like in \eqref{eq:decomposition_into_frozen_and_gas_regions}. The latter are now in two-to-one correspondence to the marked points on the Newton polygon $\Delta_\mathcal{S}$.
    \item Each gas, quasi-frozen and frozen region has $4, 1$ and $0$ cusps respectively.
    \item All zeros of $d\zeta$ on $\hat{\calR}$ are given by the tangential zeros and for $(u,v) \in \mathcal{F}_{\hat{\mathcal{S}}}$ a pair of conjugated points, see Corollary~\ref{cor:classification_of_all_dZeta_zeros}.
    \item The maps $\mathcal{A} \circ \mathcal{F}^{-1}$ and $J\Delta \circ \mathcal{F}^{-1}$ again define divergence-free fields on $\mathcal{F}_{\hat{\mathcal{S}}}$, see Proposition~\ref{prop:x_y_divergence_free}. Furthermore $\mathcal{A} \circ \mathcal{F}^{-1}$ defines a divergence-free field for the tangential zeros, see Proposition~\ref{prop:tangential_zeros_divergence_free}.
\end{itemize}

Like in \eqref{eq:def_of_complex_height_function} we again define the complex height function
\[ H(P) = H(u,v) := \frac{1}{\pi} \int_{\ell} d\zeta_{(u,v)} =: \frac{1}{\pi} g + ih \]
where $\ell$ is a path on $\hat{\calR}_+$ going from a fixed $Q\in \hat{X}_0$ to $P = \mathcal{F}^{-1}(u,v)$.

It then again follows immediately from the same arguments as in Sections~\ref{sec:04_GeometricProperties} and \ref{sec:05_SurfaceTensionMinimization} that
\begin{itemize}
    \item On $\mathcal{F}_{\hat{\mathcal{S}}}$ we have $\nabla h(u,v) = (s_1, s_2)$, $J\nabla g(u,v) = (x_1, x_2)$, see \eqref{eq:grads_of_f_and_g}.
    \item $h$ can be affinely extended to each component of $\mathcal{H} \setminus \mathcal{F}_{\hat{\mathcal{S}}}$ resulting in the extension $\hat{h} \in C^1(\mathcal{H})$ which satisfies the hexagon boundary conditions \eqref{eq:hex_boundary_height}, see Lemma~\ref{lem:extension_of_h}.
    \item $g$ can be extended to an admissible extension $\hat{g}: \mathcal{H} \to \mathbb{R}$, see Proposition~\ref{prop:generalized_Euler_Lagrange_suff_cond} and Definition~\ref{def:admissible_extension_and_boundary_condition}. This construction is the exact same as presented in Section~\ref{sec:05_SurfaceTensionMinimization} and still allows the same functional freedom, see Remark~\ref{rem:g_extension_non_unique}.
\end{itemize}

\begin{figure}[h]
\centering
\fontsize{10pt}{12pt}\selectfont
\def\svgwidth{\linewidth}
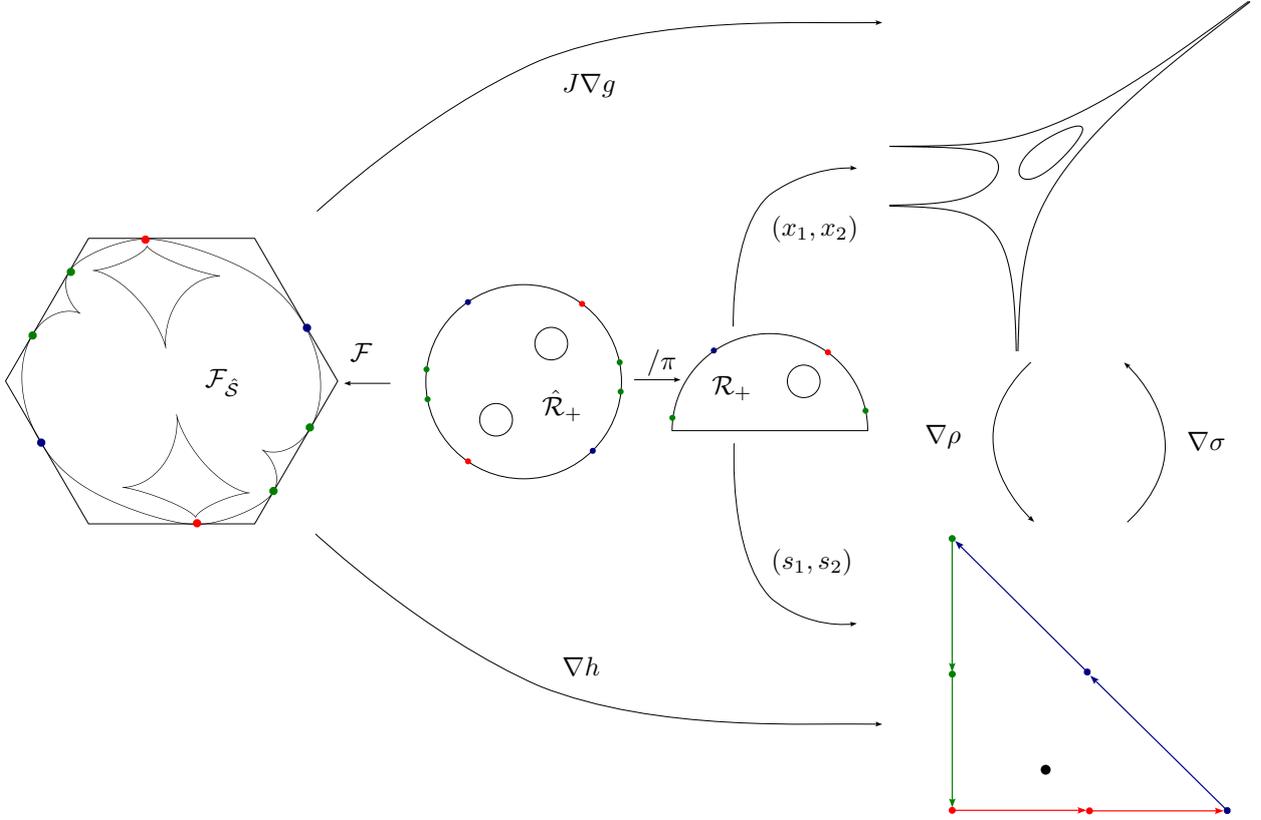
\caption{Double cover $\hat{\calR}_+$ with branch point at $z=0$ along with $\calR_+$, its factorization by $\pi z = -z$. The diffeomorphisms $\mathcal{F}$ is now defined on $\hat{\calR}_+$, while $\mathcal{A}, \Delta$ are maps from $\calR_+$. Their compositions given by gradients of the functions $g, h, \rho, \sigma$. Admissible angles $\alpha^\pm$ in green and $\beta^\pm$ in red and $\gamma^\pm$ in blue. This is an example with $n = 2$ and some angles being repeated thus leading to poles of order 2. Note that $\calR$ that defines the weights is of genus $1$ but due to the double cover there are two gas bubbles that appear in $\mathcal{F}_{\hat{\mathcal{S}}}$. We have applied a linear transformation to the dimer domain $\mathcal{H}$ to obtain the regular hexagon shape. All pictures here are real results of computation and not just abstract representations. In particular the $z$ coordinate in the picture of $\hat{\mathcal{R}}_+$ defines the conformal structure of the Riemann surface.}
\label{fig:all_maps_and_zeros_hex}
\end{figure}

Finally we obtain that the constructed extended functions $h, \hat{g}$ are indeed the minimizers of the surface tension and magnetic tension functionals. In particular $h$ is the limiting dimer height function on the hexagon.
\begin{theorem}
    \label{thm:h_is_height_fct_hex}
    Let $\mathcal{S} = \cubr{\calR, \mathcal{T}}$ be Harnack data with point $P_0 \in \calR_+$ such that $d\zeta_3$ on the double cover $\hat{\calR}$ which is two copies of $\calR$ glued along $[P_0, \tau P_0]$ satisfies $d\zeta_3(P_0) = 0$. We consider the dimer model on the hexagon with weights defined by $\mathcal{S}$.

    Let $h_b$ be the hexagon boundary conditions. Then the function $\hat{h}$ is the minimizer of
    \[ I_\sigma(\varphi) = \int_{\mathcal{H}} \sigma(\nabla \varphi) \]
    over the space $\Lip_{\Delta_\mathcal{S}}(\mathcal{H}, h_b)$ and thus the limiting dimer height function on the hexagon.

    Furtermore if $\hat{g}:\mathcal{H} \to \mathbb{R}$ is an admissible extension of $g : \mathcal{F}_{\hat{\mathcal{S}}} \to \mathbb{R}$ with boundary conditions $g_b$, then $\hat{g}$ is a minimizer of the \emph{magnetic tension functional}
    \[ I_\rho(\varphi) = \int_{\mathcal{H}} \rho(J \nabla \varphi) \]
    over the Sobolev space $W^{1,2}(\mathcal{H}, \mathbb{R})$ restricted to functions with boundary conditions $g_b$.
\end{theorem}

\begin{rem}
\label{rem:skewed_hexagon}
    For simplicity we provided an exposition for the regular hexagon here. This approach, however, extends to all hexagons of the form
    \[ \mathcal{H} = \cubr{(u,v) \;|\; |u| \leq a, |v| \leq b, |u + v| \leq c }. \]
    The differential $d\zeta_3$ then has the residues
    \[ \res_{\alpha^-_i} d\zeta_3 = -\res_{\alpha^+_i} d\zeta_3 = a, \; \res_{\beta^+_i} d\zeta_3 = -\res_{\beta^-_j} d\zeta_3 = b, \; \res_{\gamma^-_k} d\zeta_3 = -\res_{\gamma^+_k} d\zeta_3 = c. \]
    See Figure~\ref{fig:configurations_on_skewed_hex} for an example.
\end{rem}

\begin{figure}
\centering
\begin{subfigure}[p]{.49\textwidth}
    \centering
    \includegraphics[width=\linewidth]{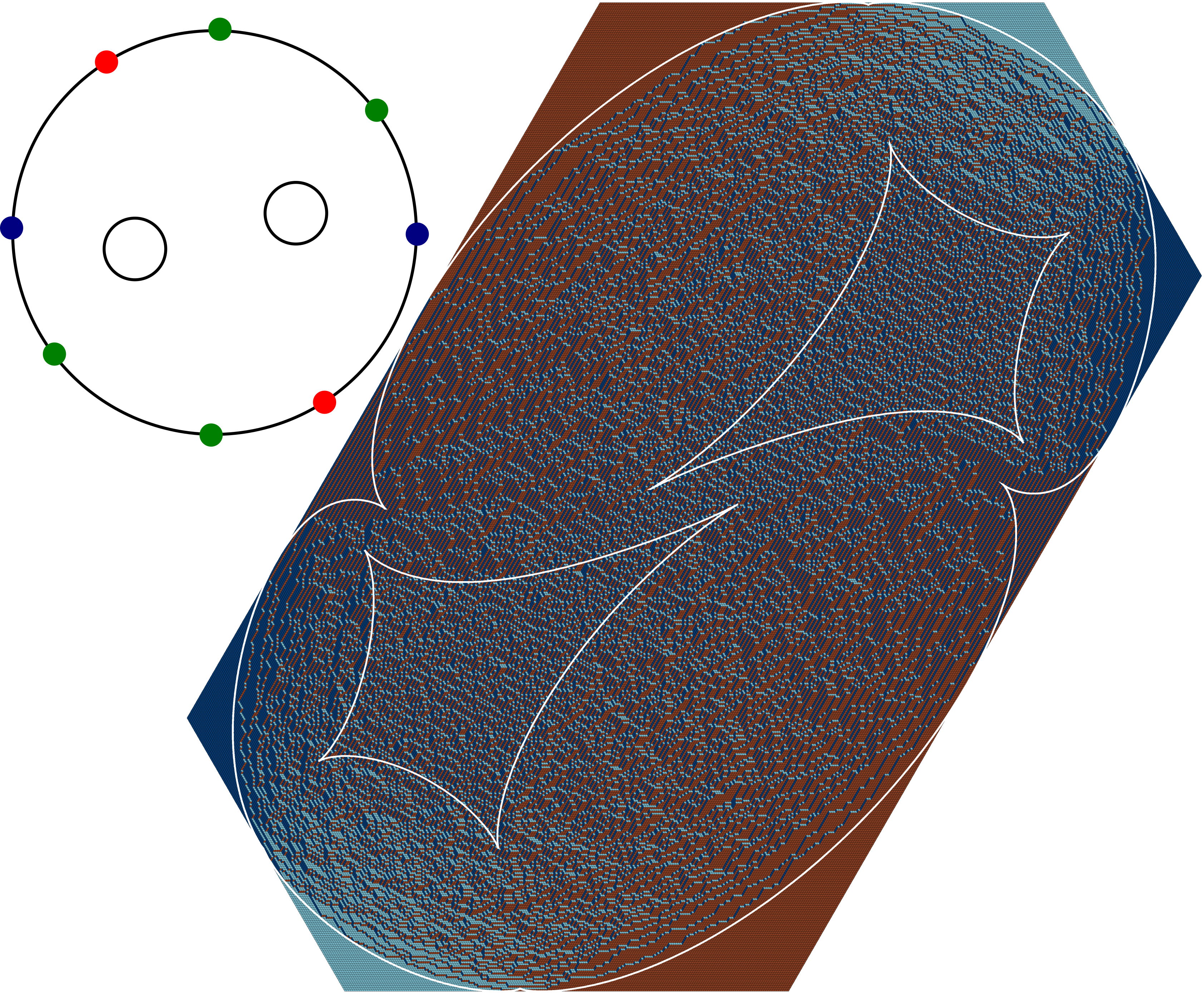}
\end{subfigure}
\begin{subfigure}[p]{.49\textwidth}
    \centering
    \includegraphics[width=\linewidth]{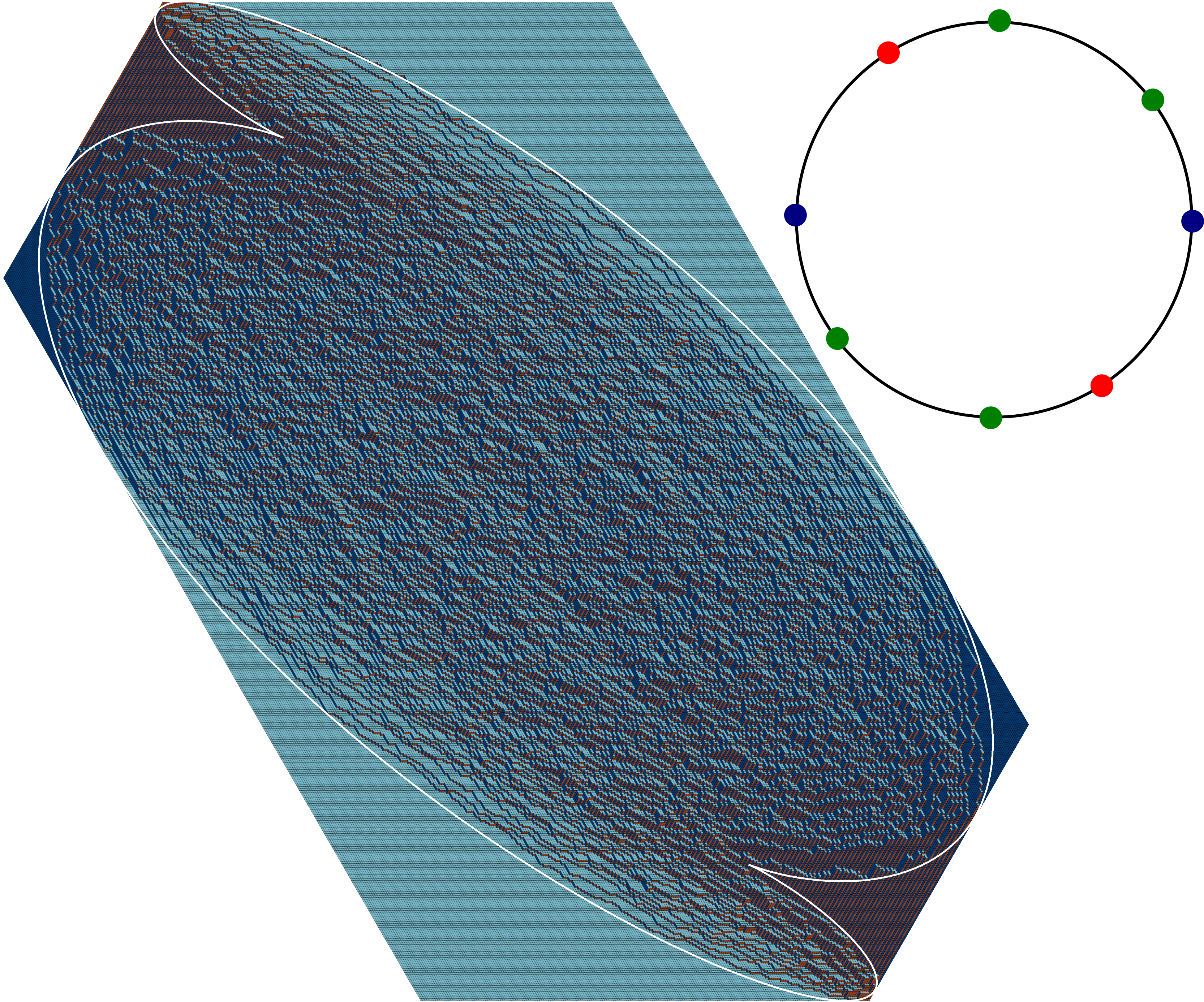}
\end{subfigure}
\caption{Two samples on different skewed hexagon domains along with calculated arctic curves. \textbf{Left}: $\hat{\calR}$ has genus $\hat{g} = 2$ and there are double poles at $\alpha^\pm$, $\beta^\pm$. \textbf{Right}: $\hat{\calR}$ has genus $\hat{g} = 0$ and there are double poles at $\alpha^\pm$, $\beta^\pm$.}
\label{fig:configurations_on_skewed_hex}
\end{figure}

\subsection{Unramified Cover}

Let $\hat{\calR}$ be an M-curve of genus $\hat{g} = 2g - 1$, $g \geq 1$ with anti-holomorphic involution $\tau$ fixing the set of real ovals $\hat{X} = \hat{X}_0 \cup \hat{X}_1 \cup \ldots \cup \hat{X}_{\hat{g}}$ which separates $\hat{\calR}$ into connected components $\hat{\calR}_+, \hat{\calR}_-$. Further let $\pi$ be a holomorphic involution that maps $\hat{X}_0$ to $\hat{X}_0$, $\hat{X}_1$ to $\hat{X}_1$ and $\hat{X}_n$ to $\hat{X}_{n+g}, \; 2\leq n\leq g$. 
We have an unramified double cover $p: \hat{\calR} \to \calR = \hat{\calR} / \pi$ of the M-curve $\calR$ of genus $g$ with real ovals $X_k=p(\hat{X}_k), k=0, \ldots, g$ and ramification cycles $X_0, X_1$.

We denote the real ovals of $\calR $ such that $p^{-1}X_0 = \hat{X}_0$, $p^{-1}X_1 = \hat{X}_1$ and $p^{-1}X_n = \hat{X}_n \cup \hat{X}_{n+g}$, $2\leq n\leq g$. $\hat{\calR}$ is an unramified covering of $\calR$ with ramified cycles $X_0, X_1$.

We assume that $\mathcal{S} = \cubr{\calR, \cubr{\alpha_i, \beta_j, \gamma_l}}$ with $i,j,l \leq n$ is a Harnack data, that is that the chosen train track angles satisfy the clustering condition \eqref{eq:clustering_hex} and denote their preimages as $\cubr{\alpha_i^\pm, \beta_j^\pm, \gamma_l^\pm} \subset \hat{X}_0$ such that they again satisfy \eqref{eq:cyclic_order_hex_double_cover}.

\begin{definition}
    \label{def:admissible_harnack_data_unramified}
    We call the Harnack data $\hat{\mathcal{S}} = \cubr{\hat{\calR}, \cubr{\alpha_i^\pm, \beta_j^\pm, \gamma_l^\pm}}$ admissible if there exist three distinct pairs of points $(R_i, \pi R_i), i=1,2,3$ on $\hat{X}_1$ such that $d\zeta_3$ defined as the meromorphic differential with imaginary periods and residues as in Table~\ref{tab:xi_residues_hex} satisfies
    \[ \int_{R_i}^{\pi R_i} d\zeta_3 = 0. \]
\end{definition}

For admissible $\hat{\mathcal{S}}$ we again consider the family of differentials 
\[ d\zeta = d\zeta_{(u,v)} = -u d\zeta_2 + v d\zeta_1 + d\zeta_3. \]

\begin{rem}
\label{rem:uniformization_calR_unramified}
    We choose a convenient uniformization by mapping $\hat{\calR}_+$ to $\mathbb{D}$ with a removed central round disk bounded by $\hat{X}_1$ and centrally symmetric pairs of round circles $\hat{X}_n = -\hat{X}_{n+g}$, $2\leq n\leq g$. See Remarks~\ref{rem:uniformization_calR} and \ref{rem:uniformization_of_hex_curve}.

    Under this uniformization we have $\pi z = -z$ and thus centrally symmetric train track parameters $\hat{\alpha}_i^- = -\hat{\alpha}_i^+$. We fix the mathematically positive orientation on $\hat{\calR}_+$, thus $\hat{X}_1$ is oriented clockwise and $\partial G_1$ counterclockwise. For considerations on $\hat{X}_1 = r S^1$ we fix the polar coordinate $z = r e^{-i\theta}$, $\theta \in [0, 2\pi)$. We then have
    \[ d\zeta(z) = f(z) dz = -i z f(z) d\theta =: \tilde{f}(\theta) d\theta \]
with real-valued $\tilde{f}$ on $\hat{X}_1$. Recall that $d\zeta_3 = d\zeta_{(0,0)}$ and we denote $\tilde{f}_3$ as the corresponding function on $\hat{X}_1$.
\end{rem}

The critical points of the integral $\int_Q^{\pi Q} d\zeta_3$ are exactly the zeros of $d\zeta_3$:

\begin{equation}
\label{eq:dZeta_3_zeros_are_critical_points_of_ell_G1}
    \frac{d}{d \theta} \br{\int_\theta^{\theta + \pi} \tilde{f}_3(\theta) d\theta} = \tilde{f}_3(\theta + \pi) - \tilde{f}_3(\theta) = -2\tilde{f}_3(\theta).
\end{equation}
Hence we obtain
\begin{corollary}
    \label{cor:admissible_dZeta3_has_6_zeros}
    For admissible data, $d\zeta_3$ has $6$ zeros on $\hat{X}_1$.
\end{corollary}
\begin{proof}
    It must have at least $6$ zeros by Definition~\ref{def:admissible_harnack_data_unramified} but also no more than $6$ due to a simple counting of poles.
\end{proof}

\begin{lemma}
    \label{lem:admissible_unramified_exists}
    The set of admissible Harnack data is not empty for any genus $g$.
\end{lemma}
\begin{proof}
    Consider the case of $\hat{g} = 1$ with the uniformization from Remark~\ref{rem:uniformization_calR_unramified}. Pick symmetric angles 
    \[ \hat{\alpha}^\pm_i = \pm 1, \; \hat{\beta}^\pm_j = \pm e^{i\frac{\pi}{3}}, \hat{\gamma}^\pm_l = \pm e^{i\frac{2\pi}{3}}. \]
    Then $d\zeta_3$ has the symmetry
    \[ -d\zeta_3(z) = d\zeta_3(e^{i \frac{\pi}{3}} z). \]
    We thus have $3$ points $R_i \in \hat{X}_1$ with
    \[ \int_{R_1}^{\pi R_1} d\zeta_3 = \int_{R_2}^{\pi R_2} d\zeta_3 = \int_{R_3}^{\pi R_3} d\zeta_3 = 0.\]
    Small perturbations of the data preserve this property. This includes perturbations adding small holes in $\hat{\calR}_+$ thus extending existence to $g>1$.
\end{proof}
Definition~\ref{def:admissible_harnack_data_unramified} is an open condition and thus minor perturbations of admissible data remain admissible. See subsection~\ref{subsec:parameter_count_hex} for a discussion on parameter counts.

We now work towards a definition of the KO-map in the unramified case.
\begin{definition}
    \label{def:regions_of_hex_unramified}
    We say that $(u,v) \in \mathcal{H}$ is in
    \begin{itemize}
        \item the liquid region $\mathcal{F}_{\hat{\mathcal{S}}}$ if $d\zeta$ has $4$ distinct zeros on $\hat{X}_1$ and a complex zero $P \in \hat{\calR}_+^\circ$.
        \item a frozen region $F_\ell$ if $d\zeta$ has $4$ distinct zeros on $\hat{X}_1$ and two zeros on the arc $\ell \subset \hat{X}_0$ bounded by two neighboring clusters of train track angles. 
        \item a quasi-frozen region $Q_\ell$ if $d\zeta$ has $4$ distinct zeros on $\hat{X}_1$ and $3$ zeros on the arc $\ell \subset \hat{X}_0$ that is bounded by two angles of the same type. 
        \item the gas region $G_n,\; n \geq 2$ if $d\zeta$ has $4$ zeros on $\hat{X}_1$ and $4$ distinct zeros on $\hat{X}_n$. 
        \item the gas region $G_1$ if $d\zeta$ has $6$ distinct zeros on $\hat{X}_1$. 
    \end{itemize}
    We denote by $\mathscr{F}, \mathscr{Q}, \mathscr{G}$ the sets of all frozen, quasi-frozen and gas regions respectively.
\end{definition}
Because we require all zeros to be distinct in these definitions, all these regions are open sets.

In general $d\zeta$ has $6 n + 2\hat{g} - 2$ zeros of which $2\hat{g} - 2$ are located in distinct pairs on $\hat{X}_2, \ldots, \hat{X}_{\hat{g}}$. Furthermore the residue of $d\zeta$ at angles within each cluster is equal and thus there is a zero between any two points of the same cluster on $\hat{X}_0$ for a total of $6(n-1)$ distinct zeros. The remaining $6$ zeros determine the region we are in.

The arctic curves are then defined as follows:
\begin{itemize}
    \item Let $F \in \mathscr{F} \cup \mathscr{Q}$ be a frozen or quasi-frozen region corresponding to the arc $\ell$. The curve $\overline{\mathcal{F}_{\hat{\mathcal{S}}}} \cap \overline{F}$ is characterized by the fact that $d\zeta$ has $4$ distinct zeros on $\hat{X}_1$ and a \emph{double zero} on $\ell$.
    \item For $(u,v) \in \overline{\mathcal{F}_{\hat{\mathcal{S}}}} \cap \overline{G_n}$, $n\geq 2$ the differential $d\zeta$ has $4$ distinct zeros on $\hat{X}_1$ and a \emph{double zero} on $\hat{X}_n$.
    \item For $(u,v) \in \overline{\mathcal{F}_{\hat{\mathcal{S}}}} \cap \overline{G_1}$ the differential $d\zeta$ has a \emph{double zero} on $\hat{X}_1$ as $4$ more distinct zeros on $\hat{X}_1$.
\end{itemize}
\begin{proposition}
    \label{prop:KO_map_injective_unramified}
    The map $\Omega: \mathcal{F}_{\hat{\mathcal{S}}} \to \hat{\calR}_+^\circ$ that maps $(u,v)$ to the unique zero $P\in \hat{\calR}_+^\circ$ of $d\zeta_{(u,v)}$ is a diffeomorphism.
\end{proposition}
\begin{proof}
To show is that $\Omega$ is injective.
    Let $d\zeta_{(u_1, v_1)}$ and $d\zeta_{(u_2, v_2)}$ have the same complex zero $P \in \hat{\calR}_+^\circ$. Then $d\xi := d\zeta_{(u_1, v_1)} - d\zeta_{(u_2, v_2)}$ is an even differential vanishing at P, i.e. 
    \[ d\xi(P) = 0, \; \pi^*d\xi = d\xi.\]
    Therefore it also vanishes at $\tau P, \pi P$ and $\tau \pi P$ and thus has $4$ complex zeros. Note that $d\xi$ is defined on $\calR$. It must have at least $2$ zeros on $X_1$ and thus has at least $4$ zeros on $\hat{X}_1$ due to the covering structure. It also has $2$ zeros on each $\hat{X}_n$, $n\geq 2$ and since the signs of residues of the angle clusters are not alternating there at least $2$ zeros on $\hat{X}_0$ between clusters. Further we have $6(n-1)$ zeros coming from within the clusters. This yields a total of 
    \[2 (\hat{g} - 1) + 6(n - 1) + 4 + 4 = 2 \hat{g} + 6n + 2 \] 
    zeros which is $2$ too many as $d\xi$ has $6n$ poles giving us a contradiction.
\end{proof}

Repeating the same arguments as in Section~\ref{sec:04_GeometricProperties} we obtain the parallelity property like Proposition~\ref{prop:parallelity_prop}.
\begin{corollary}
    \label{cor:smooth+parallelity_unramified}
    The arctic curves are smooth everywhere but at finitely many points. The tangent line at $(u,v)$ on an arctic curve is parallel to the tangent at the point $\mathcal{A}(Q)$ of the amoeba where $Q$ is the double zero of $d\zeta_{(u,v)}$. 
\end{corollary}

\begin{corollary}
    \label{cor:6cusps_unramified}
    The arctic curve $\partial G_1$ has $6$ cusps. For every other gas region the arctic curve $\partial G_n$, $n \geq 2$ has $4$ cusps.
\end{corollary}
\begin{proof}
    We use an argument similar to \cite{berggren_geometry_2023}. Due to $\hat{X}_1$ covering $X_1$ twice and the parallelity property Corollary~\ref{cor:smooth+parallelity_unramified}, going around $\partial G_1$ once corresponds to going around the oval $\mathcal{A}_1$ of the amoeba twice with opposite orientations. 

    Along this path the tangential vector to the amoeba rotates by $4\pi$. The tangential vector to $\partial G_1$ rotates by 
    \[ -2\pi = 4\pi - \pi C \]
    where $C$ is the number of cusps along the way. Therefore there are $6$ cusps. Similarly we only wind once around all other ovals $X_n$, $n \geq 2$ and therefore have $4$ cusps.
\end{proof}

\begin{corollary}
    \label{cor:outside_G1_4_zeros_on_X1_ramified}
    For $(u,v) \in \mathbb{R}^2\setminus \overline{G}_1$ the differential $d\zeta$ has exactly $4$ zeros on $\hat{X}_1$.
\end{corollary}
\begin{proof}
    This follows immediately because Proposition~\ref{prop:zeros_of_dZeta_tangential} still applies and Corollary~\ref{cor:6cusps_unramified}. 
\end{proof}

After these considerations it is clear that we have the same geometric picture as in the ramified case except now there is the central gas bubble $G_1$ that has $6$ cusps. Finally we obtain
\begin{theorem}
    \label{thm:regions_disjoint_unramified}
    The frozen, quasi-frozen, liquid and gas regions are disjoint. Furthermore their closures cover the entire hexagon
    \begin{equation}
        \mathcal{H} = \overline{\mathcal{F}_{\hat{\mathcal{S}}}} \cup \bigcup_{F \in \mathscr{F}} \overline{F} \cup \bigcup_{Q \in \mathscr{Q}} \overline{Q} \cup \bigcup_{i=1}^{\hat{g}} \overline{G}_i.
    \end{equation}
\end{theorem}

After having established that the diffeomorphism $\mathcal{F}$ is well defined and decomposes $\mathcal{H}$ into the same regions as before we can again define the complex height function 
\[ H(P) = H(u,v) := \int^P d\zeta_{(u,v)} =: \frac{1}{\pi} g + ih. \]
The function $h$ can again be extended affinely to each component to a function $\hat{h} \in C^1(\mathcal{H})$ and satisfies the hexagon boundary conditions. All the same arguments yield that $g$ can be extended to all of $\mathcal{H}\setminus G_1$ in an admissible way. The central gas bubble $G_1$ has $6$ cusps and an admissible extension of $g$ to it is not covered by the arguments in Section~\ref{sec:05_SurfaceTensionMinimization}. We construct such an extension $\hat{g}$ in subsection~\ref{subsec:extension_to_G1} thus again obtaining our main theorem.
\begin{theorem}
    \label{thm:h_is_height_fct_hex_unramified}
    Let $\hat{S} = \cubr{\hat{\mathcal{R}}, \cubr{\alpha_i^\pm, \beta_j^\pm, \gamma_l^\pm}}$ be admissible Harnack data with $\mathcal{S}$ its image under the unramified double covering $p$. Consider the dimer model on the hexagon with weights defined through $\mathcal{S}$.

    Let $h_b$ be the hexagon boundary conditions. Then the function $h$ is the minimizer of
    \[ I_\sigma(\varphi) = \int_{\mathcal{H}} \sigma(\nabla \varphi) \]
    over the space $\Lip_{\Delta_\mathcal{S}}(\mathcal{H}, h_b)$ and thus the limiting dimer height function on the hexagon.

    Furtermore if $\hat{g}:\mathcal{H} \to \mathbb{R}$ is an admissible extension of $g : \mathcal{F}_{\hat{\mathcal{S}}} \to \mathbb{R}$ with boundary conditions $g_b$, then $\hat{g}$ is a minimizer of the \emph{magnetic tension functional}
    \[ I_\rho(\varphi) = \int_{\mathcal{H}} \rho(J \nabla \varphi) \]
    over the Sobolev space $W^{1,2}(\mathcal{H}, \mathbb{R})$ restricted to functions with boundary conditions $g_b$.
\end{theorem}

\subsection{Extension to the Central Gas Bubble}
\label{subsec:extension_to_G1}
We now construct an admissible extension $\hat{g}$ of $g$ to $G_1$ thus completing the proof of Theorem~\ref{thm:h_is_height_fct_hex_unramified}.

Without loss of generality we impose an order on the $R_j = r e^{-i\theta_j}$:
\[ \theta_1 < \theta_2 < \theta_3 < \theta_1 + \pi < \theta_2 + \pi < \theta_3 + \pi. \]
Let us introduce the function 
\begin{equation}
\label{eq:ell}
    \ell(\alpha) = \int_\alpha^{\alpha + \pi} \tilde{f}_3 d\theta.
\end{equation} 
The admissibility implies $\ell(\theta_i) = \ell(\theta_i + \pi) = 0$. Furthermore by \eqref{eq:dZeta_3_zeros_are_critical_points_of_ell_G1} and Corollary~\ref{cor:admissible_dZeta3_has_6_zeros} all of these are simple zeros and thus $\ell$ has alternating signs on neighboring segments $(\theta_{i-1}, \theta_i), (\theta_i, \theta_{i+1})$. Let us assume that $\ell$ is positive on $(\theta_1, \theta_2)$.

\begin{proposition}
\label{prop:points_R_i_between_cusps}
    We let $r_i := \mathcal{F}(R_i)$ denote the points corresponding to $R_i$ in dimer coordinates.
    The $6$ points $r_1, r_2, r_3, - r_1, - r_2, - r_3$ each lie on separate components bounded by the $6$ cusps on $\partial G_1$. There is exactly one such point on each component. Furthermore $r_1$ lies on a positive component, that is $\tilde{f}_{r_1}(R_1) \geq 0$.
\end{proposition}
\begin{proof}
    Note that for small $\epsilon>0$ we have
    \[ \ell(\theta_1) = 0, \ell(\theta_1 + \epsilon) > 0 \]
    and therefore $\tilde{f}_{(0,0)}(R_1) < 0$. Consider now the values of $d\zeta_q$ as $q$ gets deformed from $(0,0)$ to $r_1$. In $d\zeta_{r_1}$ the two zeros adjacent to $r_1$ get collapsed resulting in a double zero. This means that the interval of negative values around $R_1$ gets collapsed and we have $\tilde{f}_{r_1}(R_1) > 0$, see Figure~\ref{fig:central_gas_bubble_tangential_zeros_and_signs}. But we also know from Lemma~\ref{lemma:signs_of_fq(Q)} that $d\zeta_q(Q)$ switches its sign at every cusp. The claim follows.
\end{proof}

\begin{figure}[h]
\centering
\begin{subfigure}{.33\textwidth}
    \centering
    \fontsize{10pt}{12pt}\selectfont
    \def\svgwidth{\linewidth}
    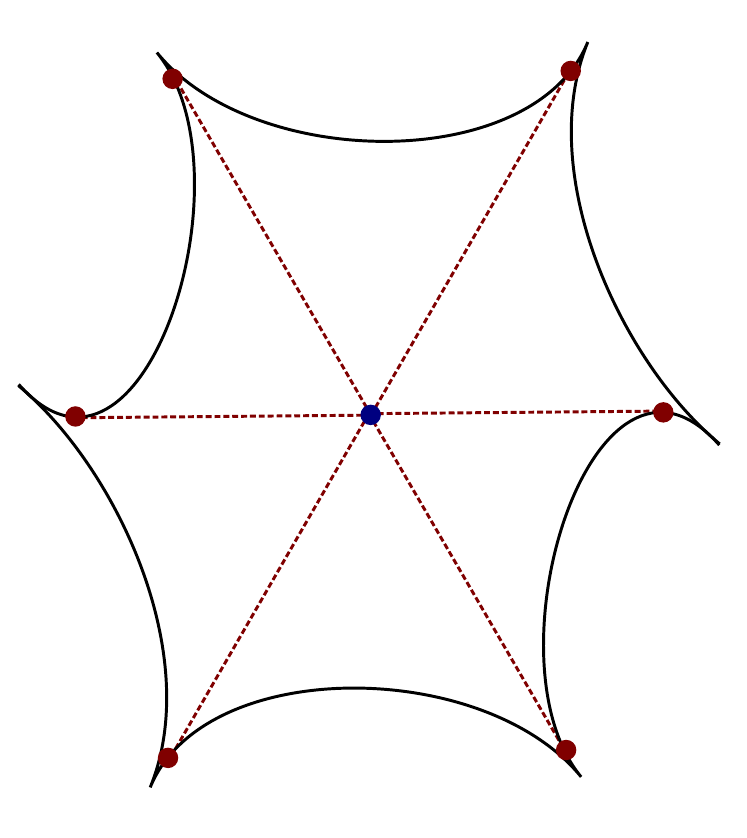
\end{subfigure}
\begin{subfigure}{.15\textwidth}
    
\end{subfigure}
\begin{subfigure}{.33\textwidth}
    \centering
    \fontsize{10pt}{12pt}\selectfont
    \def\svgwidth{\linewidth}
    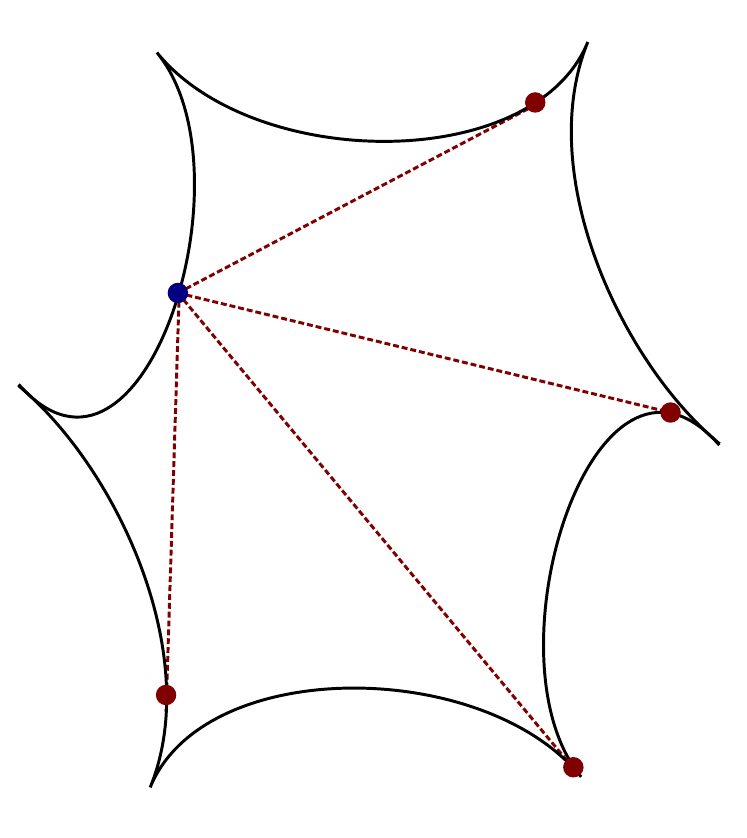
\end{subfigure}
\caption{The central symmetric gas bubble along with signs of $\tilde{f}_q$ in blue and zeros in red. \textbf{Left}: The case of $q = (0,0)$. The zeros are symmetric. \textbf{Right}: As we deform $q$ to $r_1$, the negative region on $[c_6, c_1]$ collapses resulting in a positive arc of $\tilde{f}_{r_1}$ around $R_1$. In particular $\tilde{f}_{r_1}(R_1) > 0$.}
\label{fig:central_gas_bubble_tangential_zeros_and_signs}
\end{figure}

Let the cusps $c_1, \ldots, c_6$ be ordered in mathematically positive orientation around $\partial G_1$ and $C_1, \ldots, C_6$ be the corresponding points on $\hat{X}_1$. Recall the notation $[P,Q]$ from Section~\ref{sec:05_SurfaceTensionMinimization} denoting either the negatively oriented arc on $\hat{X}_1$ between $P,Q \in \hat{X}_1$ or $[p,q]$ the corresponding positively oriented arc on $\partial G_1$ between $p,q \in \partial G_1$. We choose the indices of $C_i$ such that $r_1$ lies on $[c_6, c_1]$. See Figure~\ref{fig:central_gas_bubble_tangential_zeros_and_signs}. For every point $p = \mathcal{F}(P)$ on $[c_6, c_1]$ we consider a ray $r$ tangential to $\partial G_1$ at $p$ oriented away from $c_1$. Let it intersect $\partial G_1$ at another point $q = \mathcal{F}(Q)$. Again we define the extension as
\[ g_1(z) = g(P) + \pobr{\nabla g(P), z-p}, \]
and according to \eqref{eq:g-g_1_as_integral} we have
\[ g_1(Q) - g(Q) = \int_P^Q d\zeta_q. \]
There are 4 cases to consider and we denote the corresponding points in those cases as $p_i, q_i \in \partial G_1$ and $P_i, Q_i \in \hat{X}_1$, see Figure~\ref{fig:central_hex_bubble_tangents_4_cases}.

\textbf{Case 1}: $q_1$ lies on $[c_5, c_6]$. Then $\tilde{f}_{q_1}$ has no zeros on $(P_1, Q_1)$ and is negative on all of $[P_1,Q_1]$ and thus $\int_P^Q d\zeta_q \geq 0$.

\textbf{Case 2}: $q_2$ lies on $[c_4, c_5]$. We decompose our integral
\begin{equation}
\label{eq:decomposition_of_int_P^QdZeta_q}
    \int_{P_2}^{Q_2} d\zeta_{q_2} = \int_{P_2}^{C_1} d\zeta_{q_2} + \int_{C_1}^{\pi C_1} d\zeta_{q_2} + \int_{\pi C_1}^{Q_2} d\zeta_{q_2} =
                                 \int_{P_2}^{C_1} d\zeta_{q_2} + \ell(C_1) + \int_{\pi C_1}^{Q_2} d\zeta_{q_2},
\end{equation}
with $\ell$ given by \eqref{eq:ell}. All the terms in \eqref{eq:decomposition_of_int_P^QdZeta_q} are non-negative. Indeed, $\ell(C_1) \geq 0$ since $c_1\in [r_1, r_2]$. Furthermore $\tilde{f}_{q_2}$ is positive around $Q_2$ and has one zero on $(Q_2, P_2)$. Thus it is positive on $[P_2, C_1]$, and thus $\int_{P_2}^{C_1} d\zeta_{q_2} \geq 0$. Similarly $\int_{\pi C_1}^{Q_2} d\zeta_{q_2} \geq 0$ as $\tilde{f}_{q_2}$ has no zeros on $[\pi C_1, Q_2]$ and is positive there.

\textbf{Case 3}: $q_3$ lies on $[c_3, c_4]$. We use the same decomposition
$$
    \int_{P_3}^{Q_3} d\zeta_{q_3} = \int_{P_3}^{C_1} d\zeta_{q_3} + \ell(C_1) + \int_{\pi C_1}^{Q_3} d\zeta_{q_3}.
                                  $$
Now $\tilde{f}_{q_3}$ is positive on $[P_3, C_1]$ and negative on $[Q_3, \pi C_1]$. The orientation of the path $[Q_3, \pi C_1]$ is negative and thus all summands are again non-negative.

\textbf{Case 4}: $q_4$ lies on $[c_2, c_3]$. In this case $d\zeta_{q_4}$ has no zeros on $(P_4, Q_4)$. Furthermore $d\zeta_{q_4}(Q_4) \geq 0$ and hence
\[ \int_{P_4}^{Q_4} d\zeta_{q_4} \geq 0. \]

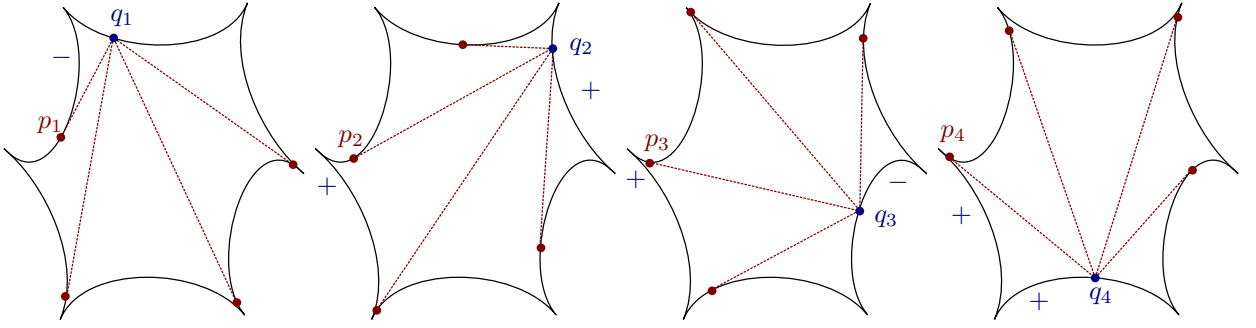
\begin{figure}[h]
\centering
\begin{subfigure}{.24\textwidth}
    \centering
    \fontsize{10pt}{12pt}\selectfont
    \def\svgwidth{\linewidth}
\begingroup%
  \makeatletter%
  \providecommand\color[2][]{%
    \errmessage{(Inkscape) Color is used for the text in Inkscape, but the package 'color.sty' is not loaded}%
    \renewcommand\color[2][]{}%
  }%
  \providecommand\transparent[1]{%
    \errmessage{(Inkscape) Transparency is used (non-zero) for the text in Inkscape, but the package 'transparent.sty' is not loaded}%
    \renewcommand\transparent[1]{}%
  }%
  \providecommand\rotatebox[2]{#2}%
  \newcommand*\fsize{\dimexpr\f@size pt\relax}%
  \newcommand*\lineheight[1]{\fontsize{\fsize}{#1\fsize}\selectfont}%
  \ifx\svgwidth\undefined%
    \setlength{\unitlength}{337.51127243bp}%
    \ifx\svgscale\undefined%
      \relax%
    \else%
      \setlength{\unitlength}{\unitlength * \real{\svgscale}}%
    \fi%
  \else%
    \setlength{\unitlength}{\svgwidth}%
  \fi%
  \global\let\svgwidth\undefined%
  \global\let\svgscale\undefined%
  \makeatother%
  \begin{picture}(1,1.06163133)%
    \lineheight{1}%
    \setlength\tabcolsep{0pt}%
    \put(0,0){\includegraphics[width=\unitlength,page=1]{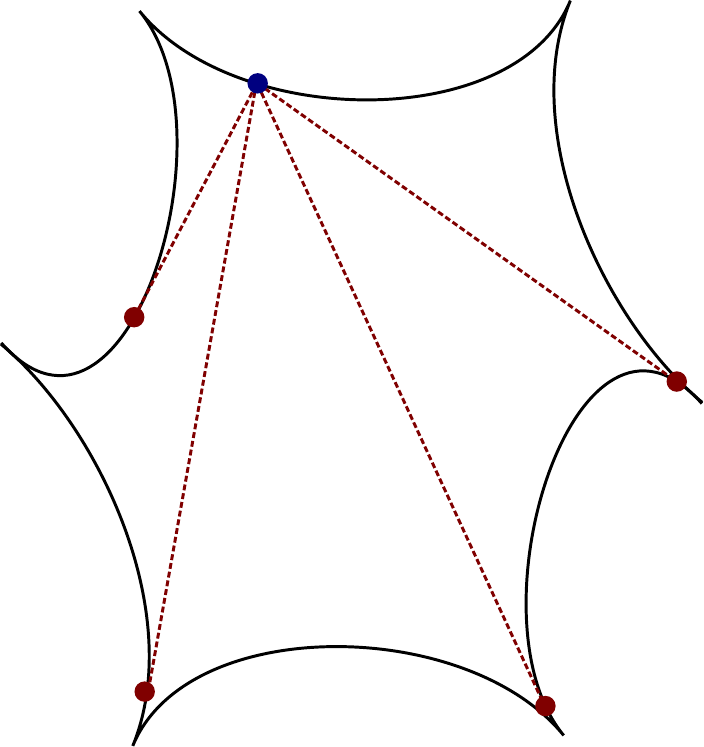}}%
    \put(0.11424387,0.64043762){\color[rgb]{0.50196078,0,0}\makebox(0,0)[lt]{\lineheight{1.25}\smash{\begin{tabular}[t]{l}$p_1$\end{tabular}}}}%
    \put(0.35527304,0.99917882){\color[rgb]{0,0,0.50196078}\makebox(0,0)[lt]{\lineheight{1.25}\smash{\begin{tabular}[t]{l}$q_1$\end{tabular}}}}%
    \put(0.15812531,0.85717708){\color[rgb]{0,0,0.50196078}\makebox(0,0)[lt]{\lineheight{1.25}\smash{\begin{tabular}[t]{l}$-$\end{tabular}}}}%
  \end{picture}%
\endgroup%

\end{subfigure}
\begin{subfigure}{.24\textwidth}
    \centering
    \fontsize{10pt}{12pt}\selectfont
    \def\svgwidth{\linewidth}
\begingroup%
  \makeatletter%
  \providecommand\color[2][]{%
    \errmessage{(Inkscape) Color is used for the text in Inkscape, but the package 'color.sty' is not loaded}%
    \renewcommand\color[2][]{}%
  }%
  \providecommand\transparent[1]{%
    \errmessage{(Inkscape) Transparency is used (non-zero) for the text in Inkscape, but the package 'transparent.sty' is not loaded}%
    \renewcommand\transparent[1]{}%
  }%
  \providecommand\rotatebox[2]{#2}%
  \newcommand*\fsize{\dimexpr\f@size pt\relax}%
  \newcommand*\lineheight[1]{\fontsize{\fsize}{#1\fsize}\selectfont}%
  \ifx\svgwidth\undefined%
    \setlength{\unitlength}{337.51127243bp}%
    \ifx\svgscale\undefined%
      \relax%
    \else%
      \setlength{\unitlength}{\unitlength * \real{\svgscale}}%
    \fi%
  \else%
    \setlength{\unitlength}{\svgwidth}%
  \fi%
  \global\let\svgwidth\undefined%
  \global\let\svgscale\undefined%
  \makeatother%
  \begin{picture}(1,1.06163133)%
    \lineheight{1}%
    \setlength\tabcolsep{0pt}%
    \put(0,0){\includegraphics[width=\unitlength,page=1]{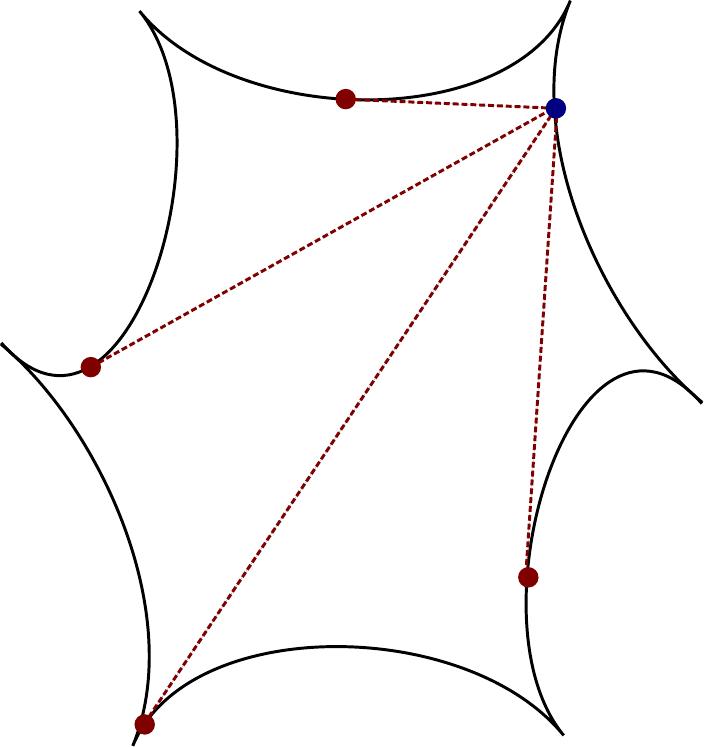}}%
    \put(0.07894328,0.59398004){\color[rgb]{0.50196078,0,0}\makebox(0,0)[lt]{\lineheight{1.25}\smash{\begin{tabular}[t]{l}$p_2$\end{tabular}}}}%
    \put(0.84667371,0.90201973){\color[rgb]{0,0,0.50196078}\makebox(0,0)[lt]{\lineheight{1.25}\smash{\begin{tabular}[t]{l}$q_2$\end{tabular}}}}%
    \put(0.88217416,0.74507043){\color[rgb]{0,0,0.50196078}\makebox(0,0)[lt]{\lineheight{1.25}\smash{\begin{tabular}[t]{l}$+$\end{tabular}}}}%
    \put(0.00638176,0.42551326){\color[rgb]{0,0,0.50196078}\makebox(0,0)[lt]{\lineheight{1.25}\smash{\begin{tabular}[t]{l}$+$\end{tabular}}}}%
  \end{picture}%
\endgroup%

\end{subfigure}
\begin{subfigure}{.24\textwidth}
    \centering
    \fontsize{10pt}{12pt}\selectfont
    \def\svgwidth{\linewidth}
\begingroup%
  \makeatletter%
  \providecommand\color[2][]{%
    \errmessage{(Inkscape) Color is used for the text in Inkscape, but the package 'color.sty' is not loaded}%
    \renewcommand\color[2][]{}%
  }%
  \providecommand\transparent[1]{%
    \errmessage{(Inkscape) Transparency is used (non-zero) for the text in Inkscape, but the package 'transparent.sty' is not loaded}%
    \renewcommand\transparent[1]{}%
  }%
  \providecommand\rotatebox[2]{#2}%
  \newcommand*\fsize{\dimexpr\f@size pt\relax}%
  \newcommand*\lineheight[1]{\fontsize{\fsize}{#1\fsize}\selectfont}%
  \ifx\svgwidth\undefined%
    \setlength{\unitlength}{337.91497607bp}%
    \ifx\svgscale\undefined%
      \relax%
    \else%
      \setlength{\unitlength}{\unitlength * \real{\svgscale}}%
    \fi%
  \else%
    \setlength{\unitlength}{\svgwidth}%
  \fi%
  \global\let\svgwidth\undefined%
  \global\let\svgscale\undefined%
  \makeatother%
  \begin{picture}(1,1.06036301)%
    \lineheight{1}%
    \setlength\tabcolsep{0pt}%
    \put(0,0){\includegraphics[width=\unitlength,page=1]{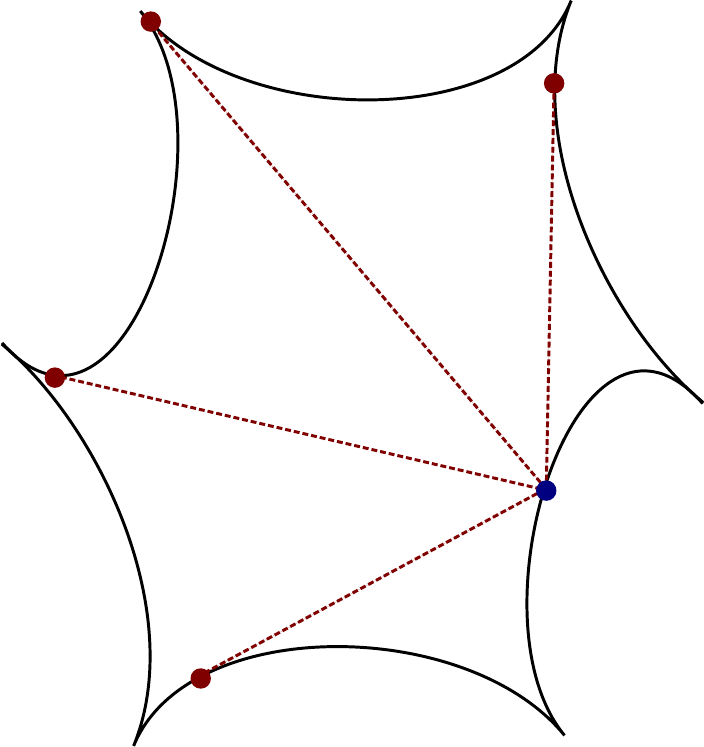}}%
    \put(0.82446232,0.32801514){\color[rgb]{0,0,0.50196078}\makebox(0,0)[lt]{\lineheight{1.25}\smash{\begin{tabular}[t]{l}$q_3$\end{tabular}}}}%
    \put(0.06118196,0.57995364){\color[rgb]{0.50196078,0,0}\makebox(0,0)[lt]{\lineheight{1.25}\smash{\begin{tabular}[t]{l}$p_3$\end{tabular}}}}%
    \put(0.86925146,0.4362555){\color[rgb]{0,0,0.50196078}\makebox(0,0)[lt]{\lineheight{1.25}\smash{\begin{tabular}[t]{l}$-$\end{tabular}}}}%
    \put(-0.00201575,0.44785247){\color[rgb]{0,0,0.50196078}\makebox(0,0)[lt]{\lineheight{1.25}\smash{\begin{tabular}[t]{l}$+$\end{tabular}}}}%
  \end{picture}%
\endgroup%

\end{subfigure}
\begin{subfigure}{.24\textwidth}
    \centering
    \fontsize{10pt}{12pt}\selectfont
    \def\svgwidth{\linewidth}
\begingroup%
  \makeatletter%
  \providecommand\color[2][]{%
    \errmessage{(Inkscape) Color is used for the text in Inkscape, but the package 'color.sty' is not loaded}%
    \renewcommand\color[2][]{}%
  }%
  \providecommand\transparent[1]{%
    \errmessage{(Inkscape) Transparency is used (non-zero) for the text in Inkscape, but the package 'transparent.sty' is not loaded}%
    \renewcommand\transparent[1]{}%
  }%
  \providecommand\rotatebox[2]{#2}%
  \newcommand*\fsize{\dimexpr\f@size pt\relax}%
  \newcommand*\lineheight[1]{\fontsize{\fsize}{#1\fsize}\selectfont}%
  \ifx\svgwidth\undefined%
    \setlength{\unitlength}{337.51127243bp}%
    \ifx\svgscale\undefined%
      \relax%
    \else%
      \setlength{\unitlength}{\unitlength * \real{\svgscale}}%
    \fi%
  \else%
    \setlength{\unitlength}{\svgwidth}%
  \fi%
  \global\let\svgwidth\undefined%
  \global\let\svgscale\undefined%
  \makeatother%
  \begin{picture}(1,1.06163133)%
    \lineheight{1}%
    \setlength\tabcolsep{0pt}%
    \put(0,0){\includegraphics[width=\unitlength,page=1]{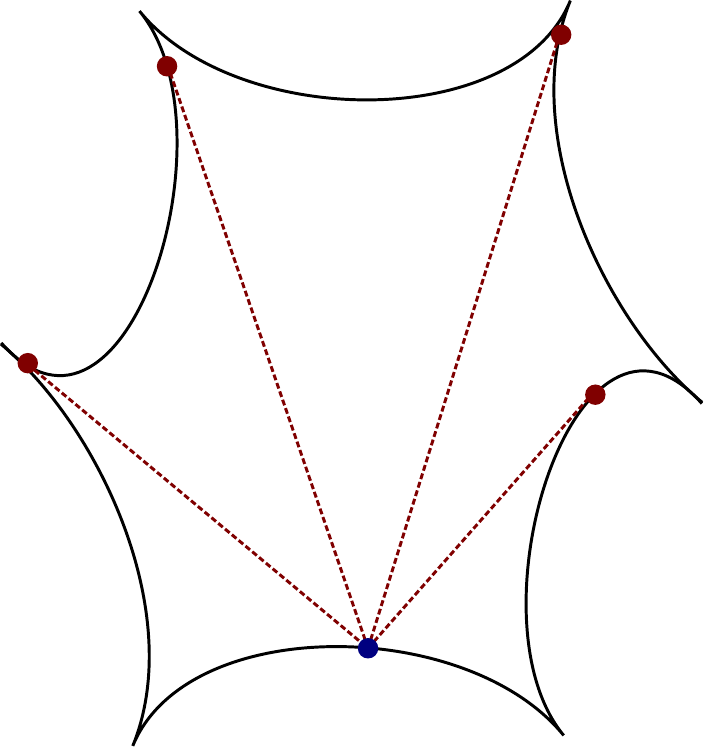}}%
    \put(0.50288009,0.07056246){\color[rgb]{0,0,0.50196078}\makebox(0,0)[lt]{\lineheight{1.25}\smash{\begin{tabular}[t]{l}$q_4$\end{tabular}}}}%
    \put(0.00961104,0.61054248){\color[rgb]{0.50196078,0,0}\makebox(0,0)[lt]{\lineheight{1.25}\smash{\begin{tabular}[t]{l}$p_4$\end{tabular}}}}%
    \put(0.04511145,0.32840749){\color[rgb]{0,0,0.50196078}\makebox(0,0)[lt]{\lineheight{1.25}\smash{\begin{tabular}[t]{l}$+$\end{tabular}}}}%
    \put(0.30295667,0.03879884){\color[rgb]{0,0,0.50196078}\makebox(0,0)[lt]{\lineheight{1.25}\smash{\begin{tabular}[t]{l}$+$\end{tabular}}}}%
  \end{picture}%
\endgroup%

\end{subfigure}
\caption{The same central bubble as in Figure~\ref{fig:central_gas_bubble_tangential_zeros_and_signs} with the four possible cases of tangential rays $r$ starting at $p_i$ along which the extension $g_1$ is defined. The sign of $\tilde{f}_{q_i}$ on the relevant arcs is indicated. 
In all $4$ cases $\int_{P_i}^{Q_i} d\zeta_{q_i} \geq 0$.}
\label{fig:central_hex_bubble_tangents_4_cases}
\end{figure}

Similarly we define the extension $g_1$ along rays tangential to points on $[c_1, c_2]$ oriented away from $c_1$ and conclude that 
\[ g_1(q) \geq g(q), \forall q \in \partial G_1. \]

Analogously, we construct the functions $g_3, g_5$ as extensions on rays on $[c_2, c_4]$ and $[c_4, c_6]$ oriented away from $c_3, c_5$ respectively and again obtain
\[ g_i(q) \geq g(q), \forall q \in \partial G_1, \; J\nabla g_i(q) = J\nabla g(P) \in \mathcal{A}_1. \]

Therefore the function $\hat{g} = \min(g_1, g_3, g_5)$ satisfies the boundary conditions $\restr{\hat{g}}{\partial G_1} = g$ and has rotated gradients living in $\mathcal{A}_1$ almost everywhere. Thus $\hat{g}$ is an admissible extension of $g$ and Theorem~\ref{thm:h_is_height_fct_hex_unramified} holds.

Note that again this extension is not unique. We have constructed a maximal extension. A minimal extension is given by $\hat{g} = \max(g_2, g_4, g_6)$. See also Remark~\ref{rem:g_extension_non_unique}.

\subsection{Parameter Count}
\label{subsec:parameter_count_hex}

We conclude this section with a parameter count.
Let us first consider the ramified case.

To that end let $\calR_+$ be uniformized by $\mathbb{D}$ with $g$ round disks removed and train track parameters $\cubr{\alpha_i, \beta_j, \gamma_l}$ chosen on $S^1$. A double cover $\hat{\calR}$ with branch point $P_0$ is then given by glueing two copies of $\calR$ along $[P_0, \frac{1}{\overline{P}_0}]$. We define $d\zeta_3$ on $\hat{\calR}$ and this data is admissible under the condition $d\zeta_3(P_0) = 0$. The real dimension of the moduli space $\mathcal{M}$ is then
$$
\dim_\mathbb{R} \mathcal{M} = \underbrace{3g}_{X_1, \ldots, X_g} + \underbrace{2}_{P_0} + \underbrace{3n}_{\cubr{\alpha_i, \beta_j, \gamma_l}} - \underbrace{2}_{d\zeta_3(P_0) = 0} - \underbrace{3}_{\text{disk automorphisms}}
                                = 3(g - 1) + 3n.
$$

On the other hand, the admissible Harnack data $\hat{\mathcal{S}}$ in the ramified case can be uniformized as in Remark~\ref{rem:uniformization_of_hex_curve}. That is we have centrally symmetric round holes and train track choices. The condition $d\zeta_3(P_0) = 0$ is now a restriction on the choice of train track parameters $\cubr{\alpha^\pm_i, \beta^\pm_j, \gamma^\pm_l}$. The parameter count is now

$$
    \dim_\mathbb{R} \mathcal{M} = \underbrace{3g}_{\hat{X}_1, \ldots, \hat{X}_g} + \underbrace{3n}_{\cubr{\alpha^\pm_i, \beta^\pm_j, \gamma^\pm_l}} - \underbrace{2}_{d\zeta_3(P_0) = 0} - \underbrace{1}_{\text{rotations}}
                                = 3(g - 1) + 3n
$$

and we have obtained the same moduli space in another realization.

Similarly for the unramified case we can consider the Harnack data $\mathcal{S} = \cubr{\calR, \cubr{\alpha_i, \beta_j, \gamma_l}}$ and take its double cover $\hat{\mathcal{S}}$ winding around some chosen oval $X_i$. The admissibility condition Definition~\ref{def:admissible_harnack_data_unramified} is an open condition, and 
we obtain 
$$
\dim_\mathbb{R} \tilde{\mathcal{M}} = \underbrace{3g}_{X_1, \ldots, X_g} + \underbrace{3n}_{\cubr{\alpha_i, \beta_j, \gamma_l}} - \underbrace{3}_{\text{disk automorphisms}} 
                                = 3(g - 1) + 3n.
$$
On the other hand, the admissible $\hat{\mathcal{S}}$ are uniformized it centrally symmetrically with $\hat{X}_1$ being a circle centered at $0$ and thus only contributing its radius parameter. The parameter count is now:
$$
    \dim_\mathbb{R} \tilde{\mathcal{M}} = \underbrace{3(g - 1)}_{\hat{X}_2, \ldots, \hat{X}_g} + \underbrace{1}_{\text{radius of } \hat{X}_1} + \underbrace{3n}_{\cubr{\alpha_i, \beta_j, \gamma_l}} - \underbrace{1}_{\text{rotations}} 
                                = 3(g - 1) + 3n.
$$
Summarizing, in the ramified and unramified cases the moduli spaces have the same dimension $3(g-1) + 3n$, and they are glued together continuously at singular $\hat{\calR}$. See Figure~\ref{fig:ramified_touching_unramified_hex} for an illustration.

\begin{figure}
\centering
\begin{subfigure}[p]{.3\textwidth}
    \centering
    \includegraphics[width=\linewidth]{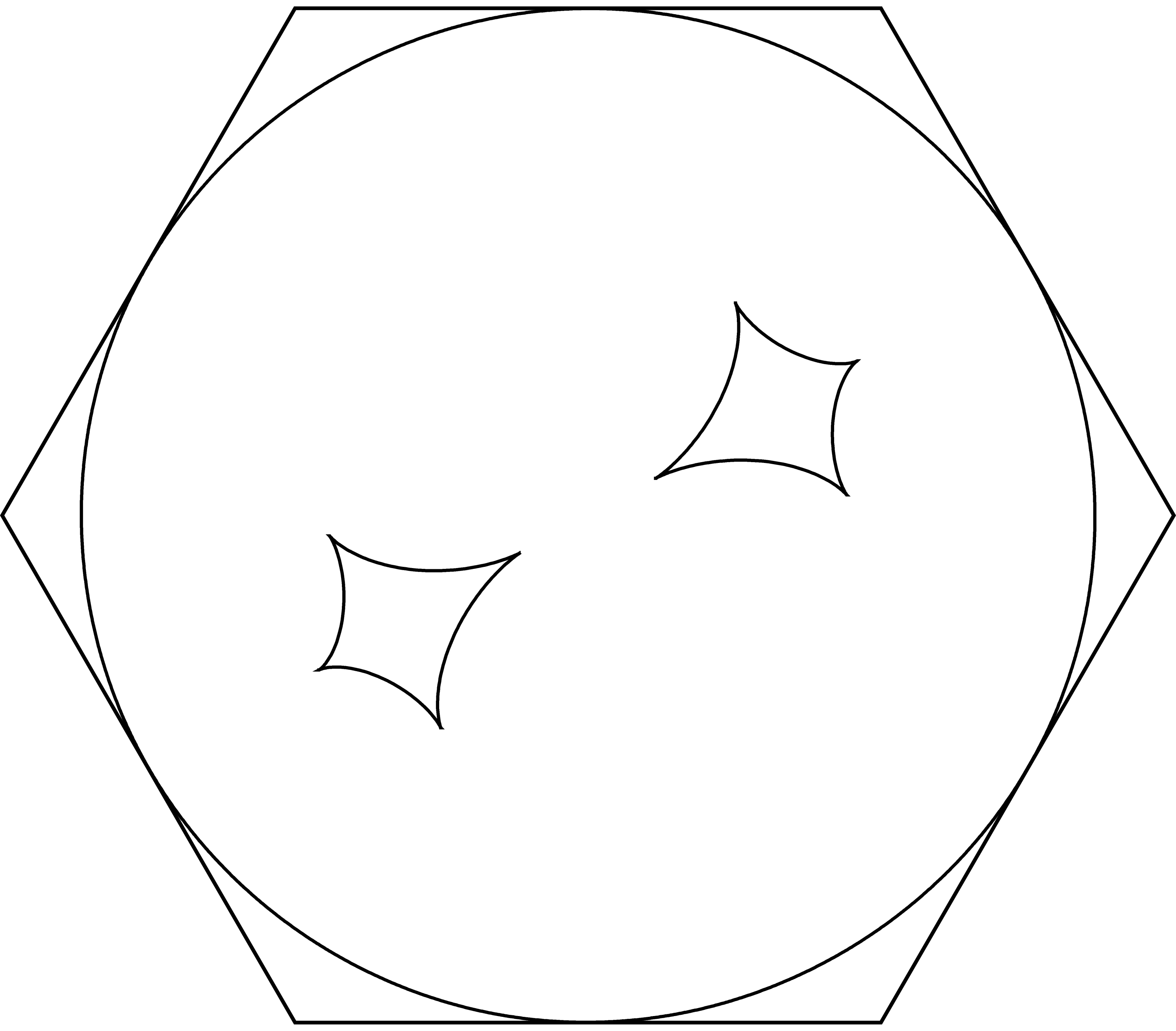}
\end{subfigure}
\begin{subfigure}[p]{.3\textwidth}
    \centering
    \includegraphics[width=\linewidth]{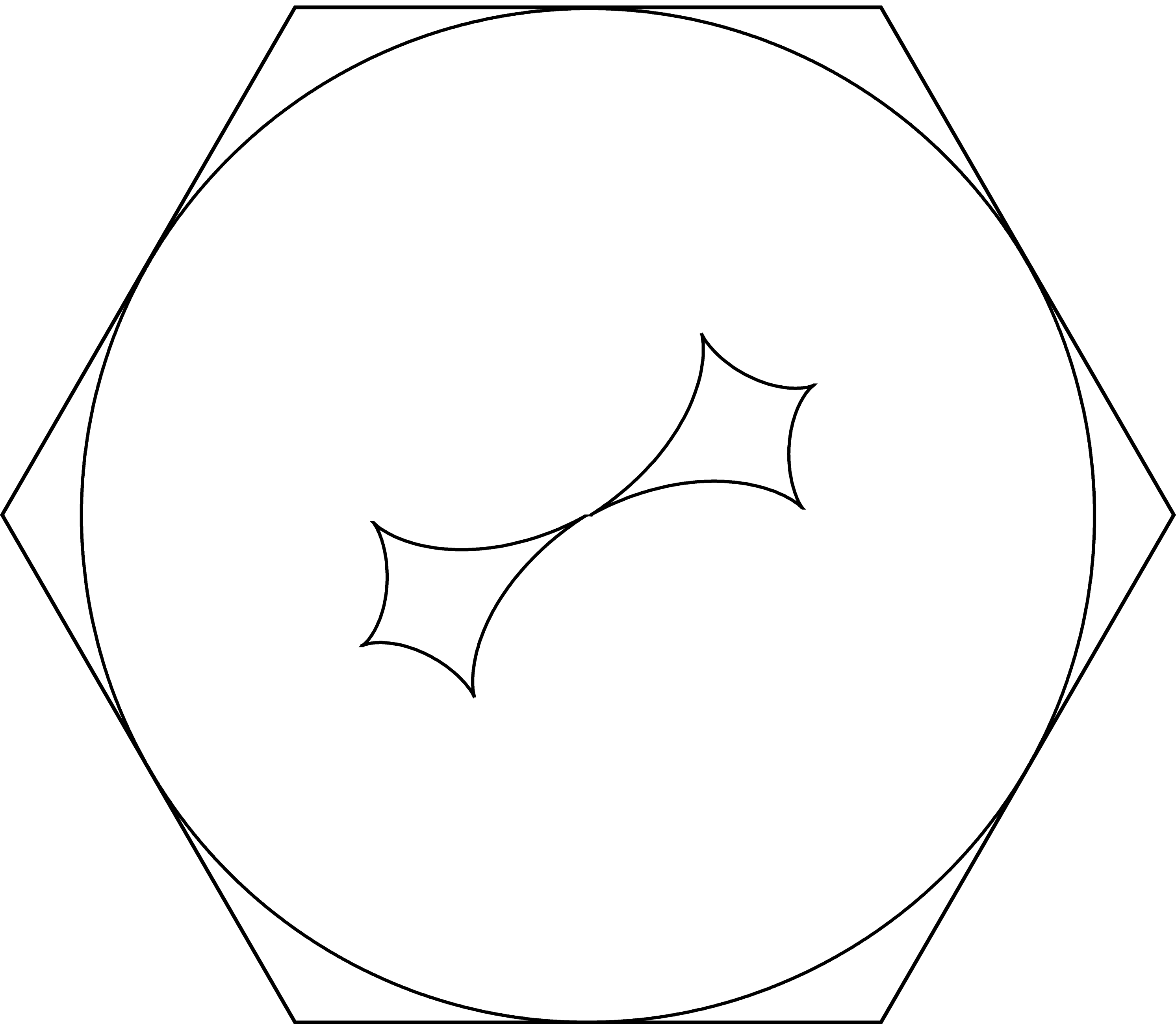}
\end{subfigure}
\begin{subfigure}[p]{.3\textwidth}
    \centering
    \includegraphics[width=\linewidth]{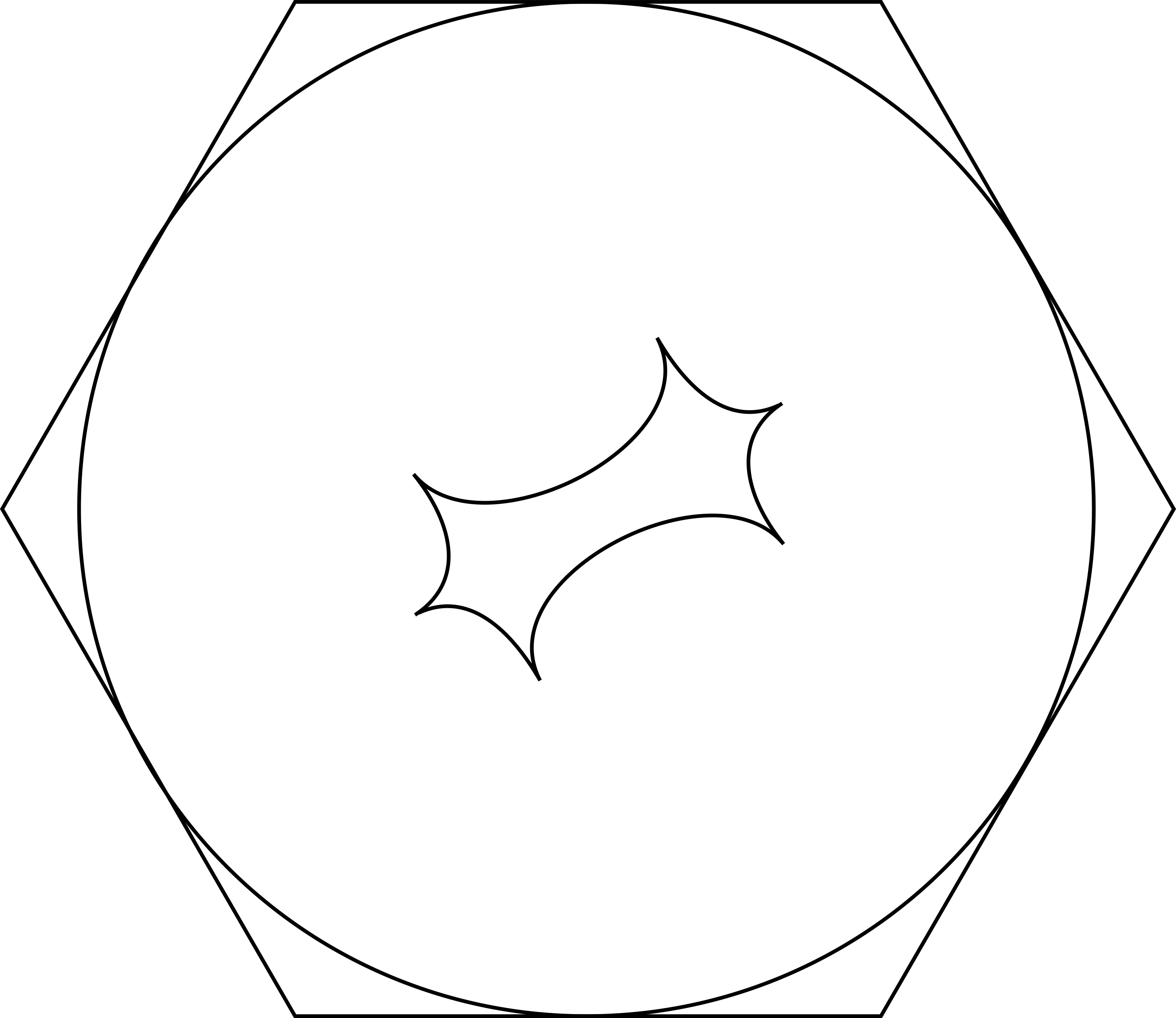}
\end{subfigure}
\caption{Three examples of arctic curves for the hexagon. These examples all correspond to genus $1$ weights with  the covering structure changing. From left to right the ramified, degenerate and unramified case are shown. We see the continuous glueing of the moduli spaces along the degenerate case where two gas bubbles touch at a cusp. The fact that glueing occurs at a cusp follows from the parallelity property along with local concavity of gas bubbles.}
\label{fig:ramified_touching_unramified_hex}
\end{figure}

\section{Schottky Uniformization and Computation}
\label{sec:10_SchottkyUniformization}

\addtocontents{toc}{\protect\setcounter{tocdepth}{1}}

The seemingly complicated and abstract parametrization of the dimer model and many of its fundamental objects via Harnack data is made concrete and computationally feasible via Schottky uniformizations. This approach is particularly well suited for M-curves that appear naturally in the dimer context as in that case convergence of the Poincaré series describing our differentials and weights can be guaranteed. The idea of using Schottky uniformizations for computational aspects of Riemann surfaces was introduced in \cite{Bobenko_1987} in the context of nonlinear integrable equations. In the interest of brevity we only exhibit the most relevant aspects and give expressions for computation of the objects discussed in earlier sections. For more background and details see \cite{belokolos_algebro-geometric_1994, BBS}.

\subsection*{Schottky uniformization of M-curves}

Let $\cubr{C_n}$, $n\leq g$ be a collection of circles bounding disjoint discs that all lie in $\mathbb{D}$. We denote by $\sigma_n$ the inversion in $C_n$ and by $\sigma_0 z = \frac{1}{\overline{z}}$ the inversion in $S^1 =: C_0$. Then the group 
\[ G = \cubr{\sigma_{i_1}^{l_1} \ldots \sigma_{i_k}^{l_k} \; | \; l_i \in \mathbb{N}, \sum l_i \in 2\mathbb{Z} } \]
generated by even compositions of these inversions is a classical Schottky group and defines an M-curve $\calR = \Omega / G$ where $\Omega$ is the discontinuity set of $G$ and the antiholomorphic involution $\tau$ is given by $\sigma_0$. The real ovals of $\calR$ fixed by $\tau$ are then just $C_0, C_1, \ldots, C_g$, and $\calR_+$ is a disk with $g$ removed disks. This brings the uniformization in line with the above exposition and the figures presented in this paper. 

Note that we have $3$ real parameters per circle and Möbius transformations of the defining circles $\cubr{C_n}$ do not change $\calR$ and so we have a moduli space of dimension $3(g-1)$ for $g\geq 2$. In fact this parametrization yields a complete description of the moduli space of M-curves, see Remark~\ref{rem:uniformization_calR}.
In \cite{BBS} we used a M\"obius equivalent uniformization by normalizing $X_0 = \mathbb{R}$, $\tau z = \bar{z}$.


The meromorphic differentials $d\zeta_i$ have a simple representation in terms of Schottky groups as introduced above. Note that this has the same form as \eqref{eq:dzeta_i_isoradial}, \eqref{eq:zeta_i_isoradial} averaged over $G$.
\begin{proposition}
\label{prop:zeta_Schottky}
    Let $G$ be a Schottky group defined by $\cubr{C_n}$ as described above. Then the for the square grid the differentials $d\zeta_k$ and integrals $\zeta_k$ are given by the following series (provided they converge):
    \begin{eqnarray}
    & d\zeta^{\alpha_i}(z) = \sum_{\sigma \in G} \br{\frac{1}{ z - \sigma \alpha_i^-} - \frac{1}{ z - \sigma \alpha_i^+} } dz,\quad
    d\zeta^{\beta_j}(z) = \sum_{\sigma \in G} \br{\frac{1}{ z - \sigma \beta_j^-} - \frac{1}{ z - \sigma \beta_j^+} } dz,
    \label{eq:dzeta_Schottky}\\
    & \zeta^{\alpha_i}(z) = \sum_{\sigma \in G} \log \{  z, \sigma\alpha_i^-, z_0, \sigma \alpha_i^+\}, \quad
    \zeta^{\beta_j}(z) = \sum_{\sigma \in G} \log \{  z, \sigma\beta_j^-, z_0, \sigma \beta_j^+\}.
    \label{eq:zeta_Schottky}
    \end{eqnarray}
    Here $z_0 \in S^1$ is the starting integration point and the Harnack data are cyclicaly ordered $\beta_j^- < \alpha_i^+ < \beta_j^+ < \alpha_l^-  \in S^1$.
\end{proposition}

Similar formulas are valid for $d\zeta_3$, and also in the context of hexagonal grids, as well as for the period matrices needed for the calculation of the weights \eqref{eq:Fock_face_weight} using theta functions.

Similarly the maps mentioned in Remarks~\ref{rem:uniformization_of_hex_curve} and $\ref{rem:uniformization_calR_unramified}$ provide symmetric Schottky uniformizations of $\hat{\calR}$. We summarize the chosen uniformization types for the ramified, degenerate and unramified cases of the hexagon in Figure~\ref{fig:three_types_of_uniformizations}.

\begin{figure}[h]
\centering
\begin{subfigure}{.3\textwidth}
    \centering
    \fontsize{10pt}{12pt}\selectfont
    \def\svgwidth{\linewidth}
\begingroup%
  \makeatletter%
  \providecommand\color[2][]{%
    \errmessage{(Inkscape) Color is used for the text in Inkscape, but the package 'color.sty' is not loaded}%
    \renewcommand\color[2][]{}%
  }%
  \providecommand\transparent[1]{%
    \errmessage{(Inkscape) Transparency is used (non-zero) for the text in Inkscape, but the package 'transparent.sty' is not loaded}%
    \renewcommand\transparent[1]{}%
  }%
  \providecommand\rotatebox[2]{#2}%
  \newcommand*\fsize{\dimexpr\f@size pt\relax}%
  \newcommand*\lineheight[1]{\fontsize{\fsize}{#1\fsize}\selectfont}%
  \ifx\svgwidth\undefined%
    \setlength{\unitlength}{122.23113966bp}%
    \ifx\svgscale\undefined%
      \relax%
    \else%
      \setlength{\unitlength}{\unitlength * \real{\svgscale}}%
    \fi%
  \else%
    \setlength{\unitlength}{\svgwidth}%
  \fi%
  \global\let\svgwidth\undefined%
  \global\let\svgscale\undefined%
  \makeatother%
  \begin{picture}(1,1.01499454)%
    \lineheight{1}%
    \setlength\tabcolsep{0pt}%
    \put(0,0){\includegraphics[width=\unitlength,page=1]{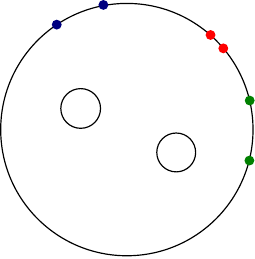}}%
    \put(0.24757933,0.28001951){\color[rgb]{0,0,0}\makebox(0,0)[lt]{\lineheight{1.25}\smash{\begin{tabular}[t]{l}$\hat{\mathcal{R}}^+$\end{tabular}}}}%
    \put(0.66048258,0.53618715){\color[rgb]{0,0,0}\makebox(0,0)[lt]{\lineheight{1.25}\smash{\begin{tabular}[t]{l}$\hat{X}_1$\end{tabular}}}}%
    \put(0.25323409,0.69900285){\color[rgb]{0,0,0}\makebox(0,0)[lt]{\lineheight{1.25}\smash{\begin{tabular}[t]{l}$\hat{X}_2 = -\hat{X}_1$\end{tabular}}}}%
    \put(0,0){\includegraphics[width=\unitlength,page=2]{Figures/Uniformization2_Ramified.pdf}}%
    \put(0.54778169,0.09762342){\color[rgb]{0,0,0.50196078}\makebox(0,0)[lt]{\lineheight{1.25}\smash{\begin{tabular}[t]{l}$\gamma_1^-$\end{tabular}}}}%
    \put(0.40712197,0.89193553){\color[rgb]{0,0,0.50196078}\makebox(0,0)[lt]{\lineheight{1.25}\smash{\begin{tabular}[t]{l}$\gamma_1^+ = -\gamma_1^-$\end{tabular}}}}%
  \end{picture}%
\endgroup%

\end{subfigure}
\begin{subfigure}{.3\textwidth}
    \centering
    \fontsize{10pt}{12pt}\selectfont
    \def\svgwidth{\linewidth}
\begingroup%
  \makeatletter%
  \providecommand\color[2][]{%
    \errmessage{(Inkscape) Color is used for the text in Inkscape, but the package 'color.sty' is not loaded}%
    \renewcommand\color[2][]{}%
  }%
  \providecommand\transparent[1]{%
    \errmessage{(Inkscape) Transparency is used (non-zero) for the text in Inkscape, but the package 'transparent.sty' is not loaded}%
    \renewcommand\transparent[1]{}%
  }%
  \providecommand\rotatebox[2]{#2}%
  \newcommand*\fsize{\dimexpr\f@size pt\relax}%
  \newcommand*\lineheight[1]{\fontsize{\fsize}{#1\fsize}\selectfont}%
  \ifx\svgwidth\undefined%
    \setlength{\unitlength}{122.22529791bp}%
    \ifx\svgscale\undefined%
      \relax%
    \else%
      \setlength{\unitlength}{\unitlength * \real{\svgscale}}%
    \fi%
  \else%
    \setlength{\unitlength}{\svgwidth}%
  \fi%
  \global\let\svgwidth\undefined%
  \global\let\svgscale\undefined%
  \makeatother%
  \begin{picture}(1,1.01504306)%
    \lineheight{1}%
    \setlength\tabcolsep{0pt}%
    \put(0,0){\includegraphics[width=\unitlength,page=1]{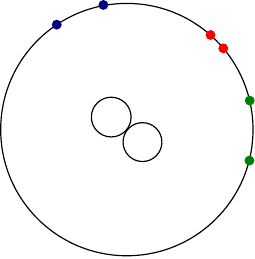}}%
    \put(0.24759117,0.28003289){\color[rgb]{0,0,0}\makebox(0,0)[lt]{\lineheight{1.25}\smash{\begin{tabular}[t]{l}$\hat{\mathcal{R}}^+$\end{tabular}}}}%
    \put(0.58115182,0.5656281){\color[rgb]{0,0,0}\makebox(0,0)[lt]{\lineheight{1.25}\smash{\begin{tabular}[t]{l}$\hat{X}_1$\end{tabular}}}}%
    \put(0.39954077,0.67119393){\color[rgb]{0,0,0}\makebox(0,0)[lt]{\lineheight{1.25}\smash{\begin{tabular}[t]{l}$\hat{X}_2 = -\hat{X}_1$\end{tabular}}}}%
    \put(0,0){\includegraphics[width=\unitlength,page=2]{Figures/Uniformization4_Degenerate.pdf}}%
    \put(0.54780787,0.09762808){\color[rgb]{0,0,0.50196078}\makebox(0,0)[lt]{\lineheight{1.25}\smash{\begin{tabular}[t]{l}$\gamma_1^-$\end{tabular}}}}%
    \put(0.40714143,0.89197816){\color[rgb]{0,0,0.50196078}\makebox(0,0)[lt]{\lineheight{1.25}\smash{\begin{tabular}[t]{l}$\gamma_1^+ = -\gamma_1^-$\end{tabular}}}}%
  \end{picture}%
\endgroup%

\end{subfigure}
\begin{subfigure}{.3\textwidth}
    \centering
    \fontsize{10pt}{12pt}\selectfont
    \def\svgwidth{\linewidth}
\begingroup%
  \makeatletter%
  \providecommand\color[2][]{%
    \errmessage{(Inkscape) Color is used for the text in Inkscape, but the package 'color.sty' is not loaded}%
    \renewcommand\color[2][]{}%
  }%
  \providecommand\transparent[1]{%
    \errmessage{(Inkscape) Transparency is used (non-zero) for the text in Inkscape, but the package 'transparent.sty' is not loaded}%
    \renewcommand\transparent[1]{}%
  }%
  \providecommand\rotatebox[2]{#2}%
  \newcommand*\fsize{\dimexpr\f@size pt\relax}%
  \newcommand*\lineheight[1]{\fontsize{\fsize}{#1\fsize}\selectfont}%
  \ifx\svgwidth\undefined%
    \setlength{\unitlength}{135.94494181bp}%
    \ifx\svgscale\undefined%
      \relax%
    \else%
      \setlength{\unitlength}{\unitlength * \real{\svgscale}}%
    \fi%
  \else%
    \setlength{\unitlength}{\svgwidth}%
  \fi%
  \global\let\svgwidth\undefined%
  \global\let\svgscale\undefined%
  \makeatother%
  \begin{picture}(1,1.01504304)%
    \lineheight{1}%
    \setlength\tabcolsep{0pt}%
    \put(0,0){\includegraphics[width=\unitlength,page=1]{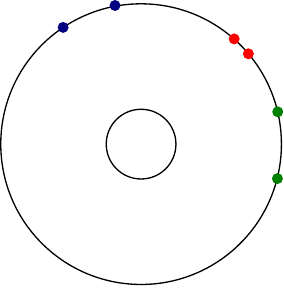}}%
    \put(0.28736708,0.23515007){\color[rgb]{0,0,0}\makebox(0,0)[lt]{\lineheight{1.25}\smash{\begin{tabular}[t]{l}$\hat{\mathcal{R}}^+$\end{tabular}}}}%
    \put(0,0){\includegraphics[width=\unitlength,page=2]{Figures/Uniformization3_Unramified.pdf}}%
    \put(0.56760566,0.10364715){\color[rgb]{0,0,0.50196078}\makebox(0,0)[lt]{\lineheight{1.25}\smash{\begin{tabular}[t]{l}$\gamma_1^-$\end{tabular}}}}%
    \put(0.39080651,0.89448385){\color[rgb]{0,0,0.50196078}\makebox(0,0)[lt]{\lineheight{1.25}\smash{\begin{tabular}[t]{l}$\gamma_1^+ = -\gamma_1^-$\end{tabular}}}}%
    \put(0.47244208,0.49782373){\color[rgb]{0,0,0}\makebox(0,0)[lt]{\lineheight{1.25}\smash{\begin{tabular}[t]{l}$\hat{X}_1$\end{tabular}}}}%
  \end{picture}%
\endgroup%

\end{subfigure}
\caption{Uniformizations of $\hat{\calR}$ for three different cases of the hexagon. Central symmetry $\pi z = -z$ satisfied by the real ovals and train track parameters in all cases. \textbf{Left}: Ramified case with branch point at $0$ with condition $d\zeta_3(0) = 0$. \textbf{Middle}: Degenerate case with branch point on real oval. \textbf{Right}: Unramified case. We have that $d\zeta_3$ has $6$ zeros on $\hat{X}_1$. See Remarks~\ref{rem:uniformization_of_hex_curve}, \ref{rem:uniformization_calR_unramified} as well as Figure~\ref{fig:ramified_touching_unramified_hex} for arctic curves for these $3$ types. }
\label{fig:three_types_of_uniformizations}
\end{figure}

\subsection*{Computation}

One of the benefits of working with Schottky groups is that the expressions for the Fock weights $K_{wb}$, the differentials $d\zeta_k$ and their integrals $\zeta_k$ in terms of (converging) Poincaré series allow for practical numerical computation. 

We utilize this approach to obtain samples of dimer configurations with weights coming from some given Harnack data $\cubr{\mathcal{S}, \mathcal{T}}$ overlaid with predicted arctic curves on both the Aztec diamond and the hexagon. 

We use the jtem library \cite{schmies_computational_2005} for summations over Schottky groups. For computation of theta functions see \cite{deconinck_2004}. A simple Metropolis-Hastings sampling algorithm with those weights that we have implemented on the GPU for better performance, see \cite{keating_random_2018,BBS} yields approximate samples. In the examples of the Aztec diamond and hexagon there are shuffling algorithms available that obtain efficient exact samples, see \cite{Gorin_2021} for an overview. We have chosen not to use those as we hope to be able extend our methods to other setups and regions in the future.

Proposition~\ref{prop:zeta_Schottky} allows for computation of $d\zeta_k$ as well as $\zeta_k$. Evaluating $(\zeta_1, \zeta_2)$ on $X = \cubr{C_n}$ gives the real ovals of the amoeba and computing \eqref{eq:xy_on_real_ovals} on $X$ yields the predicted arctic curves. We see that the predictions and sampled configurations match well, see e.g. Figures~\ref{fig:main_aztec_config}, \ref{fig:main_hex_config}.

For the hexagon in the ramified case we impose the condition that the Harnack data $\hat{\mathcal{S}}$ must define an admissible $d\zeta_3 = f_3(z) dz$, i.e. $f_3(0) = 0$, see Definition~\ref{def:hex_Harnack_data_admissible_ramified}. This is generically not true for symmetric $\hat{\mathcal{S}}$. In order to find an admissible $\hat{S}$ we run an optimization on the train track parameters $\cubr{\alpha^\pm_i, \beta^\pm_j, \gamma^\pm_l}$ utilizing derivatives of $f_3$ at the point $0$ with respect to the train track parameters. This does not change the Riemann surface and typically only moves the train track parameters a little bit to obtain admissible data. Since the train track angles do not move much this approach yields good control of the final picture when chosing the initial parameters $\hat{\mathcal{S}}$.

\begin{figure}[h]
\centering
\begin{subfigure}{.16\textwidth}
    \centering
    \fontsize{10pt}{12pt}\selectfont
    \def\svgwidth{\linewidth}
    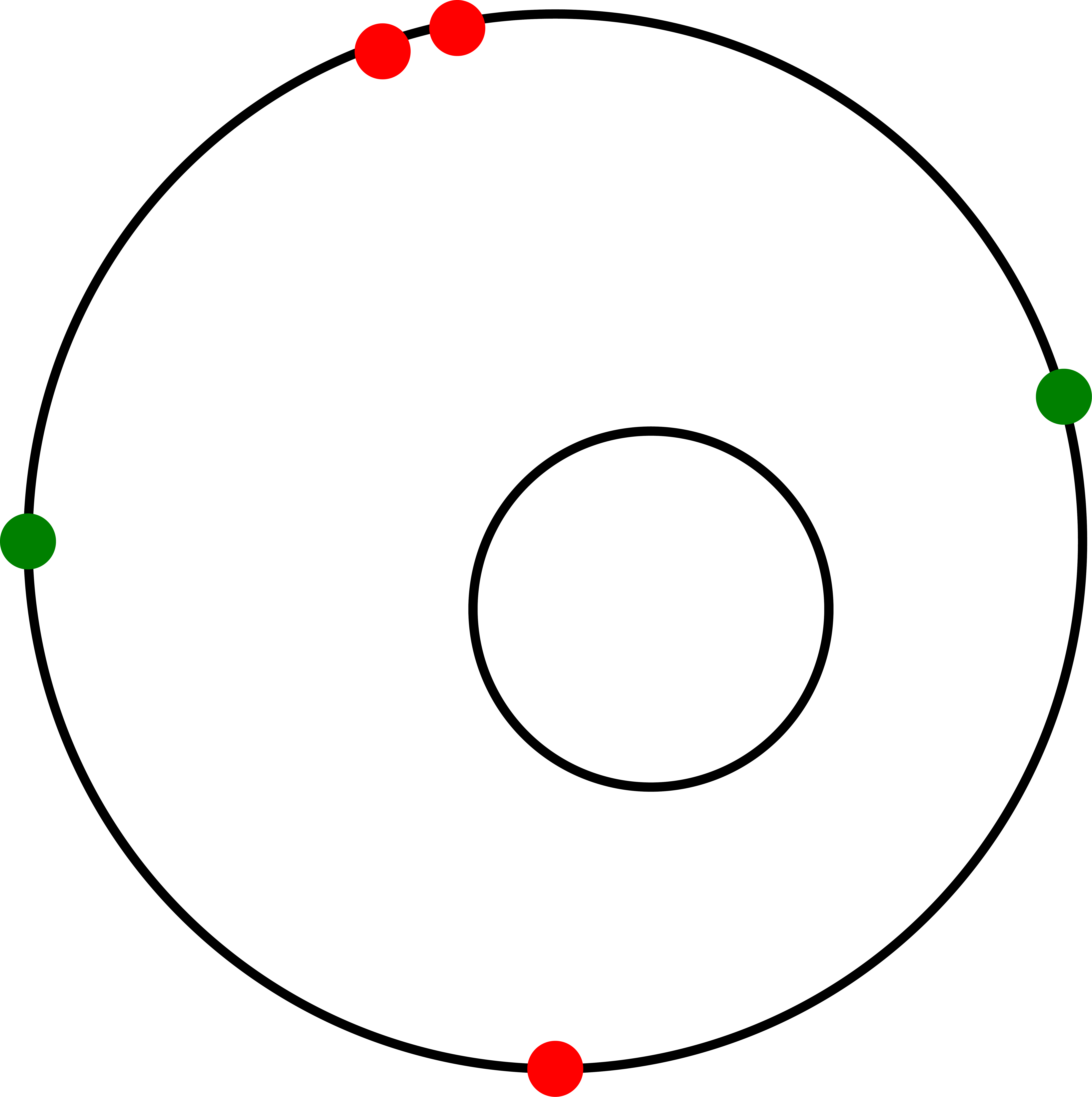
\end{subfigure}
\begin{subfigure}{.16\textwidth}
    \centering
    \fontsize{10pt}{12pt}\selectfont
    \def\svgwidth{\linewidth}
    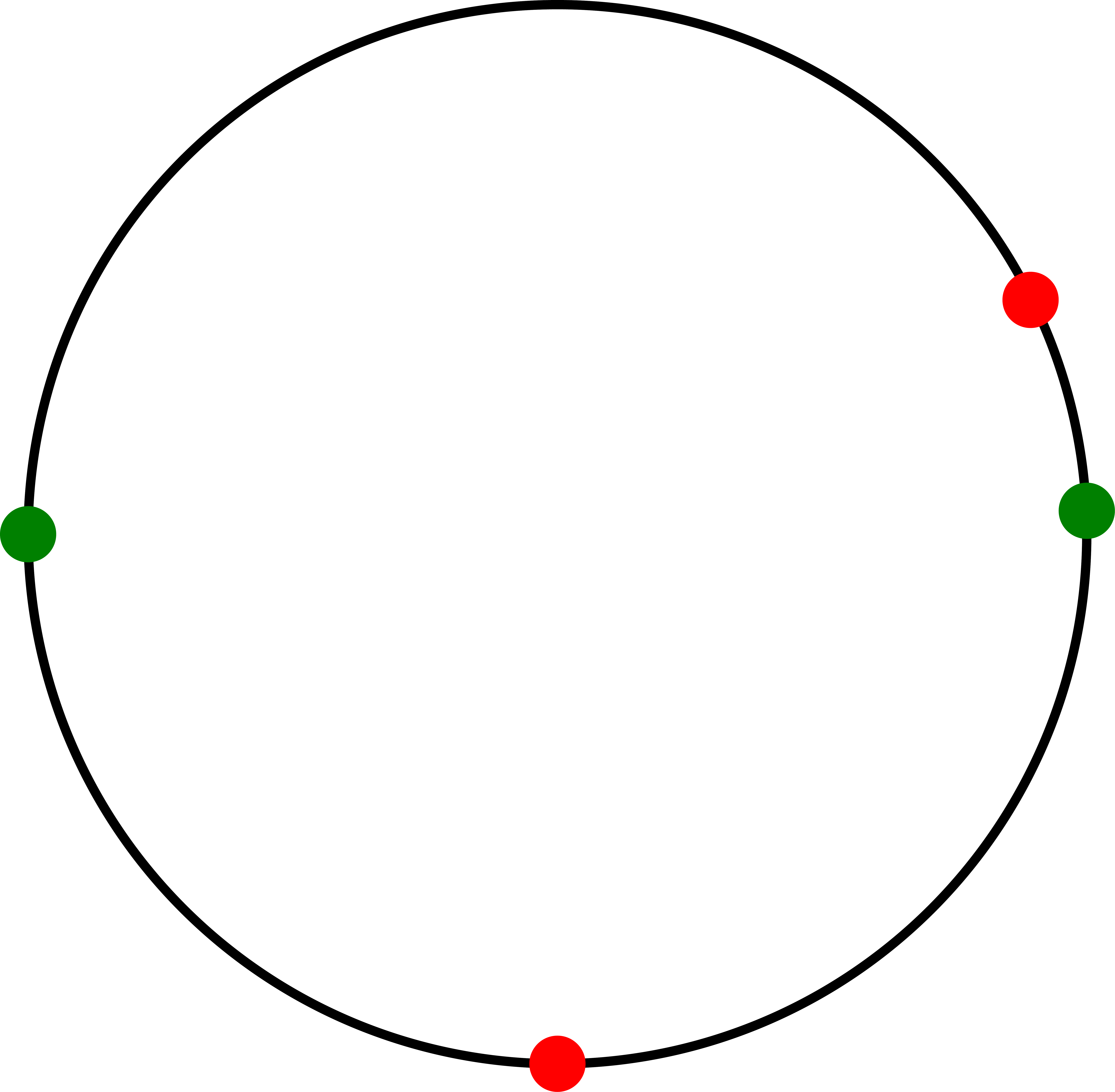
\end{subfigure}
\begin{subfigure}{.16\textwidth}
    \centering
    \fontsize{10pt}{12pt}\selectfont
    \def\svgwidth{\linewidth}
    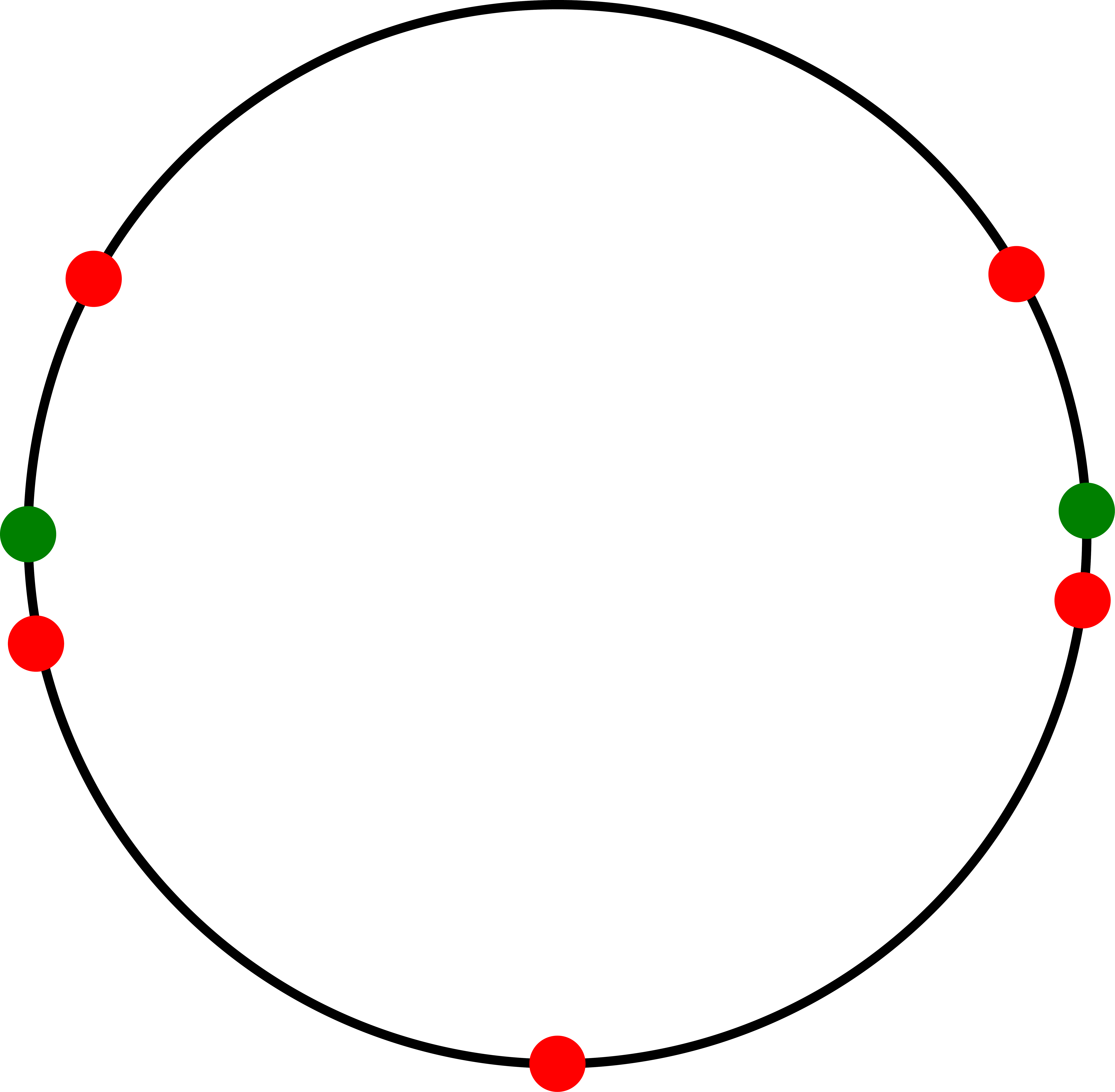
\end{subfigure}
\begin{subfigure}{.16\textwidth}
    \centering
    \fontsize{10pt}{12pt}\selectfont
    \def\svgwidth{\linewidth}
    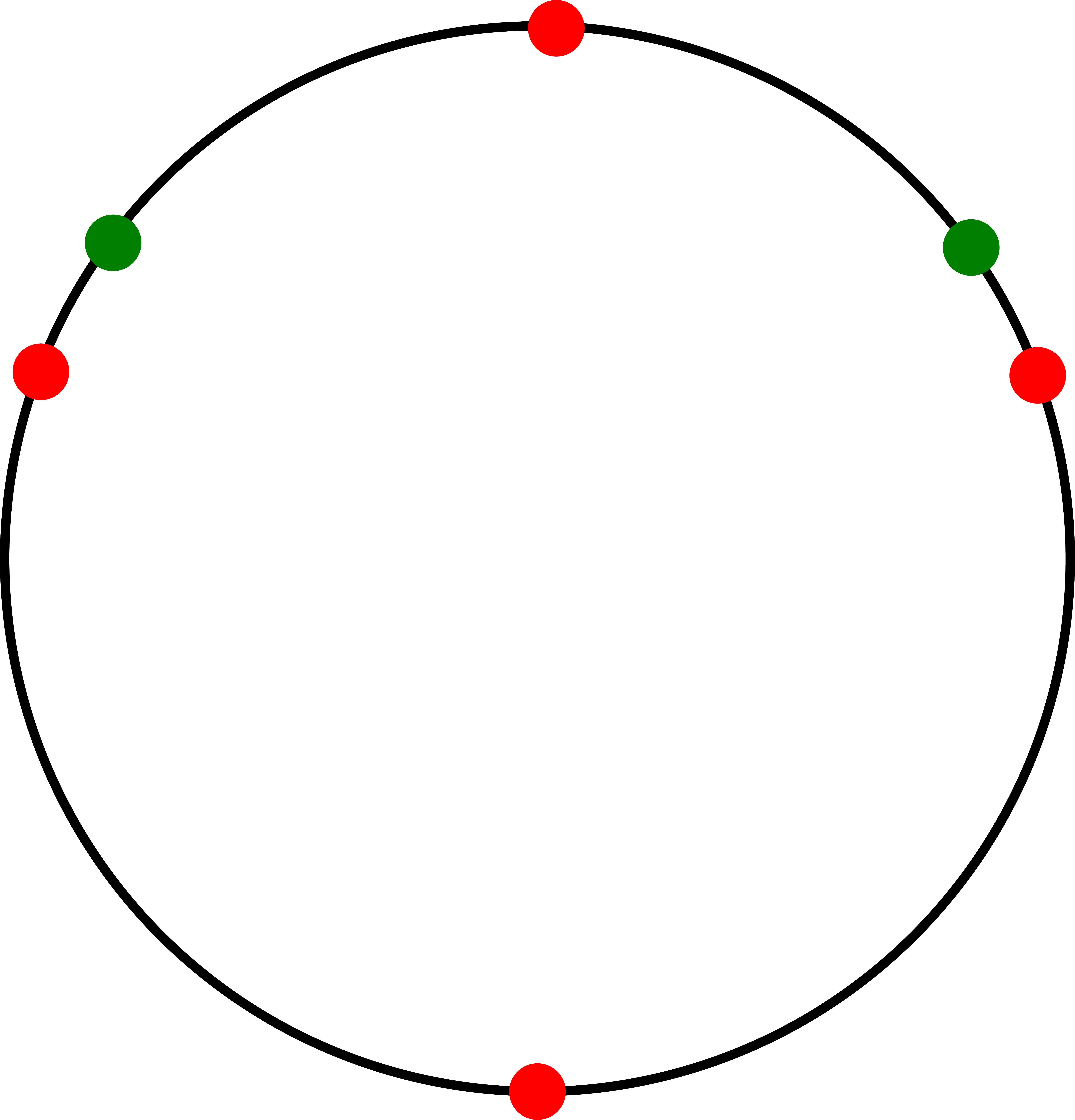
\end{subfigure}
\begin{subfigure}{.16\textwidth}
    \centering
    \fontsize{10pt}{12pt}\selectfont
    \def\svgwidth{\linewidth}
    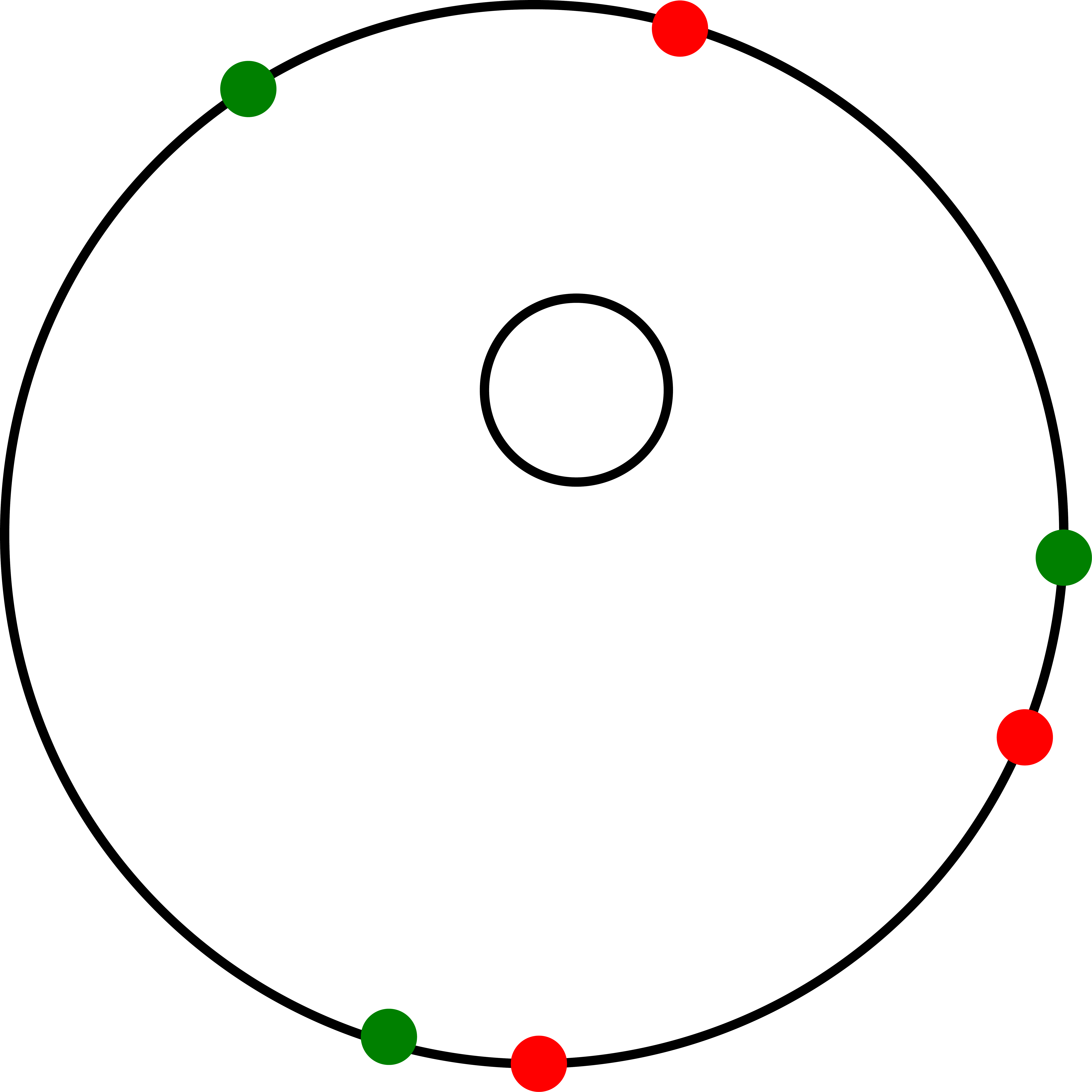
\end{subfigure}
\begin{subfigure}{.16\textwidth}
    \centering
    \fontsize{10pt}{12pt}\selectfont
    \def\svgwidth{\linewidth}
    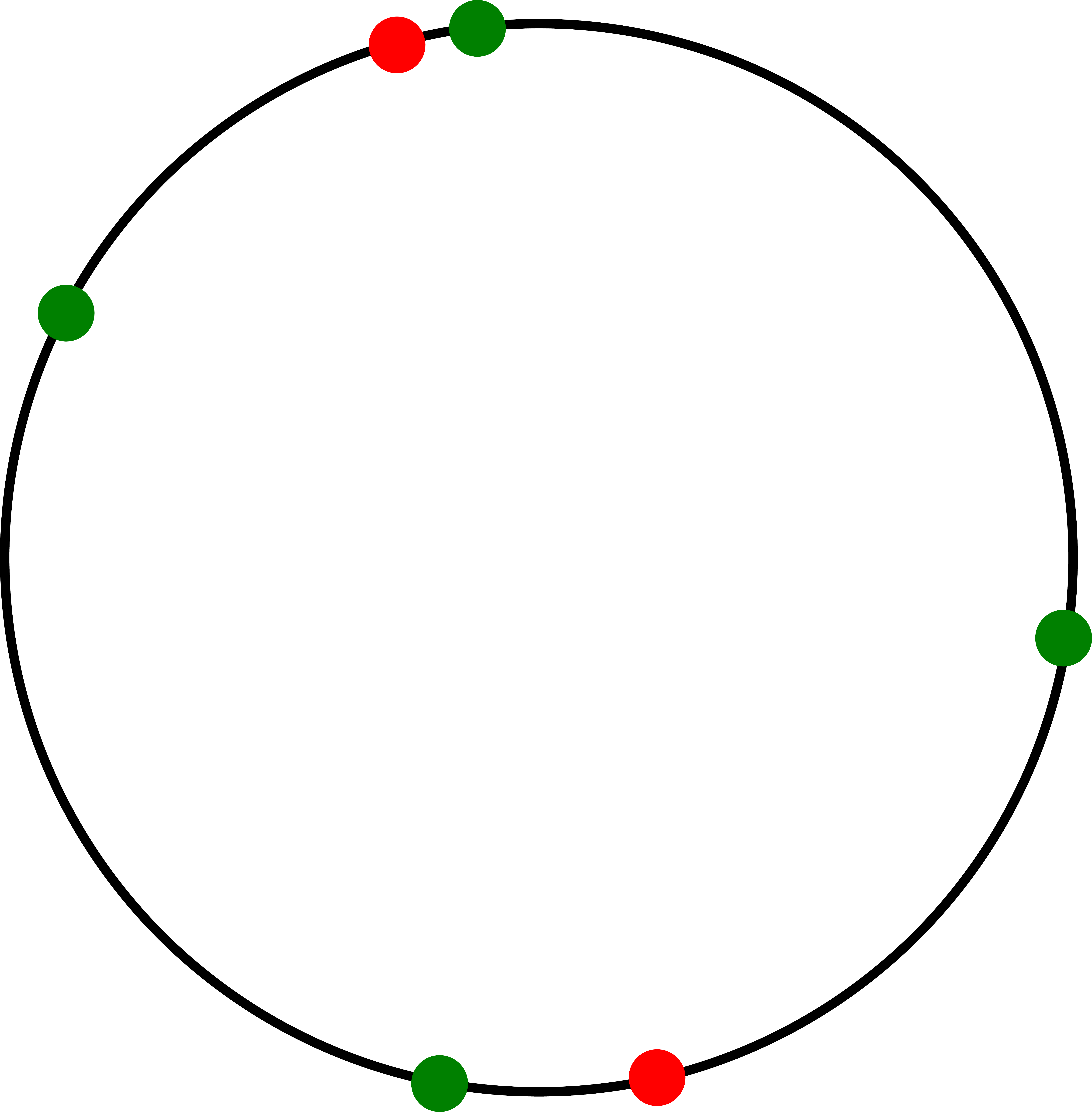
\end{subfigure}
\caption{The Arctic Font Spectrum: Harnack data corresponding to the limit shapes in Figure~\ref{fig:Arctic_font}.}
\label{fig:arctic_spectrum}
\end{figure}
\bibliographystyle{acm}

\bibliography{refs}

\end{document}